\documentclass[12pt]{article}
\usepackage{graphicx}
\usepackage{natbib}
\usepackage{url} 
\usepackage{amsmath,amssymb}
\usepackage{booktabs}
\usepackage{color}
\usepackage{float}
\usepackage{amsthm}
\usepackage{authblk}
\newcommand{\blind}{1}

\newcounter{regime}

\newcounter{mystep}

\newtheorem{theorem}{Theorem}
\newtheorem{lemma}{Lemma}

\newtheorem{prop}{Proposition}
\newtheorem{assumption}{Assumption}
\newtheorem{condition}{Condition}
\newtheorem{remark}{Remark}
\newtheorem{corollary}{Corollary}

\def\E{\mathbb E}
\def\p{\mathbb P}
\def\var{\mathrm{var}}

\usepackage{hyperref}
\hypersetup{
    colorlinks=true,
    urlcolor=blue,
    citecolor=blue
}

\providecommand{\legend}[1]{\par\smallskip\begin{minipage}{0.98\textwidth}
\footnotesize\emph{Notes:} #1\end{minipage}}

\begin{document}
\tolerance=2500
\emergencystretch=2em

\def\spacingset#1{\renewcommand{\baselinestretch}{#1}\small\normalsize}
\spacingset{1}

\if1\blind
{
  \title{\bf Causal inference in two-sided randomization designs: factorial regression, two-way clustering, and covariate adjustment}

 \author[1]{Pengfei Tian}
\author[2]{Jizhou Liu}
\author[3]{Lei Shi}
\author[4]{Peng Ding}

\affil[1]{\small Qiuzhen College, Tsinghua University,
\texttt{\href{mailto:tpf24@mails.tsinghua.edu.cn}{tpf24@mails.tsinghua.edu.cn}}}

\affil[2]{\small HSBC Business School, Peking University,
\texttt{\href{mailto:jizhou.liu@phbs.pku.edu.cn}{jizhou.liu@phbs.pku.edu.cn}}}

\affil[3]{\small Adobe Research,
\texttt{\href{mailto:leis@adobe.com}{leis@adobe.com}}}

\affil[4]{\small Department of Statistics, University of California, Berkeley,
\texttt{\href{mailto:pengdingpku@berkeley.edu}{pengdingpku@berkeley.edu}}}

  \date{}
  \maketitle
}
\fi
	
\if0\blind
{
  \title{\bf Causal inference in two-sided randomization designs: factorial regression, two-way clustering, and covariate adjustment}
  \date{}
  \maketitle
}
\fi

\begin{abstract}
We study randomized experiments involving two interacting populations, such as buyers and sellers in a marketplace. In the two-sided experiments we consider, we randomize the two populations separately and independently. For a pair consisting of one member from each population, the two assignments jointly determine one of four exposure conditions. Under a local interference assumption, we consider a broad class of linear estimands, including total, interaction, and buyer- and seller-side spillover effects. Our first main result establishes that researchers can estimate these effects using ordinary least squares and conduct asymptotically valid design-based inference using the conventional two-way cluster-robust variance estimator, clustered at the buyers' and sellers' levels. Our second main result develops a sharper variance estimator for a single linear estimand that better preserves dependence within the buyer and seller dimensions and is asymptotically less conservative than the two-way clustered estimator and existing alternatives. Our third main result establishes the theory for covariate adjustment and recommends a two-way analysis-of-variance-type covariate representation to ensure efficiency gains.
\end{abstract}

\bigskip


\newpage

\spacingset{1.9}

\section{Introduction}
Randomized experiments increasingly arise in settings where outcomes are indexed by
interactions between two sets of units \citep{johari2022experimental}. Examples include transactions between buyers and
sellers \citep{blake2014marketplace}, matches between riders and drivers \citep{azagirre2024better}, and relationships in financial networks \citep{comola2021treatment}. In such settings, an intervention assigned to one unit can affect outcomes
involving other units through their interactions, creating interaction effects
and spillovers that are not naturally accommodated by standard one-population
experimental designs \citep{holtz2025reducing}. Related work studies
interference generated through market-equilibrium responses
\citep{wager2021experimenting,munro2025treatment}.

Two-sided experiments naturally address this challenge. In a buyer-seller marketplace, the experimenter may randomize buyers and sellers
separately and independently. The joint assignments place each buyer--seller pair into one of four exposure
conditions. This design, also called a simple multiple randomization design
\citep{bajari2023experimental}, enables the study of several causal effects,
including total effects, buyer- and seller-side spillover effects, and
interaction effects. Under suitable restrictions on interference, recent work has developed
design-based estimation and asymptotic inference
\citep{masoero2026multiple}, Fisherian randomization inference
\citep{liu2025randomization}, and covariate adjustment
\citep{sudijono2026regression}.

However, practical challenges remain for statistical inference in two-sided experiments.
Existing design-based variance estimators are algebraically involved and therefore less attractive to empirical researchers. In addition, available conservative variance bounds can be substantially larger than the design-based variance, resulting in overly wide confidence intervals. These issues are especially important because applied researchers often
rely on regression-based workflows: specify a linear regression, report standard errors from
familiar software, and interpret coefficients as treatment-effect estimates. This paper studies
when such regression-based procedures are justified in two-sided experiments and how
design-based inference can be made less conservative.

We first show that the usual difference-in-means estimators for the total effect,
the buyer- and seller-side spillover effects, and the interaction effect can all be represented
through a saturated linear regression. 
We then develop a central limit theorem (CLT) for these estimators using
a new proof strategy based on a H\'ajek  projection \citep{hajek1968asymptotic} that separates the estimation
error into buyer-side and seller-side first-order components and an interaction
remainder. This decomposition clarifies the two sources of randomization uncertainty and
shows that the interaction remainder is asymptotically negligible under mild regularity
conditions. Finally, we show that the conventional two-way cluster-robust variance estimator \citep{cameron2011robust} from the saturated linear regression, with clusters defined by the two sides of the market, is asymptotically conservative for the design-based variance.
Thus, empirical researchers can obtain valid large-sample design-based inference by fitting the saturated regression and using two-way cluster-robust standard errors. This procedure supports joint inference for multiple effects without requiring a correctly specified linear outcome model.

For a single linear estimand, we further develop a sharper contrast-specific design-based variance estimator.
Existing approaches apply conservative bounds after expanding the variance into separate exposure-specific components. Our approach retains more of the dependence structure within the buyer and seller sides before bounding the remaining unobservable terms. This yields a class of conservative variance estimators indexed by tuning parameters. The class contains the two-way clustered variance as a special case and admits a closed-form optimal member. We show that the optimized estimator remains asymptotically conservative and is sharper than the conservative variance targets proposed in the existing literature
\citep{masoero2026multiple,sudijono2026regression}.

We also study regression adjustment with pre-treatment covariates. We consider both Fisher's
additive specification \citep{Fisher1935}, which uses a common covariate slope, and Lin's fully interacted
specification \citep{lin2013agnostic}, which allows the slope to vary across exposure cells. For each specification, we
establish a design-based CLT and show that the usual two-way
cluster-robust covariance matrix from the adjusted regression remains asymptotically
conservative. We further consider adjustments that decompose covariates into buyer-, seller-, and buyer-seller components. For Lin's adjustment, this decomposition improves asymptotic efficiency, while Fisher's adjustment enjoys the same guarantee under an additional condition.

The rest of the paper is organized as follows. Section~\ref{sec:setup} introduces the design, interference restriction, and causal estimands. Sections~\ref{sec:estimation-inference} and \ref{sec:variance-inference} develop point estimation, asymptotic distributions, and joint and scalar inference. Section~\ref{sec:regression-adjustment} studies regression adjustment. Sections~\ref{sec:simulation} and \ref{sec:real-data} present simulations and the empirical illustration.
The Supplementary Materials contain four sections: variance comparisons
(Section~\ref{app:additional-results}), proofs for unadjusted estimation and inference 
(Section~\ref{app:main-proofs}), supporting sampling and regression results 
(Section~\ref{app:technical-identities}), and regression-adjustment results and proofs 
(Section~\ref{app:regression-adjustment-proofs}).
The replication package provides code and additional numerical analyses.

\section{Setup and Notation}\label{sec:setup}

We study a two-sided experiment in which outcomes are indexed by pairs of units. For
concreteness, consider a buyer-seller marketplace with buyers indexed by
$i\in[I]:=\{1,\ldots,I\}$ and sellers indexed by
$j\in[J]:=\{1,\ldots,J\}$. The unit of analysis is a buyer-seller pair
$(i,j)\in[I]\times[J]$. Although we use buyer-seller terminology throughout, the same
notation applies to other two-sided environments, such as riders and drivers in ride-hailing
platforms, viewers and influencers in content platforms, or advertisers and content creators
in digital advertising markets. Let
$W=(W_{ij})_{i\in[I],\,j\in[J]}\in\{0,1\}^{I\times J}$ denote a generic
pair-level assignment matrix, where $W_{ij}=1$ means that pair $(i,j)$ receives the direct
intervention and $W_{ij}=0$ otherwise. For $w\in\{0,1\}^{I\times J}$, let
$Y_{ij}(w)$ denote the potential outcome of pair $(i,j)$ under assignment $w$. The observed
outcome is $Y_{ij}^{\mathrm{obs}}$. We first introduce the randomization
mechanism and then impose a local interference condition that reduces the full schedule of
potential outcomes to four exposure-specific potential outcomes.

\subsection{Randomization mechanism and local interference}\label{subsec:randomization}
In a buyer-seller marketplace, buyers and sellers are randomized separately and independently. Their joint assignments determine one of four pair-level exposure states. We now define the design formally.
Let
\[
W^B=(W_1^B,\ldots,W_I^B)^\top\in\{0,1\}^I,
\qquad
W^S=(W_1^S,\ldots,W_J^S)^\top\in\{0,1\}^J
\]
denote the buyer-side and seller-side assignments, respectively. The numbers of treated buyers and sellers are fixed at \(I_1\) and \(J_1\), so that $\sum_{i=1}^I W_i^B=I_1$ and
$\sum_{j=1}^J W_j^S=J_1$. Let $I_0=I-I_1$ and $J_0=J-J_1$. 
We formalize the randomization mechanism as follows.

\begin{assumption}[Independent two-sided complete randomization]\label{ass:smrd-randomization}
The buyer-side and seller-side assignments are generated by independent complete
randomizations. Specifically, for any $w^B\in\{0,1\}^I$ and $w^S\in\{0,1\}^J$ satisfying
$\sum_{i=1}^I w_i^B=I_1$ and $\sum_{j=1}^J w_j^S=J_1$,
\[
\p(W^B=w^B,\;W^S=w^S)
=
\p(W^B=w^B)\cdot \p(W^S=w^S)
=
\binom{I}{I_1}^{-1}\binom{J}{J_1}^{-1}.
\]
The pair-level assignment matrix is induced by
$W_{ij}=W_i^B W_j^S$ for $(i,j)\in[I]\times[J]$.
\end{assumption}
Denote the corresponding assignment fractions by
\(
e_p^B={I_p}/{I}
\) and 
\(
e_q^S={J_q}/{J}
\) 
for \(p,q\in\{0,1\}.\)
Assumption~\ref{ass:smrd-randomization} is imposed by design and therefore holds by
construction under two-sided experiments. It gives independent randomization across the two
sides while allowing dependence among pair-level assignments that share the same buyer or
seller. In particular, for any fixed pair $(i,j)$,
\[
\p(W_i^B=p,\;W_j^S=q)=e_p^B e_q^S,
\qquad p,q\in\{0,1\}.
\]
However, the pair-level assignments $\{W_{ij}\}_{i,j}$ are not mutually independent. 
Each pair therefore falls into one of four side-level assignment states,
\[
(W_i^B,W_j^S)\in\{(0,0),(1,0),(0,1),(1,1)\},
\]
and receives the direct intervention only in state \((1,1)\). This two-sided dependence
structure drives the inference procedures developed below.

We next restrict how the assignment may affect pair-level potential outcomes using the local interference condition from the two-sided experiments literature
\citep{masoero2026multiple,sudijono2026regression,liu2025randomization}.
Under local interference, the potential outcome of a pair may depend on its own treatment
status and on the fractions of treated pairs in its buyer row and seller column. Under our
two-sided complete randomization, these fractions are fixed by the side-level assignments, so local interference reduces to four exposure-specific potential outcomes indexed by
$(W_i^B,W_j^S)$.

\begin{assumption}[Local interference]\label{ass:local-interference}
For each pair $(i,j)$, the potential outcome depends on the assignment only through
$(w_i^B,w_j^S)$. Thus,
\(
Y_{ij}(w)=Y_{ij}(w_i^B,w_j^S)
\)
and
\(
Y_{ij}^{\mathrm{obs}}
=
Y_{ij}(W_i^B,W_j^S).
\)
\end{assumption}

\subsection{Estimands}\label{subsec:estimands}
Under Assumption~\ref{ass:local-interference}, each pair has four potential outcomes indexed
by the buyer- and seller-side assignments. For each $(p,q)\in\{0,1\}^2$, define
\begin{equation}\label{eq:mean-potential-outcome}
\bar Y(p,q)
=
\frac{1}{IJ}\sum_{i=1}^I\sum_{j=1}^J Y_{ij}(p,q).
\end{equation}
We also write
\(
\bar Y_{i\cdot}(p,q)
=
{J}^{-1}\sum_{j=1}^J Y_{ij}(p,q),
\) and \(
\bar Y_{\cdot j}(p,q)
=
{I}^{-1}\sum_{i=1}^I Y_{ij}(p,q)
\)
for the corresponding buyer- and seller-specific averages.
We collect the four exposure-specific means in the vector
\begin{equation}\label{eq:Ybar-vector}
\bar{\mathbf Y}
=
\bigl(
\bar Y(1,1),\,
\bar Y(1,0),\,
\bar Y(0,1),\,
\bar Y(0,0)
\bigr)^\top .
\end{equation}
The vector $\bar{\mathbf Y}$ collects the four finite-population mean potential outcomes.
The estimands considered below are linear transformations of this vector $F\bar {\mathbf Y}$. For instance, $F$ can be
\begin{equation}\label{eq:F-main}
F
=
\begin{pmatrix}
1 & 0 & 0 & -1 \\
0 & 1 & 0 & -1 \\
0 & 0 & 1 & -1 \\
1 & -1 & -1 & 1
\end{pmatrix},
\end{equation}
where the four rows correspond respectively to
\begin{align}
\tau_{\mathrm{tot}}
&:= \bar Y(1,1)-\bar Y(0,0), \notag\\
\tau_{B,\mathrm{spill}}
&:= \bar Y(1,0)-\bar Y(0,0), \label{estimand:spill}\\
\tau_{S,\mathrm{spill}}
&:= \bar Y(0,1)-\bar Y(0,0), \notag\\
\tau_{\mathrm{int}}
&:= \bar Y(1,1)-\bar Y(1,0)-\bar Y(0,1)+\bar Y(0,0).
\label{eq:special tau}
\end{align}
We refer to the four contrasts as the total effect, buyer-side spillover effect,
seller-side spillover effect, and interaction effect, respectively.
These four contrasts are also studied by
\citet{masoero2026multiple} and \citet{sudijono2026regression}.

The total effect compares joint treatment with joint control. The buyer-side spillover effect
compares $(1,0)$ with $(0,0)$, capturing treatment exposure through the buyer side for an
untreated pair; the seller-side spillover effect is analogous. The interaction effect satisfies
\(
\tau_{\mathrm{int}}
=
\tau_{\mathrm{tot}}
-
\tau_{B,\mathrm{spill}}
-
\tau_{S,\mathrm{spill}}
\)
and measures the deviation of the joint-treatment effect from the sum of the two one-sided
exposure effects, equivalently the usual interaction contrast in a $2\times2$ factorial structure.
\section{Point Estimator and its Asymptotic Distribution}\label{sec:estimation-inference}

This section introduces the moment estimator of the exposure-specific mean vector
$\bar{\mathbf Y}$ and establishes its saturated-regression representation. We then derive
a first-order decomposition and the joint asymptotic distribution of the estimator.

\subsection{Plug-in and saturated-regression estimators}
\label{subsec:plugin-reg-vector}

For each $(p,q)\in\{0,1\}^2$, define the observed cell mean
\begin{equation*}
\widehat Y(p,q)
=
\frac{1}{I_pJ_q}
\sum_{i=1}^I\sum_{j=1}^J
Y_{ij}^{\mathrm{obs}}
\mathbf 1\{W_i^B=p,\;W_j^S=q\}.
\end{equation*}
For each $(p,q)$, $\widehat Y(p,q)$ is a subarray mean under independent row and column sampling, as in the ``pigeonhole'' framework of \citet{cornfield1956average}.
Collect these four observed means in the vector
\begin{equation*}
\widehat{\mathbf Y}
=
\bigl(
\widehat Y(1,1),\,
\widehat Y(1,0),\,
\widehat Y(0,1),\,
\widehat Y(0,0)
\bigr)^\top .
\end{equation*}
The vector $\widehat{\mathbf Y}$ is the natural plug-in estimator of the finite-population
causal object $\bar{\mathbf Y}$.

We next give an equivalent regression representation. Let
$
G_{pq,ij}
=
\mathbf 1\{W_i^B=p,\;W_j^S=q\}$ for $ p,q\in\{0,1\}$, and define the four-dimensional regressor
\[
z_{ij}
=
\bigl(
G_{11,ij},\,
G_{10,ij},\,
G_{01,ij},\,
G_{00,ij}
\bigr)^\top .
\]
Let
$\widehat{\boldsymbol\mu}
=
(\widehat\mu_{11},\widehat\mu_{10},\widehat\mu_{01},\widehat\mu_{00})^\top$
be the ordinary least squares (OLS) coefficient vector from the saturated cell-indicator regression:
\begin{equation}\label{eq:reg-Gij-sec31}
\widehat{\boldsymbol\mu}
\in
\arg\min_{\mu\in\mathbb R^4}
\sum_{i=1}^I\sum_{j=1}^J
\left\{
Y_{ij}^{\mathrm{obs}}-z_{ij}^\top\mu
\right\}^2.
\end{equation}
The associated residual is
$
\widehat u_{ij}
=
Y_{ij}^{\mathrm{obs}}-z_{ij}^\top\widehat{\boldsymbol\mu}=Y_{ij}^{\mathrm{obs}}-\widehat Y(W_i^B,W_j^S).$
By standard OLS properties, the cellwise first-order conditions give
$\widehat{\boldsymbol\mu}=\widehat{\mathbf Y}$.
\begin{remark}[Factor-based regression]
\label{rem:factor-based-point}
The two side-level assignments induce a $2^2$ factorial structure. 
In the terminology of \citet{zhao2022factorial}, the saturated
cell-indicator regression above is a treatment-based regression, whereas
the equivalent factor-based regression is
\(Y_{ij}^{\mathrm{obs}}
\sim
1+W_i^B+W_j^S+W_i^BW_j^S.
\)
Its four coefficients correspond to
$\widehat Y(0,0)$, $\widehat\tau_{B,\mathrm{spill}}$,
$\widehat\tau_{S,\mathrm{spill}}$, and
$\widehat\tau_{\mathrm{int}}$, respectively, while
$\widehat\tau_{\mathrm{tot}}$ is the sum of these three coefficients. \citet{sudijono2026regression} use this factor-based
representation to express the interaction-effect estimator as the interaction
coefficient and motivate subsequent regression adjustment.
Up to an invertible linear transformation, these two regressions are
equivalent for both point estimation and covariance estimation.
We therefore focus on the OLS regression in \eqref{eq:reg-Gij-sec31} throughout.
\end{remark}

\begin{remark}[Joint-treatment-only special case]
\label{rem:joint-treatment-only}
A useful special case is one in which the potential outcome depends on the
side-level assignments only through their product
\(D_{ij}=W_i^BW_j^S\). Equivalently,
\(Y_{ij}(1,0)=Y_{ij}(0,1)=Y_{ij}(0,0)\) for all $(i,j)$.

A natural reduced regression is
\(Y_{ij}^{\mathrm{obs}}\sim 1+D_{ij}\).
Its slope coefficient can be written as
\[
\widehat\tau_D
=
F_D\widehat{\mathbf Y},
\qquad
F_D=
(
1,-\omega_{10},-\omega_{01},-\omega_{00}
),
\qquad
\omega_{pq}
=
\frac{e_p^Be_q^S}{1-e_1^Be_1^S},
\quad (p,q)\neq(1,1).
\]
Under the restriction above,
\(F_D\bar{\mathbf Y}
=\bar Y(1,1)-\bar Y(0,0)
\).
Hence the vector asymptotic theory developed below applies directly to this
reduced-regression estimand.
\end{remark}
\subsection{First-order covariance and vector CLT}\label{subsec:twoway-vector}
The key technical device is a H\'ajek projection of $\widehat{\mathbf Y}$. Viewing
$\widehat{\mathbf Y}-\bar{\mathbf Y}$ as a function of
$(W^B,W^S)$, we have
\begin{align}\label{eq:HoeffdingEquation}
    \widehat{\mathbf Y}-\bar{\mathbf Y} = \E\left(
\widehat{\mathbf Y}-\bar{\mathbf Y}
\mid W^B
\right) + \E\left(
\widehat{\mathbf Y}-\bar{\mathbf Y}
\mid W^S
\right) + R_{BS},
\end{align}
where $R_{BS}$ is a buyer-seller interaction remainder and negligible under the regularity conditions below.
To further express the first two terms on the right-hand side of \eqref{eq:HoeffdingEquation}, for each buyer $i$ and seller $j$,  define the four-dimensional buyer-side score vector
\[
A_i
=
\left(
\frac{\bar Y_{i\cdot}(1,1)}{e_1^B},
\frac{\bar Y_{i\cdot}(1,0)}{e_1^B},
-\frac{\bar Y_{i\cdot}(0,1)}{e_0^B},
-\frac{\bar Y_{i\cdot}(0,0)}{e_0^B}
\right)^\top ,
\]
and the four-dimensional seller-side score vector
\[
B_j
=
\left(
\frac{\bar Y_{\cdot j}(1,1)}{e_1^S},
-\frac{\bar Y_{\cdot j}(1,0)}{e_0^S},
\frac{\bar Y_{\cdot j}(0,1)}{e_1^S},
-\frac{\bar Y_{\cdot j}(0,0)}{e_0^S}
\right)^\top .
\]
Let
$
\bar A={I}^{-1}\sum_{i=1}^I A_i,
$ and $
\bar B={J}^{-1}\sum_{j=1}^J B_j.
$
The two first-order projections are
\begin{equation*}
\E\left(
\widehat{\mathbf Y}-\bar{\mathbf Y}
\mid W^B
\right) = \frac1I\sum_{i=1}^I (W_i^B-e_1^B)A_i, \qquad \E\left(
\widehat{\mathbf Y}-\bar{\mathbf Y}
\mid W^S
\right) =\frac1J\sum_{j=1}^J (W_j^S-e_1^S)B_j.\end{equation*} 
The buyer-side complete randomization determines the first projection, whereas the
seller-side complete randomization determines the second. The independence of the two
side-level randomizations makes these first-order terms independent. We can therefore view
their sum as a four-dimensional stratified permutation statistic with two strata,
corresponding to the buyer and seller sides. For a vector finite population $C_1,\ldots,C_N$, write
$\mathcal S_N(C)=(N-1)^{-1}\sum_{k=1}^N
(C_k-\bar C)(C_k-\bar C)^\top$. 
Let $\Sigma$ denote the covariance matrix of the sum of the two first-order projections on the right-hand side of \eqref{eq:HoeffdingEquation}.
Using the covariance representation for stratified permutation statistics in
\citet{tian2025stratified}, we obtain
\begin{equation}\label{eq:Sigma-vector}
\Sigma
=
\frac{e_1^B e_0^B}{I}
\cdot \mathcal S_I(A)
+
\frac{e_1^S e_0^S}{J}
\cdot\mathcal S_J(B).
\end{equation}
The first component of $\Sigma$ captures buyer-side randomization uncertainty, while the second captures seller-side randomization uncertainty.

\begin{condition}\label{cond:vector-clt}
As $I,J\to\infty$, assume:
\begin{enumerate}
    \item there exists a constant $\underline e>0$ such that $\min\{e_0^B,e_1^B,e_0^S,e_1^S\}\ge \underline e;$
    \item $\Sigma$ is nondegenerate in the sense that
    there exists a constant $c_\Sigma>0$ such that $\lambda_{\min}(\Sigma)
    \ge
    c_\Sigma\{I^{-1}+J^{-1}\}$ for all sufficiently large $I,J$;
    \item there exist constants $\delta>0$ and $C_Y<\infty$ such that,
    for all $(p,q)\in\{0,1\}^2$,
    \[
    \frac{1}{IJ}
    \sum_{i=1}^I\sum_{j=1}^J
    \left|Y_{ij}(p,q)-\bar Y(p,q)\right|^{2+\delta}
    \le C_Y.
    \]
\end{enumerate}
\end{condition}
Condition~\ref{cond:vector-clt}(i) requires nonvanishing treatment and
control fractions on both sides.
Condition~\ref{cond:vector-clt}(ii) prevents the leading covariance matrix
from degenerating faster than its natural scale \(I^{-1}+J^{-1}\). It requires
sufficient heterogeneity in the buyer- and seller-side first-order projection
scores and rules out settings in which these projections become nearly
constant across both sides.
Condition~\ref{cond:vector-clt}(iii) imposes a uniform
\((2+\delta)\)-moment bound on the potential outcomes.
Taken together, Conditions~\ref{cond:vector-clt}(i)--(iii), along with
Jensen's inequality, imply bounded \((2+\delta)\)-moments for $A_i-\bar A$ and
$B_j-\bar B$ and verify the Lyapunov condition for the stratified permutation CLT of
\cite{tuvaandorj2024combinatorial}. They also imply
\(\var(R_{BS})=O((IJ)^{-1})\), so the remainder is asymptotically negligible
relative to the first-order terms.
\begin{theorem}[Vector CLT for exposure-specific means]
\label{thm:vector-clt}
Under Assumptions~\ref{ass:smrd-randomization} and~\ref{ass:local-interference}, and
Condition~\ref{cond:vector-clt},
\[
\Sigma^{-1/2}
\left(
\widehat{\mathbf Y}-\bar{\mathbf Y}
\right)
\Rightarrow
N(0,\mathrm I_4),
\]
where $\Sigma$ is defined in \eqref{eq:Sigma-vector} and
$\mathrm I_4$ denotes the $4\times4$ identity matrix.
\end{theorem}

Theorem~\ref{thm:vector-clt} gives the joint design-based distribution of the four exposure-specific mean estimators. It implies the CLT for $F\widehat{\mathbf Y}$ immediately.

\section{Variance Estimators and Inference}\label{sec:variance-inference}
\subsection{Two-way cluster-robust covariance}\label{subsec:joint-inference}
The two-sided randomization creates dependence along both dimensions of the
outcome array. The buyer-side assignment $W_i^B$ enters the exposure state of
every pair in row $i$, while the seller-side assignment $W_j^S$ enters the
exposure state of every pair in column $j$. This naturally suggests clustering
by both buyers and sellers.

Conventional two-way cluster-robust inference is typically developed under
stochastic dependence
\citep{cameron2011robust,yap2025asymptotic}.
A design-based perspective on regression uncertainty is developed by
\citet{abadie2020sampling}, while \citet{abadie2023should} emphasize
the sampling- and assignment-design foundations of clustered inference.
\citet{xu2024clustering} study multiway clustering in a design-based framework
with clustered sampling and assignment, showing that the conventional multiway
cluster-robust variance estimator need not be conservative in that setting.
Here we consider two-sided experiments under two independent complete
randomizations. Treating the potential outcomes as fixed, we show below that
the conventional two-way cluster-robust covariance estimator is asymptotically
conservative under this randomization structure.

Let $Z$ be the $IJ\times 4$ design matrix whose row corresponding to pair $(i,j)$ is
$z_{ij}^\top$. Let $Z_i$ denote the submatrix of $Z$ containing all rows associated with buyer
$i$, and let $\widehat u_i$ denote the corresponding vector of OLS residuals. Similarly, let
$Z_{\cdot j}$ and $\widehat u_{\cdot j}$ denote the seller-side analogues. Define $\Gamma_B
=
\sum_{i=1}^I
Z_i^\top \widehat u_i\widehat u_i^\top Z_i$, $
\Gamma_S
=
\sum_{j=1}^J
Z_{\cdot j}^\top \widehat u_{\cdot j}\widehat u_{\cdot j}^\top Z_{\cdot j}$,
and
$
\Gamma_{BS}
=
\sum_{i=1}^I\sum_{j=1}^J
z_{ij}z_{ij}^\top \widehat u_{ij}^2.$
The two-way clustered covariance matrix \citep{cameron2011robust} for $\widehat{\mathbf Y}$ is
\begin{equation}\label{eq:twoway-cov-vector}
\widehat\Omega_{\mathrm{2w}}
=
(Z^\top Z)^{-1}
\bigl(
\Gamma_B+\Gamma_S-\Gamma_{BS}
\bigr)
(Z^\top Z)^{-1}.
\end{equation}

We now state the asymptotic conservativeness for the two-way
cluster-robust covariance estimator. Split $A_i$ and $B_j$
by assignment state. Let
$A_i^{(1)}=(\bar Y_{i\cdot}(1,1)/e_1^B,
\bar Y_{i\cdot}(1,0)/e_1^B,0,0)^\top$ and
$A_i^{(0)}=(0,0,\bar Y_{i\cdot}(0,1)/e_0^B,
\bar Y_{i\cdot}(0,0)/e_0^B)^\top$, and define
$B_j^{(1)}$ and $B_j^{(0)}$ analogously. Then
$A_i=A_i^{(1)}-A_i^{(0)}$ and $B_j=B_j^{(1)}-B_j^{(0)}$.
\begin{prop}
\label{prop:twoway-vector-conservative}
Under Assumptions~\ref{ass:smrd-randomization} and~\ref{ass:local-interference}, suppose
Condition~\ref{cond:vector-clt} holds. Suppose further that there exists a constant
$C_4<\infty$ such that, for all $(p,q)\in\{0,1\}^2$,
\[
\frac{1}{IJ}
\sum_{i=1}^I\sum_{j=1}^J
\{Y_{ij}(p,q)-\bar Y(p,q)\}^4
\le C_4.
\]
Then
\begin{equation}\label{eq:Omega2w-conservative}
\widehat\Omega_{\mathrm{2w}}-\Sigma
=\frac 1 I \mathcal S_I\left(e_1^B A^{(1)}+e_0^BA^{(0)}\right)+\frac 1 J \mathcal S_J\left(e_1^S B^{(1)}+e_0^SB^{(0)}\right)
+o_{\mathbb P} (I^{-1}+J^{-1}).
\end{equation}
\end{prop}

Proposition~\ref{prop:twoway-vector-conservative} gives an explicit
first-order decomposition of the conservativeness of
$\widehat\Omega_{\mathrm{2w}}$. The first two terms
on the right-hand side of \eqref{eq:Omega2w-conservative} are positive
semidefinite. 
\begin{remark}[Finite-sample positive semidefiniteness]
\label{rem:twoway-psd}
As noted by \citet{cameron2011robust}, the conventional
two-way cluster-robust covariance need not be positive semidefinite in finite
samples. Under the conditions of Proposition~\ref{prop:twoway-vector-conservative}, however, $\Gamma_{BS}$ is first-order negligible, so omitting it yields a
positive-semidefinite, first-order-equivalent covariance matrix; see
\eqref{eq:intersection-negligible} in Section~\ref{app:proof-two-way}.
\end{remark}

\begin{remark}[Multiple treatment levels]
\label{rem:multiple-levels}

The regression formulation extends naturally to fixed multiple levels of $W^B$ and $W^S$ via a
saturated regression on all buyer-seller treatment cells, equivalently on
buyer- and seller-level dummies and their interactions, with two-way clustered
standard errors. Under factor coding, centering determines how effects are
averaged over the other factor \citep{zhao2022factorial}.
Strip-plot design is a special case of this multi-level factorial structure:
within each block, each treatment level is assigned to a single row or column, so that
$I_p=J_q=1$ in our notation. Their variance estimation therefore relies on replication across
blocks rather than within treatment cells \citep{alqallaf2019causal}.
\end{remark}

\subsection{Improving the variance estimator for a scalar parameter}
\label{subsec:scalar-inference}
We next improve variance estimation for a fixed scalar linear contrast.

For a fixed contrast vector $\beta$, let
$\tau_\beta=\beta^\top\bar{\mathbf Y}$ and
$\widehat\tau_\beta=\beta^\top\widehat{\mathbf Y}$.
The vector CLT in Theorem~\ref{thm:vector-clt} immediately gives a CLT for $\widehat\tau_\beta$ with the asymptotic variance defined as $\sigma_\beta^2
=
\beta^\top\Sigma\beta.$

To express $\sigma_\beta^2$ in a form useful for variance estimation, define the
state-specific buyer-side projections
\(
a_i(p)
=\beta^\top A_i^{(p)},
\) for \(
p\in\{0,1\}\),
and the state-specific seller-side projections
\(
b_j(q)
=\beta ^\top B_j^{(q)},
\) for \(
q\in\{0,1\}.\)
These quantities satisfy
\[
\beta^\top A_i
=
a_i(1)-a_i(0),
\qquad
\beta^\top B_j
=
b_j(1)-b_j(0).
\]

Define 
\(
S_{a,p}^2
=
(I-1)^{-1}
\sum_{i=1}^I
\{a_i(p)-\bar a(p)\}^2
\) and \(S_{a,\Delta}^2
=(I-1)^{-1}\sum_{i=1}^I
\{\beta^\top (A_i-\bar A) \}^2\).
Similarly, define $S^2_{b,q}$ and $S^2_{b,\Delta}$.
The asymptotic variance of $\widehat\tau_\beta$ is
\begin{equation}\label{eq:scalar-leading-variance}
\sigma_\beta^2
=
\frac{e_1^Be_0^B}{I}S_{a,\Delta}^2
+
\frac{e_1^Se_0^S}{J}S_{b,\Delta}^2.
\end{equation}
By the Cauchy--Schwarz inequality, we have
\(
S_{a,\Delta}^2
\le
(S_{a,1}+S_{a,0})^2
\) and \(
S_{b,\Delta}^2
\le
(S_{b,1}+S_{b,0})^2.
\)
Hence,
\begin{equation}\label{eq:Vopt-population}
\sigma_\beta^2
\le
\frac{e_1^Be_0^B}{I}
(S_{a,1}+S_{a,0})^2
+
\frac{e_1^Se_0^S}{J}
(S_{b,1}+S_{b,0})^2.
\end{equation}
We denote the right-hand side of \eqref{eq:Vopt-population} by
$V_{\mathrm{opt},\beta}^{\circ}$ and estimate it using sample analogues of
the state-specific standard deviations.
\begin{remark}[Exactness of variance targets for a scalar parameter]
\label{rem:scalar-variance-exactness}
The requirement for $V_{\mathrm{opt},\beta}^{\circ}=\sigma_\beta^2$ is weaker than $\beta^\top\Omega_{\mathrm{2w}}^\circ\beta=\sigma_\beta^2$. On the buyer side, exactness of the optimized target requires that there exist constants $c_1^B,c_0^B\ge0$ with $c_1^B+c_0^B=1$ such that for all $i$,
\[
c_1^B \{a_i(1)-\bar a(1)\}+c_0^B\{a_i(0)-\bar a(0)\}=0.
\]
Exactness of the two-way target is the special case
$c_1^B=e_1^B$ and $c_0^B=e_0^B$.
The seller-side conditions are analogous. 
Corollary~\ref{cor:scalar-variance-exactness} in Section~\ref{app:optimized-variance-family}
formalizes this comparison.
\end{remark}
This Cauchy--Schwarz bound on the cross-state covariance term parallels the improved conservative variance bound
for completely randomized experiments in \citet[Problem~4.5]{ding2024first};
see also \citet{sengupta2026design}. 
For comparison with other conservative variance estimators,
Section~\ref{app:optimized-variance-family} embeds
$V_{\mathrm{opt},\beta}^{\circ}$ in a two-parameter family of conservative targets.
The family contains the first-order target of the two-way cluster-robust
variance estimator and, for the buyer-spillover contrast, the target associated
with \citet{liu2025randomization} as special cases.

We next construct the sample analogue of $V_{\mathrm{opt},\beta}^{\circ}$. For $p,q\in\{0,1\}$, buyers $i$ with $W_i^B=p$ and sellers $j$ with $W_j^S=q$, define $\widehat{\bar Y}_{i\cdot}(p,q) = J_q^{-1}\sum_{j=1}^J Y_{ij}^{\mathrm{obs}}\mathbf 1\{W_j^S=q\}$ and $\widehat{\bar Y}_{\cdot j}(p,q) = I_p^{-1}\sum_{i=1}^I Y_{ij}^{\mathrm{obs}}\mathbf 1\{W_i^B=p\}$. For each buyer $i$ with $W_i^B=p$, define $\widehat a_i(p) = (e_p^B)^{-1}\sum_{q=0}^1 \beta_{pq}\widehat{\bar Y}_{i\cdot}(p,q)$. Let $\widehat{\bar a}(p) = I_p^{-1}\sum_{i:W_i^B=p}\widehat a_i(p)$ and $s_{a,p}^2 = (I_p-1)^{-1}\sum_{i:W_i^B=p} \{\widehat a_i(p)-\widehat{\bar a}(p)\}^2$ denote its sample mean and variance. Define the seller-side quantities $\widehat b_j(q)$, $\widehat{\bar b}(q)$, and $s_{b,q}^2$ analogously. We call the resulting estimator optimized because its population target minimizes a two-parameter family of conservative variance targets; see Section~\ref{app:optimized-variance-family}. Specifically, define
\begin{equation}\label{eq:Vopt-sample} \widehat V_{\mathrm{opt},\beta} = \frac{e_1^Be_0^B}{I} \bigl(s_{a,1}+s_{a,0}\bigr)^2 + \frac{e_1^Se_0^S}{J} \bigl(s_{b,1}+s_{b,0}\bigr)^2. \end{equation}

Let $\Omega_{\mathrm{2w}}^\circ$ denote the deterministic
first-order target defined in
\eqref{eq:Omega2w-pop-target} in Section~\ref{app:proof-two-way}.
\begin{prop}[Validity of the optimized variance estimator]
\label{prop:optimized-scalar-variance}
Under the conditions of
Proposition~\ref{prop:twoway-vector-conservative},
\(
\widehat V_{\mathrm{opt},\beta}
=
V_{\mathrm{opt},\beta}^{\circ}
+
o_{\mathbb P}(I^{-1}+J^{-1}),
\)
and
\(
\sigma_\beta^2\le V_{\mathrm{opt},\beta}^{\circ}
\le \beta^\top \Omega_{\textup{2w}}^\circ \beta.
\)
\end{prop}
Combining Proposition~\ref{prop:optimized-scalar-variance} with the CLT for $\widehat\tau_\beta$ yields
asymptotically conservative Wald confidence intervals based on
$\widehat V_{\mathrm{opt},\beta}$.

Our optimized variance improves on the existing variance estimators from the literature. Define the probability limit of variance estimators in \citet{masoero2026multiple} and \citet{sudijono2026regression} as
\[
V_{\mathrm{Mas},\beta}^{\circ}
=
\sum_{p,q}
\beta_{pq}^2
\var\{\widehat Y(p,q)\}
+
\frac{1}{2}
\sum_{(p,q)\neq(p',q')}
|\beta_{pq}\beta_{p'q'}|
\left[
\var\{\widehat Y(p,q)\}
+
\var\{\widehat Y(p',q')\}
\right],
\]
and
\[
V_{\mathrm{Sud},\beta}^{\circ}
=
\left[
\sum_{p,q}
|\beta_{pq}|
\sqrt{\var\{\widehat Y(p,q)\}}
\right]^2.
\]

\begin{prop}
\label{prop:scalar-variance-comparison}
For any fixed contrast vector $\beta$, we have
\(
V_{\mathrm{opt},\beta}^{\circ}
\le
V_{\mathrm{Sud},\beta}^{\circ}
\le
V_{\mathrm{Mas},\beta}^{\circ}.
\)
\end{prop}

Proposition~\ref{prop:scalar-variance-comparison} shows that
$V_{\mathrm{opt},\beta}^{\circ}$ is sharper because it preserves jointly observed
within-state variation and bounds only the unobserved cross-state covariance.

To make the source of the improvement clear, Table~\ref{tab:bspill-var-targets}
in Section~\ref{app:buyer-spillover-illustration} compares the variance
targets using the buyer-spillover effect as an example. The variance of the
buyer-spillover estimator contains buyer-side and seller-side components.
The seller-side covariance is consistently estimable, whereas the buyer-side
cross-state covariance is not identified. Our estimator preserves the former
and bounds only the latter. By contrast, the Masoero- and Sudijono-type
constructions separate the exposure-specific components on both sides before
bounding their covariance and therefore discard estimable seller-side dependence.
The buyer-spillover contrast also illustrates
Remark~\ref{rem:scalar-variance-exactness}. Since $b_j(1)=0$ for every
seller $j$, the optimized target is automatically exact on the seller side,
whereas the two-way target is exact on that side only if
$b_j(0)=\bar b(0)$ for all $j$.

The optimized estimator has a direct regression implementation. After fitting the saturated
regression, one separates the usual buyer- and seller-cluster score contributions by assignment
state and changes only their final aggregation. Section~\ref{app:cluster-score-implementation} in the Supplementary Materials
gives the exact score identities, the aggregation rule, and its first-order equivalence to
\(\widehat V_{\mathrm{opt},\beta}\).

\section{Regression Adjustment using Covariates}\label{sec:regression-adjustment}
A regression-based view makes covariate adjustment natural. Once pre-treatment
covariates are available, we can add them to the exposure-dummy regression.
Regression adjustment has been extensively studied for randomized experiments
and online experimentation
\citep{freedman2008regression,lin2013agnostic,deng2013improving}.
In completely randomized experiments, two canonical specifications are
Fisher's additive regression \citep{Fisher1935}, and Lin's fully
interacted regression \citep{lin2013agnostic}. Fisher's regression uses a common covariate slope across treatment arms. Lin's regression allows the slope to vary by treatment arm. In two-sided experiments, we analogously add covariates to the saturated exposure-dummy regression, using either a common slope or cell-specific slopes across the four exposure cells.

Let $X_{ij}\in\mathbb R^d$ be fixed pre-treatment covariates for pair $(i,j)$, where $d$ is fixed. Define the finite-population covariate mean $\bar X = (IJ)^{-1}\sum_{i=1}^I\sum_{j=1}^J X_{ij}.$ 
We define Fisher's regression in two-sided experiments as \begin{equation}\label{eq:reg-fisher-adjustment} (\widehat{\boldsymbol\mu}_{\mathrm F},\widehat\theta_{\mathrm F}) \in \arg\min_{\mu\in\mathbb R^4,\theta\in\mathbb R^d} \sum_{i=1}^I\sum_{j=1}^J \left[ Y_{ij}^{\mathrm{obs}} - \sum_{p=0}^1\sum_{q=0}^1 G_{pq,ij}\mu_{pq} - (X_{ij}-\bar X)^\top\theta \right]^2,
\end{equation}
where $\widehat{\boldsymbol\mu}_{\mathrm F} = (\widehat\mu_{\mathrm F,11}, \widehat\mu_{\mathrm F,10}, \widehat\mu_{\mathrm F,01}, \widehat\mu_{\mathrm F,00})^\top.$ Here the covariate adjustment uses a single slope $\widehat\theta_{\mathrm F}$ shared across the four exposure cells. Define Lin's regression in two-sided experiments as
\begin{equation}\label{eq:reg-lin-adjustment}
\begin{aligned}
 (\widehat{\boldsymbol\mu}_{\mathrm L},\{\widehat\theta_{\mathrm L,pq}\}_{p,q}) \in \arg\min_{\mu\in\mathbb R^4,\{\theta_{pq}\in\mathbb R^d\}_{p,q}} \sum_{i=1}^I\sum_{j=1}^J\! \left[ Y_{ij}^{\mathrm{obs}} - \sum_{p=0}^1\sum_{q=0}^1 G_{pq,ij} \left\{ \mu_{pq} + (X_{ij}-\bar X)^\top\theta_{pq} \right\} \right]^2,
\end{aligned}
\end{equation}
where $\widehat{\boldsymbol\mu}_{\mathrm L} = (\widehat\mu_{\mathrm L,11}, \widehat\mu_{\mathrm L,10}, \widehat\mu_{\mathrm L,01}, \widehat\mu_{\mathrm L,00})^\top.$
Equivalently, Lin's regression fits a separate covariate slope within each exposure cell. In both regressions, the four exposure-dummy coefficients form a regression-adjusted estimator of $\bar{\mathbf Y}$. We use the two-way clustered covariance matrix from the same regression for inference.

\subsection{Fisher's regression in two-sided experiments}
\label{subsec:fisher-adjustment}

For each $(p,q)\in\{0,1\}^2$, define the observed covariate mean in exposure cell $(p,q)$ as
\(
\widehat X(p,q)
=
{(I_pJ_q)}^{-1}
\sum_{i=1}^I\sum_{j=1}^J
X_{ij}G_{pq,ij}.
\)
The first-order conditions for the exposure-dummy coefficients in
\eqref{eq:reg-fisher-adjustment} imply
\begin{equation}\label{eq:mu-fisher-cell-representation}
\widehat\mu_{\mathrm F,pq}
=
\widehat Y(p,q)
-
\{\widehat X(p,q)-\bar X\}^\top\widehat\theta_{\mathrm F},
\qquad (p,q)\in\{0,1\}^2.
\end{equation}
To describe the large-sample behavior of this estimator, define
$\dot X_{ij}=X_{ij}-\bar X$ and
$\dot Y_{ij}(p,q)=Y_{ij}(p,q)-\bar Y(p,q)$.
By the Frisch--Waugh--Lovell theorem \citep{frisch1933partial,lovell1963seasonal}, the population Fisher slope
$\theta_{\mathrm F}^\star$ equals
\begin{equation}\label{eq:theta-fisher-star}
\theta_{\mathrm F}^\star
\in
\arg\min_{\theta\in\mathbb R^d}
\sum_{p=0}^1\sum_{q=0}^1
e_p^Be_q^S
\frac{1}{IJ}
\sum_{i=1}^I\sum_{j=1}^J
\left[
\dot Y_{ij}(p,q)
-
\dot X_{ij}^\top\theta
\right]^2.
\end{equation}
Equivalently, let
\begin{align}\label{eq:QX_def}
    Q_{X}^{\circ}
=
\frac{1}{IJ}
\sum_{i=1}^I\sum_{j=1}^J
\dot X_{ij}\dot X_{ij}^\top,\qquad q_{\mathrm F}^{\circ}
=
\sum_{p=0}^1\sum_{q=0}^1
e_p^Be_q^S
\left\{
\frac{1}{IJ}
\sum_{i=1}^I\sum_{j=1}^J
\dot X_{ij}\dot Y_{ij}(p,q)
\right\}.
\end{align}
When $Q_{X}^{\circ}$ is nonsingular, we have
\(
\theta_{\mathrm F}^\star
=
(Q_{X}^{\circ})^{-1}q_{\mathrm F}^{\circ},
\)
and the minimizer in \eqref{eq:theta-fisher-star} is unique.

Define
\(
Y_{ij}^{\mathrm F}(p,q)
=
Y_{ij}(p,q)
-
\dot X_{ij}^\top\theta_{\mathrm F}^\star
\).
For each $(p,q)$, its finite-population mean is
$(IJ)^{-1}\sum_{i,j}Y_{ij}^{\mathrm F}(p,q)=\bar Y(p,q)$.
Let $A_i^{\mathrm F}$, $B_j^{\mathrm F}$, and $\Sigma_{\mathrm F}$ be defined as
$A_i$, $B_j$, and $\Sigma$ in Section~\ref{subsec:twoway-vector}, but with
$Y_{ij}(p,q)$ replaced by $Y_{ij}^{\mathrm F}(p,q)$. That is,
\[
\Sigma_{\mathrm F}
=
\frac{e_1^B e_0^B}{I}
\cdot \mathcal S_I(A^{\mathrm F})
+
\frac{e_1^S e_0^S}{J}
\cdot\mathcal S_J(B^{\mathrm F}).
\]

We first impose a regularity condition shared by the Fisher and Lin adjustments. Let $\|\cdot\|$ denote the Euclidean norm and $\|\cdot\|_{\mathrm F}$ denote the Frobenius norm.
\begin{condition}[Covariate regularity]
\label{cond:covariate-regularity}
As $I,J\to\infty$, the covariate dimension $d$ is fixed,
and
\(
\lambda_{\min}(Q_{X}^{\circ})\ge c_Q>0
\)
for some constant $c_Q>0$. In addition, the covariate and outcome--covariate moments satisfy
\[
\frac{1}{IJ}\sum_{i=1}^I\sum_{j=1}^J
\|\dot X_{ij}\|^2
+
\frac{1}{IJ}\sum_{i=1}^I\sum_{j=1}^J
\|\dot X_{ij}\dot X_{ij}^\top-Q_{X}^{\circ}\|_{\mathrm F}^2
=O(1),
\]
and, for every $(p,q)\in\{0,1\}^2$,
\[
\frac{1}{IJ}\sum_{i=1}^I\sum_{j=1}^J
|\dot Y_{ij}(p,q)|^2
+
\frac{1}{IJ}\sum_{i=1}^I\sum_{j=1}^J
\|\dot X_{ij}\dot Y_{ij}(p,q)-M_{XY}(p,q)\|^2
=O(1),
\]
where
\(
M_{XY}(p,q)
=
(IJ)^{-1}\sum_{i=1}^I\sum_{j=1}^J
\dot X_{ij}\dot Y_{ij}(p,q).
\)
\end{condition}

\begin{condition}\label{cond:fisher-adjustment}
The residualized potential outcomes
$\{Y_{ij}^{\mathrm F}(p,q):p,q\in\{0,1\}\}$ satisfy
Condition~\ref{cond:vector-clt} with leading covariance matrix
$\Sigma_{\mathrm F}$.
\end{condition}
\begin{theorem}
\label{thm:fisher-adjustment}
Under Assumptions~\ref{ass:smrd-randomization} and~\ref{ass:local-interference}, and
Condition~\ref{cond:covariate-regularity} and \ref{cond:fisher-adjustment}, we have
\[
\Sigma_{\mathrm F}^{-1/2}
\left(
\widehat{\boldsymbol\mu}_{\mathrm F}
-
\bar{\mathbf Y}
\right)
\Rightarrow
N(0,\mathrm I_4).
\]
\end{theorem}

Let \(\widehat\Omega_{\mathrm{2w}}^{\mathrm F}\) be the dummy-coefficient block of the two-way cluster-robust covariance matrix from Fisher's regression
\eqref{eq:reg-fisher-adjustment}.
 We impose the following additional fourth-moment
condition for the variance estimator.

\begin{condition}\label{cond:fisher-variance}
Assume that there exists a constant
\(C_{F,4}<\infty\) such that, for every \((p,q)\in\{0,1\}^2\), we have
\[
\frac{1}{IJ}
\sum_{i=1}^I\sum_{j=1}^J
\{Y_{ij}^{\mathrm F}(p,q)-\bar Y(p,q)\}^4
\le C_{F,4}.
\]
\end{condition}

Define $A_i^{\mathrm F,(1)}$, $A_i^{\mathrm F,(0)}$,
$B_j^{\mathrm F,(1)}$, and $B_j^{\mathrm F,(0)}$
from the Fisher-residualized potential outcomes in the same way as
$A_i^{(1)}$, $A_i^{(0)}$, $B_j^{(1)}$, and $B_j^{(0)}$ defined in Section~\ref{app:proof-two-way}.

\begin{theorem}
\label{thm:fisher-2w-conservative}
Under Assumptions~\ref{ass:smrd-randomization} and~\ref{ass:local-interference},
and Conditions~\ref{cond:covariate-regularity}, \ref{cond:fisher-adjustment} and~\ref{cond:fisher-variance}, we have
\begin{equation}
\label{eq:fisher-2w-conservative}
\widehat\Omega_{\mathrm{2w}}^{\mathrm F}-\Sigma_{\mathrm F}
=
\frac{1}{I}
\mathcal S_I\left(
e_1^B A^{\mathrm F,(1)}
+
e_0^B A^{\mathrm F,(0)}
\right)
+
\frac{1}{J}
\mathcal S_J\left(
e_1^S B^{\mathrm F,(1)}
+
e_0^S B^{\mathrm F,(0)}
\right)
+
o_{\mathbb P}(I^{-1}+J^{-1}).
\end{equation}
\end{theorem}

Together, Theorem \ref{thm:fisher-adjustment} and \ref{thm:fisher-2w-conservative} show that directly running Fisher's regression with the usual two-way cluster-robust covariance matrix yields asymptotically valid, conservative inference for the exposure-mean vector and its linear contrasts.

\subsection{Lin's regression in two-sided experiments}
\label{subsec:lin-adjustment}

For Lin's regression, the covariate slope is allowed to vary across the four exposure
cells. The first-order condition for the exposure-dummy coefficient in cell \((p,q)\) implies
\begin{equation}\label{eq:mu-lin-cell-representation}
\widehat\mu_{\mathrm L,pq}
=
\widehat Y(p,q)
-
\{\widehat X(p,q)-\bar X\}^\top\widehat\theta_{\mathrm L,pq},
\qquad (p,q)\in\{0,1\}^2.
\end{equation}
For each \((p,q)\in\{0,1\}^2\), by the Frisch--Waugh--Lovell theorem,  the population
Lin slope equals
\begin{equation}\label{eq:theta-lin-star}
\theta_{\mathrm L,pq}^\star
\in
\arg\min_{\theta\in\mathbb R^d}
\frac{1}{IJ}
\sum_{i=1}^I\sum_{j=1}^J
\left[
\dot Y_{ij}(p,q)
-
\dot X_{ij}^\top\theta
\right]^2.
\end{equation}
Equivalently, recall $Q_{X}^{\circ}$ in \eqref{eq:QX_def} and let 
\[
q_{\mathrm L,pq}^{\circ}
=M_{XY}(p,q)=
\frac{1}{IJ}
\sum_{i=1}^I\sum_{j=1}^J
\dot X_{ij}\dot Y_{ij}(p,q).
\]
When \(Q_{X}^{\circ}\) is nonsingular, the minimizer in
\eqref{eq:theta-lin-star} is unique and
\(
\theta_{\mathrm L,pq}^\star
=
(Q_{X}^{\circ})^{-1}q_{\mathrm L,pq}^{\circ}.
\)

Define the residualized potential outcomes
\(
Y_{ij}^{\mathrm L}(p,q)
=
Y_{ij}(p,q)
-
\dot X_{ij}^\top\theta_{\mathrm L,pq}^\star.
\)
Since \((IJ)^{-1}\sum_{i,j}\dot X_{ij}=0\), the finite-population mean of
\(Y_{ij}^{\mathrm L}(p,q)\) equals \(\bar Y(p,q)\). Let
\(A_i^{\mathrm L}\), \(B_j^{\mathrm L}\), and \(\Sigma_{\mathrm L}\) be defined as
\(A_i\), \(B_j\), and \(\Sigma\) in Section~\ref{subsec:twoway-vector}, but with
\(Y_{ij}(p,q)\) replaced by \(Y_{ij}^{\mathrm L}(p,q)\). That is,
\[
\Sigma_{\mathrm L}
=
\frac{e_1^B e_0^B}{I}
\cdot \mathcal S_I(A^{\mathrm L})
+
\frac{e_1^S e_0^S}{J}
\cdot\mathcal S_J(B^{\mathrm L}).
\]

\begin{condition}\label{cond:lin-adjustment}
The residualized potential outcomes
    \(\{Y_{ij}^{\mathrm L}(p,q):p,q\in\{0,1\}\}\) satisfy
    Condition~\ref{cond:vector-clt} with leading covariance matrix
    \(\Sigma_{\mathrm L}\).
\end{condition}

\begin{theorem}
\label{thm:lin-adjustment}
Under Assumptions~\ref{ass:smrd-randomization} and~\ref{ass:local-interference}, and
Condition~\ref{cond:covariate-regularity} and \ref{cond:lin-adjustment}, we have
\[
\Sigma_{\mathrm L}^{-1/2}
\left(
\widehat{\boldsymbol\mu}_{\mathrm L}
-
\bar{\mathbf Y}
\right)
\Rightarrow
N(0,\mathrm I_4).
\]
\end{theorem}

Let \(\widehat\Omega_{\mathrm{2w}}^{\mathrm L}\) be the dummy-coefficient block of the two-way cluster-robust covariance matrix from Lin's regression
\eqref{eq:reg-lin-adjustment}.
We impose the following additional fourth-moment condition for the variance estimator.

\begin{condition}\label{cond:lin-variance}
Assume that there exists a constant
\(C_{L,4}<\infty\) such that, for every \((p,q)\in\{0,1\}^2\),
\[
\frac{1}{IJ}
\sum_{i=1}^I\sum_{j=1}^J
\{Y_{ij}^{\mathrm L}(p,q)-\bar Y(p,q)\}^4
\le C_{L,4}.
\]
\end{condition}

Define $A_i^{\mathrm L,(1)}$, $A_i^{\mathrm L,(0)}$,
$B_j^{\mathrm L,(1)}$, and $B_j^{\mathrm L,(0)}$
from the Lin-residualized potential outcomes in the same way as
$A_i^{(1)}$, $A_i^{(0)}$, $B_j^{(1)}$, and $B_j^{(0)}$ defined in Section~\ref{app:proof-two-way}.
\begin{theorem}
\label{thm:lin-2w-conservative}
Under Assumptions~\ref{ass:smrd-randomization} and~\ref{ass:local-interference},
and Conditions~\ref{cond:covariate-regularity}, \ref{cond:lin-adjustment} and~\ref{cond:lin-variance}, we have
\begin{equation}\label{eq:lin-2w-conservative}
\widehat\Omega_{\mathrm{2w}}^{\mathrm L}-\Sigma_{\mathrm L}
=
\frac{1}{I}
\mathcal S_I\left(
e_1^B A^{\mathrm L,(1)}
+
e_0^B A^{\mathrm L,(0)}
\right)
+
\frac{1}{J}
\mathcal S_J\left(
e_1^S B^{\mathrm L,(1)}
+
e_0^S B^{\mathrm L,(0)}
\right)
+
o_{\mathbb P}(I^{-1}+J^{-1}).
\end{equation}
\end{theorem}

Similarly, Theorem \ref{thm:lin-adjustment} and \ref{thm:lin-2w-conservative} show that directly running Lin's fully interacted regression with the usual two-way cluster-robust covariance matrix yields asymptotically valid, conservative inference for the exposure-mean vector and its fixed linear contrasts.

\subsection{No-harm guarantee}
\label{subsec:no-harm-guarantee}
We next study conditions under which covariate adjustment is guaranteed not to reduce asymptotic efficiency.
To align the adjustment with the two-way structure, define the buyer-, seller-, and buyer-seller
components of the centered covariate by
\[
X_i^B=\bar X_{i\cdot}-\bar X,\qquad
X_j^S=\bar X_{\cdot j}-\bar X,\qquad
X_{ij}^{BS}
=
X_{ij}-\bar X_{i\cdot}-\bar X_{\cdot j}+\bar X,
\]
so that
\(
X_{ij}-\bar X=X_i^B+X_j^S+X_{ij}^{BS}.
\)
Let
\(
X_{ij}^{\mathrm A}
=
\{(X_i^B)^\top,(X_j^S)^\top,(X_{ij}^{BS})^\top\}^\top.
\) 
Consider Fisher's and Lin's regressions with $\dot X_{ij}$ replaced by
$X_{ij}^{\mathrm A}$, and refer to them as the
\emph{ANOVA-decomposed Fisher} and
\emph{ANOVA-decomposed Lin} regressions, respectively. 
Let $\Sigma_{\mathrm{AF}}$ and $\Sigma_{\mathrm{AL}}$ denote their
leading design covariance matrices.
For $(p,q)\in\{0,1\}^2$, define the buyer- and seller-side
cell-specific population projection slopes by
\[
\gamma_{pq}^{B}
=
\{\mathcal S_I(X^B)\}^{-1}
\frac{1}{I-1}
\sum_{i=1}^I
X_i^B
\{\bar Y_{i\cdot}(p,q)-\bar Y(p,q)\},
\]
and
\[
\gamma_{pq}^{S}
=
\{\mathcal S_J(X^S)\}^{-1}
\frac{1}{J-1}
\sum_{j=1}^J
X_j^S
\{\bar Y_{\cdot j}(p,q)-\bar Y(p,q)\}.
\]
Define the corresponding common Fisher slopes as
\(
\gamma_{\mathrm F}^{B}
=
\sum_{p,q}e_p^Be_q^S\gamma_{pq}^{B},
\) and \(
\gamma_{\mathrm F}^{S}
=
\sum_{p,q}e_p^Be_q^S\gamma_{pq}^{S}.
\)

\begin{prop}[Efficiency of ANOVA-decomposed regression adjustment]
\label{prop:anova-adjustment-efficiency}
Suppose $\mathcal S_I(X^B)$ and $\mathcal S_J(X^S)$ are nonsingular.
Then the following results hold.

\begin{enumerate}
\item If, for every $(p,q)\in\{0,1\}^2$,
\[(\gamma_{pq}^{B}-\gamma_{\mathrm F}^{B})^\top
\mathcal S_I(X^B)\gamma_{\mathrm F}^{B}=0,\quad 
(\gamma_{pq}^{S}-\gamma_{\mathrm F}^{S})^\top
\mathcal S_J(X^S)\gamma_{\mathrm F}^{S}=0,\]
then
\(
\Sigma_{\mathrm{AF}}\preceq\Sigma.
\)

\item Without any additional assumption, we have
\(
\Sigma_{\mathrm{AL}}
\preceq
\Sigma,
\Sigma_{\mathrm F},
\Sigma_{\mathrm L},
\Sigma_{\mathrm{AF}}.
\)
\end{enumerate}
\end{prop}
Proposition~\ref{prop:anova-adjustment-efficiency}(i) gives a sufficient
condition under which the ANOVA-decomposed Fisher adjustment is no-harm. It holds, in particular, when
$\gamma_{11}^B=\gamma_{10}^B=\gamma_{01}^B=\gamma_{00}^B$
and
$\gamma_{11}^S=\gamma_{10}^S=\gamma_{01}^S=\gamma_{00}^S$.
Proposition~\ref{prop:anova-adjustment-efficiency}(ii) shows that the ANOVA-decomposed Lin adjustment is asymptotically
no-harm without additional restrictions, although its fully interacted
specification estimates more parameters and may therefore incur greater
finite-sample variability.
\begin{remark}[Side-specific covariates]
\label{rem:side-specific-covariates}

For the effects in \eqref{eq:F-main}, if $X_{ij}=X_i^B$, Fisher's regression
leaves the estimators of $\tau_{S,\mathrm{spill}}$ and $\tau_{\mathrm{int}}$
unchanged; only those of $\tau_{\mathrm{tot}}$ and
$\tau_{B,\mathrm{spill}}$ can change, paralleling
\citet{zhao2022splitplot}. For Lin's regression, the seller-side first-order
component is unchanged, $B_j^{\mathrm L}=B_j$, while the buyer-side component
is residualized, yielding $\Sigma_{\mathrm L}\preceq\Sigma$.
The case $X_{ij}=X_j^S$ is symmetric. We give more details in Proposition~\ref{prop:side-specific-covariates} in the Supplementary Materials.

\end{remark}

\subsection{Improved variance estimation for scalar parameters after regression adjustment}
\label{subsec:adjusted-optimized-scalar}

The contrast-specific improvement in Section~\ref{subsec:scalar-inference}
extends to Fisher's and Lin's adjustments, including their ANOVA-decomposed versions,
for which the same argument applies after replacing the original covariates by
$X_{ij}^{\mathrm A}$. For \(m\in\{\mathrm F,\mathrm L\}\), first residualize the observed outcome
using the estimated common slope from Fisher's regression or the estimated cell-specific slopes
from Lin's regression. The coefficient vector \(\widehat{\boldsymbol\mu}_m\) is exactly the
vector of exposure-cell means of this feasible residualized outcome. We can therefore apply the
buyer- and seller-side projection construction in \eqref{eq:Vopt-sample} to the residualized
outcome and obtain an adjusted optimized variance estimator
\(\widehat V_{\mathrm{opt},\beta}^m\) for the estimator
\(\beta^\top\widehat{\boldsymbol\mu}_m\).

Section~\ref{app:adjusted-optimized-scalar} in the Supplementary Materials gives the formal construction
and result. Under the same conditions used for the corresponding adjusted vector CLT and
two-way covariance theorem, \(\widehat V_{\mathrm{opt},\beta}^m\) consistently estimates a
conservative target for \(\beta^\top\Sigma_m\beta\). This target is no larger than the first-order target of the usual two-way clustered variance from the same adjusted
regression, so the resulting normal interval is asymptotically valid and weakly sharper at the
level of conservative population targets. This comparison does not impose a finite-sample
ordering between the two reported variance estimates.

\section{Simulation}\label{sec:simulation}

This section studies the finite-sample behavior of the inference procedures developed above.
The simulations follow \citet{sudijono2026regression}: for each
data-generating process and treatment fraction, we first fix a finite population of potential
outcomes and then repeatedly draw buyer and seller assignments independently from their
complete-randomization distributions.
Unless stated otherwise, each setting uses $I=200$, $J=150$, and $5{,}000$
independent assignment draws. We consider treatment fractions
$e_1^B=e_1^S\in\{0.1,0.2,0.5\}$.  The corresponding treated-side counts are
$(I_1,J_1)=(20,15),(40,30),(100,75)$. 

\subsection{Data-generating processes}

We consider two data-generating processes. The first is a Gaussian factor model.
Let $U_i,V_j\in\mathbb R^2$ be centered buyer- and seller-side Gaussian factors and let
$E_{ij,pq}$ be independent cell-specific Gaussian arrays that are double-centered. Potential
outcomes are generated as
\[
Y_{ij}(p,q)
=
\mu_{pq}+(\ell^B_{pq})^\top U_i+(\ell^S_{pq})^\top V_j
+\sigma_{pq}E_{ij,pq},
\qquad p,q\in\{0,1\}.
\]
In $(00,01;10,11)$ cell order, the loading vectors are
\[
\ell^B=
\begin{pmatrix}
(0.4,1.0)&(-0.6,1.4)\\
(2.0,-0.3)&(1.3,0.7)
\end{pmatrix},\qquad
\ell^S=
\begin{pmatrix}
(1.1,-0.4)&(2.0,0.3)\\
(-0.2,1.5)&(0.6,1.8)
\end{pmatrix},
\]
and $(\sigma_{00},\sigma_{10},\sigma_{01},\sigma_{11})=(0.5,0.8,0.8,1.2)$.
We set
$(\mu_{00},\mu_{10},\mu_{01},\mu_{11})=(0,1,1,4)$, so the finite-population values of the total,
buyer spillover, seller spillover, and interaction effect are $4$, $1$, $1$, and $2$, respectively, up to
numerical precision.  The cell-specific loadings make the leading covariance of the four
exposure means full rank, while retaining substantial asymmetry across sides and exposure
states. For regression adjustment, the four fixed covariates are noisy proxies for the two
coordinates of $U_i$ and $V_j$, with independent Gaussian measurement noise of standard
deviation $0.65$.

The second design is the strategic marketplace design of \citet{sudijono2026regression}. We draw and fix
$m_{ij}\sim\mathrm{Exp}(1)$, $r_i^c\sim\mathrm{Unif}(0,0.2)$, and
$r_j^a\sim\mathrm{Unif}(0,0.2)$. To place the design on the bounded-moment scale used in
our asymptotic theory, we normalize the strategic component by $I+J$ and set
\[
Y_{ij}(w)=\frac{\{m_{ij}+\eta w_{ij}\}}{I+J}
\left\{
Jr_i^c\{\bar m_{i\cdot}+\eta\bar w_{i\cdot}\}
+Ir_j^a\{\bar m_{\cdot j}+\eta\bar w_{\cdot j}\}
\right\}+H_{ij}(w),
\qquad \eta=5.
\]
Here $H_{ij}(w)$ is a mean-zero cell-specific buyer-seller background component with scale
$0.25$ and the same full-rank loading structure as the Gaussian design.  This modest component
prevents asymptotic degeneracy without changing the strategic mechanism that drives the
marketplace outcome.  We use the fixed pre-treatment covariates
$X_{ij}=(\widehat m_{ij}r_i^c,\widehat m_{ij}r_j^a)^\top$, where
$\widehat m_{ij}=m_{ij}(1+\xi_{ij})$ and $\xi_{ij}\sim N(0,0.1^2)$.
Under two-sided experiments, $w_{ij}=pq$, $\bar w_{i\cdot}=e_1^S p$, and
$\bar w_{\cdot j}=e_1^B q$, so the four potential outcomes are constructed separately for
each treatment fraction.

For every generated finite population, we audit the assumptions used by the theory. In
particular, we record the scaled minimum eigenvalues
$\lambda_{\min}(\Sigma)/(I^{-1}+J^{-1})$,
$\lambda_{\min}(\Sigma_{\mathrm F})/(I^{-1}+J^{-1})$, and
$\lambda_{\min}(\Sigma_{\mathrm L})/(I^{-1}+J^{-1})$, the ranks of these matrices, the minimum
eigenvalue of the covariate moment matrix, and the largest cell-specific second and fourth
moments. All reported designs have rank four before and after Fisher and Lin residualization.

\subsection{Comparison of variance estimators for the buyer-spillover effect}

We first focus on the buyer-spillover estimand
$\tau_{B,\mathrm{spill}}$. We compare the two-way clustered covariance
from the saturated regression, the optimized estimator in
Section~\ref{subsec:scalar-inference}, and the Liu-, Sudijono-, and Masoero-type conservative
benchmarks. Table~\ref{tab:scalar-variance-comparison} reports coverage and average interval
length for both designs and all three treatment fractions.

\begin{table}[!htb]
\centering
\caption{Scalar coverage and confidence-interval length for the buyer-spillover estimand.}
\label{tab:scalar-variance-comparison}
\small
\begin{tabular}{lrrrrrr}
\toprule
& \multicolumn{2}{c}{Fraction $0.1$}
& \multicolumn{2}{c}{Fraction $0.2$}
& \multicolumn{2}{c}{Fraction $0.5$} \\
\cmidrule(lr){2-3}\cmidrule(lr){4-5}\cmidrule(lr){6-7}
Method & CP(\%) & CI length & CP(\%) & CI length & CP(\%) & CI length \\
\midrule
\multicolumn{7}{l}{\textit{Panel A: Gaussian asymmetric design}} \\
Two-way & 96.5 & 1.959 & 98.3 & 1.434 & 99.3 & 1.370 \\
Optimized & 95.2 & 1.848 & 96.3 & 1.238 & 97.2 & 1.138 \\
Liu & 95.4 & 1.868 & 96.7 & 1.258 & 97.6 & 1.169 \\
Sudijono & 95.6 & 1.894 & 97.0 & 1.293 & 98.2 & 1.223 \\
Masoero & 99.2 & 2.464 & 98.8 & 1.536 & 98.5 & 1.253 \\
\addlinespace
\multicolumn{7}{l}{\textit{Panel B: Strategic marketplace design}} \\
Two-way & 96.7 & 0.475 & 98.4 & 0.386 & 99.1 & 0.361 \\
Optimized & 95.2 & 0.441 & 95.9 & 0.332 & 97.3 & 0.301 \\
Liu & 95.2 & 0.444 & 96.2 & 0.336 & 97.5 & 0.306 \\
Sudijono & 95.8 & 0.453 & 96.8 & 0.348 & 98.1 & 0.321 \\
Masoero & 98.9 & 0.582 & 98.8 & 0.409 & 98.3 & 0.326 \\
\bottomrule
\end{tabular}
\par\smallskip
\begin{minipage}{0.94\textwidth}
\footnotesize\emph{Notes:} CP is empirical coverage of the nominal 95\%
confidence interval. ``Two-way'' uses
$\beta^\top\widehat\Omega_{\mathrm{2w}}\beta$, where
$\widehat\Omega_{\mathrm{2w}}$ is the conventional covariance in
\eqref{eq:twoway-cov-vector}, including the buyer--seller intersection
subtraction. Each setting uses $5{,}000$ independent assignment draws.
\end{minipage}
\end{table}

The ``Two-way'' row is exactly the conventional covariance in
\eqref{eq:twoway-cov-vector}, evaluated as
\(\beta^\top\widehat\Omega_{\mathrm{2w}}\beta\) for the buyer-spillover
contrast. Thus, it includes the buyer- and seller-cluster score components and subtracts the
buyer-seller intersection component; it is neither the additive analogue nor the optimized
bound. Its relatively long intervals reflect the conservative \((1,1)\) first-order target
characterized in Section~\ref{app:optimized-variance-family}. The intersection correction is
of smaller order and therefore does not generally eliminate this first-order conservative gap.

The optimized procedure has the shortest intervals in every setting, followed closely by the
Liu-type construction; the Sudijono- and Masoero-type bounds are more conservative. The two-way procedure also overcovers, especially under balanced assignment. Coverage for the
optimized procedure is near or above the nominal level throughout, including the deliberately
difficult treatment fraction of \(0.1\). These rankings are consistent with
Proposition~\ref{prop:scalar-variance-comparison}: preserving jointly observed within-state
variation improves precision while retaining conservative coverage.

\subsection{Regression adjustment}

We next compare three regression specifications: the unadjusted saturated regression,
Fisher's common-slope regression in \eqref{eq:reg-fisher-adjustment}, and Lin's fully
interacted regression in \eqref{eq:reg-lin-adjustment}.  Crucially, every point estimator is
paired with the two-way clustered covariance matrix from the same regression.  Thus,
this experiment directly evaluates Theorems~\ref{thm:fisher-2w-conservative} and
\ref{thm:lin-2w-conservative}; it does not use a separately residualized optimized bound.
We report bias, empirical standard deviation, RMSE, interval length, and coverage.

\begin{table}[htbp]
\centering
\caption{Regression adjustment for the buyer-spillover estimand at treatment fraction $0.2$.}
\label{tab:regadj-combined-0_2}
\small
\begin{tabular}{lrrrrr}
\toprule
Method & Bias ($\times 100$) & SD & RMSE & CI length & CP(\%) \\
\midrule
\multicolumn{6}{l}{\textit{Panel A: Gaussian asymmetric design}} \\
Unadjusted & 0.949 & 0.293 & 0.293 & 1.434 & 98.3 \\
Fisher & 0.739 & 0.272 & 0.272 & 1.373 & 98.8 \\
Lin & 0.271 & 0.175 & 0.175 & 0.826 & 98.2 \\
\addlinespace
\multicolumn{6}{l}{\textit{Panel B: Strategic marketplace design}} \\
Unadjusted & -0.049 & 0.080 & 0.080 & 0.386 & 98.4 \\
Fisher & -0.058 & 0.079 & 0.079 & 0.383 & 98.4 \\
Lin & -0.024 & 0.079 & 0.079 & 0.381 & 98.3 \\
\bottomrule
\end{tabular}
\par\smallskip
\begin{minipage}{0.90\textwidth}
\footnotesize\emph{Notes:} Each method uses the two-way clustered
covariance from the same regression. Results are based on $5{,}000$
independent assignment draws.
\end{minipage}
\end{table}

In the Gaussian design, Lin's adjustment has the smallest dispersion, RMSE, and interval length,
whereas Fisher's delivers a more modest gain. In the marketplace design, all three specifications
perform similarly because the available covariates add little predictive power. Bias is
negligible and coverage is conservative in both panels, consistent with
Theorems~\ref{thm:fisher-2w-conservative} and~\ref{thm:lin-2w-conservative}. The same
qualitative pattern holds at the other treatment fractions. The only mild anomaly is slight
undercoverage for Lin in the most unbalanced Gaussian setting; its deviation from nominal is
within two Monte Carlo standard errors and is therefore compatible with simulation noise.

\subsection{Joint inference}

To evaluate the vector results, we form the three-dimensional parameter consisting of the
total effect, buyer-spillover effect, and seller-spillover effect. These contrasts are linearly
independent, and the interaction effect is their linear combination.  We report coverage of the
95\% Wald ellipsoid based on the two-way clustered covariance from the unadjusted,
Fisher, and Lin regressions. For point estimation, we report the Euclidean norm of the bias
vector, total standard deviation (the square root of the trace of the Monte Carlo covariance),
and vector RMSE. For inference, we report joint coverage and the geometric mean of the three
full ellipsoid-axis lengths, which is the multivariate analogue of confidence-interval length.
Table~\ref{tab:joint-wald-simulation-0_2} shows the results at treatment fraction \(0.2\).
Balanced-assignment and sample-size results are in the replication package.

\begin{table}[!htb]
\centering
\caption{Joint estimation and 95\% Wald inference at treatment fraction $0.2$.}
\label{tab:joint-wald-simulation-0_2}
\footnotesize
\begin{tabular}{lrrrrr}
\toprule
Method & Bias norm ($\times 100$) & Total SD & RMSE & CP(\%) & GM axis length \\
\midrule
\multicolumn{6}{l}{\textit{Panel A: Gaussian asymmetric design}} \\
Unadjusted & 1.512 & 0.579 & 0.579 & 96.1 & 1.933 \\
Fisher & 1.118 & 0.478 & 0.478 & 96.8 & 1.713 \\
Lin & 0.495 & 0.350 & 0.350 & 94.5 & 1.117 \\
\addlinespace
\multicolumn{6}{l}{\textit{Panel B: Strategic marketplace design}} \\
Unadjusted & 0.186 & 0.179 & 0.179 & 95.7 & 0.578 \\
Fisher & 0.185 & 0.175 & 0.175 & 95.7 & 0.572 \\
Lin & 0.197 & 0.171 & 0.171 & 95.6 & 0.558 \\
\bottomrule
\end{tabular}
\par\smallskip
\begin{minipage}{0.96\textwidth}
\scriptsize\emph{Notes:} The three coordinates are the total,
buyer-spillover, and seller-spillover effects. Bias norm is the Euclidean norm
of the Monte Carlo mean error; Total SD is the square root of the trace of the
Monte Carlo covariance; RMSE is the root mean squared Euclidean error. GM axis
length is the average geometric mean of the three full axis lengths of the
95\% Wald ellipsoid, equal in draw $r$ to
$2\sqrt{\chi^2_{3,0.95}}\det(\widehat\Sigma_r)^{1/6}$.
Each method uses the two-way clustered covariance from the same
regression. Results use $5{,}000$ independent assignment draws.
\end{minipage}
\end{table}

In the Gaussian design, Lin has the smallest total dispersion, vector RMSE, and ellipsoid
axis length, with Fisher between Lin and the unadjusted specification. Lin's mild
undercoverage is within two Monte Carlo standard errors of the nominal level. In the
marketplace design, the three methods have similar point-estimation error and joint-region
size, and all have conservative coverage. All estimated contrast covariance matrices are
positive definite. The balanced designs are more conservative, consistent with the
nonvanishing positive-semidefinite gap characterized by the theory.

\section{Application to a Household Financial Network}
\label{sec:real-data}

We illustrate our inference methods using the randomized savings-account experiment conducted near
Pokhara, Nepal \citep{prina2015banking}, and the associated dyadic network data analyzed by
\citet{comola2021treatment} and \citet{gao2024endogenous}. The analysis file contains 915 households across 19 villages,
and the experiment offered access to a formal savings account to a random subset of these
households. Following the analysis in
\citet{liu2025randomization}, we classify a household as a buyer if it experienced a death or
livestock shock at baseline and as a seller otherwise. This partition is determined before
treatment and gives an economic interpretation to the two sides: shocked households are more
likely to demand informal financial support, whereas non-shocked households are potential
providers of that support.

The outcome is a directed measure of change in the financial network. Let
$G_{ij}^{(t)}$ indicate a reported financial link from household $i$ to household $j$ at survey
wave $t\in\{0,1\}$. We semi-row-standardize each adjacency matrix, so every non-isolated row
sums to one and every isolated row remains zero, and define
\(
Y_{ij}^{\mathrm{obs}}
=
G_{ij}^{(1)}-G_{ij}^{(0)}.
\)
The local-interference interpretation is that the change in a directed relationship depends on
the savings-account offers made to the two households in that dyad, but not on offers made to
other households. The four exposure conditions are therefore indexed by whether the shocked
household and the non-shocked household were offered an account.

\subsection{Sample construction and estimation}
\label{subsec:real-data-design}

Following \citet{liu2025randomization}, we reinterpret the household assignment after the
baseline buyer-seller partition as a two-sided complete-randomization design. We pool buyers
and sellers across villages and condition on their observed global treatment margins. Three
villages contain no buyers under this partition and hence contribute no buyer-seller dyads.
The resulting pooled sample spans 16 villages: the buyer side contains \(I=76\) households, of
which \(I_1=43\) were offered an account, and the seller side contains \(J=684\) households,
of which \(J_1=338\) were offered an account. We regard the observed side-level assignments as
one realization of independent complete randomizations with these fixed margins, as in
Assumption~\ref{ass:smrd-randomization}.

The observed financial network is block diagonal by village. We retain the complete directed
buyer-seller dyads recorded within villages and define cross-village dyads as structural zeros,
which yields one complete \(76\times684\) buyer-seller outcome matrix. The same convention is
used for the pair covariates in the regression-adjusted analysis. This pooled construction is an
empirical reinterpretation designed to illustrate the two-sided methods. We therefore apply
\(\widehat{\mathbf Y}\), the two-way covariance, and
\(\widehat V_{\mathrm{opt},\beta}\) directly to this single market, without village-specific
estimation or weighting.

\subsection{Estimated effects and variance comparison}
\label{subsec:real-data-results}

All four unadjusted point estimates are close to
zero, every optimized confidence interval contains zero, and the joint null for the total,
buyer-spillover, and seller-spillover effects is not rejected
(\(p=0.881\)). The estimated joint covariance matrix is positive definite.
The replication package provides detailed results, variance decompositions,
and analyses of two additional demographic outcomes.

We next use the baseline absolute differences in household size and in the number of children
under 16 as pair-level covariates.

\begin{table}[!htb]
\centering
\caption{Regression-adjusted estimates with two-way clustered standard errors}
\label{tab:real-data-regression-adjustment}
\footnotesize
\begin{tabular}{llrrr}
\toprule
Specification & Estimand & Estimate ($\times 10^4$)
& Standard error ($\times 10^4$) & $p$-value \\
\midrule
Unadjusted & Total effect & -0.92 & 2.93 & 0.753 \\
 & Buyer spillover & -1.97 & 2.84 & 0.488 \\
 & Seller spillover & -2.38 & 3.37 & 0.479 \\
 & Interaction effect & 3.43 & 4.28 & 0.422 \\
\addlinespace
Fisher & Total effect & -0.93 & 2.93 & 0.752 \\
 & Buyer spillover & -1.98 & 2.83 & 0.485 \\
 & Seller spillover & -2.39 & 3.37 & 0.478 \\
 & Interaction effect & 3.44 & 4.27 & 0.420 \\
\addlinespace
Lin & Total effect & -0.91 & 2.91 & 0.753 \\
 & Buyer spillover & -1.96 & 2.82 & 0.486 \\
 & Seller spillover & -2.32 & 3.44 & 0.500 \\
 & Interaction effect & 3.37 & 4.27 & 0.429 \\
\bottomrule
\end{tabular}
\par\smallskip
\begin{minipage}{0.96\textwidth}
\scriptsize\emph{Notes:}
Estimates and standard errors are multiplied by $10^4$.
Fisher uses common slopes and Lin uses exposure-cell-specific slopes for two
baseline pair covariates: absolute differences in household size and in the
number of children under 16. Cross-village dyads and their pair covariates are
defined as structural zeros. Standard errors use the two-way clustered
covariance from the same pooled regression.
\end{minipage}
\end{table}

Table~\ref{tab:real-data-regression-adjustment} reports Fisher's common-slope and Lin's
cell-specific-slope regressions, pairing each point
estimator with the two-way covariance matrix from the same regression, as required by
Section~\ref{sec:regression-adjustment}. The three specifications produce very similar point
estimates and standard errors: relative to the unadjusted specification, the adjusted standard
errors change by less than one percent except for Lin's seller-spillover standard error, which
increases by about \(2.2\%\). This similarity may reflect the limited explanatory power of the
two baseline pair covariates for the sparse network-change outcome: only \(3{,}512\) of the
\(51{,}984\) pooled dyads are observed within villages, and only 66 record a nonzero change.
Consequently, covariate residualization changes the clustered score variation only modestly.
The slight increase under Lin is also compatible with our theory, which
establishes asymptotic validity and conservativeness for each adjusted specification but does not
impose a finite-sample variance ordering across the three regressions. None of the four effects is
statistically distinguishable from zero under any specification.

\bibliographystyle{main}
\bibliography{main}

\setcounter{section}{0}
\setcounter{subsection}{0}
\setcounter{equation}{0}
\setcounter{figure}{0}
\setcounter{table}{0}

\setcounter{theorem}{0}
\setcounter{lemma}{0}
\setcounter{prop}{0}
\setcounter{condition}{0}
\setcounter{remark}{0}
\setcounter{definition}{0}
\setcounter{assumption}{0}
\setcounter{corollary}{0}
\setcounter{coro}{0}

\renewcommand{\thesection}{S\arabic{section}}
\renewcommand{\thesubsection}{S\arabic{section}.\arabic{subsection}}
\renewcommand{\theequation}{S\arabic{equation}}
\renewcommand{\thefigure}{S\arabic{figure}}
\renewcommand{\thetable}{S\arabic{table}}

\renewcommand{\thetheorem}{S\arabic{theorem}}
\renewcommand{\thelemma}{S\arabic{lemma}}
\renewcommand{\theprop}{S\arabic{prop}}
\renewcommand{\thecondition}{S\arabic{condition}}
\renewcommand{\theremark}{S\arabic{remark}}
\renewcommand{\thedefinition}{S\arabic{definition}}
\renewcommand{\theassumption}{S\arabic{assumption}}
\renewcommand{\thecorollary}{S\arabic{corollary}}
\renewcommand{\thecoro}{S\arabic{coro}}

\makeatletter
\renewcommand{\theHsection}{supp.section.\arabic{section}}
\renewcommand{\theHsubsection}{supp.subsection.\arabic{section}.\arabic{subsection}}
\renewcommand{\theHequation}{supp.equation.\arabic{equation}}
\renewcommand{\theHfigure}{supp.figure.\arabic{figure}}
\renewcommand{\theHtable}{supp.table.\arabic{table}}

\providecommand{\theHtheorem}{}
\providecommand{\theHlemma}{}
\providecommand{\theHprop}{}
\providecommand{\theHcondition}{}
\providecommand{\theHremark}{}
\providecommand{\theHdefinition}{}
\providecommand{\theHassumption}{}
\providecommand{\theHcorollary}{}
\providecommand{\theHcoro}{}
\renewcommand{\theHtheorem}{supp.theorem.\arabic{theorem}}
\renewcommand{\theHlemma}{supp.lemma.\arabic{lemma}}
\renewcommand{\theHprop}{supp.prop.\arabic{prop}}
\renewcommand{\theHcondition}{supp.condition.\arabic{condition}}
\renewcommand{\theHremark}{supp.remark.\arabic{remark}}
\renewcommand{\theHdefinition}{supp.definition.\arabic{definition}}
\renewcommand{\theHassumption}{supp.assumption.\arabic{assumption}}
\renewcommand{\theHcorollary}{supp.corollary.\arabic{corollary}}
\renewcommand{\theHcoro}{supp.coro.\arabic{coro}}
\makeatother

\clearpage

\renewcommand{\thepage}{S\arabic{page}} 
\setcounter{page}{1}

\tolerance=2500
\emergencystretch=2em

\begin{center}
\Huge
Supplementary Materials
\end{center}

These Supplementary Materials are organized into four sections.
Section~\ref{app:additional-results} compares conservative variance targets.
Section~\ref{app:main-proofs} proves the results for unadjusted estimation and inference.
Section~\ref{app:technical-identities} collects supporting sampling, regression, and cluster-score results.
Section~\ref{app:regression-adjustment-proofs} presents regression-adjustment results and proofs.

Sections, equations, and results with an S prefix belong to these Supplementary Materials;
references without an S prefix point to the main text.
We use the notation and randomization assumptions of the main article. In
particular, $\rho_{I,J}=I^{-1}+J^{-1}$; $\|\cdot\|$ denotes the Euclidean norm
for vectors and the Frobenius norm for matrices, and $\|\cdot\|_{\mathrm{op}}$
denotes the operator norm. For a scalar finite population $C_1,\ldots,C_N$,
write $S_C^2=\mathcal S_N(C)$ and $S_C=(S_C^2)^{1/2}$.

\section{Variance comparisons}\label{app:additional-results}
This section compares first-order conservative variance targets. We characterize
the two-parameter variance family underlying the optimized estimator and illustrate
its relation to the alternative targets using the buyer-spillover effect.
\subsection{Conservative variance family}
\label{app:optimized-variance-family}
The main text introduces $V_{\mathrm{opt},\beta}^\circ$ for inference on a scalar parameter.
Here we show that it is the optimal member of a two-parameter family of
conservative variance targets that includes the first-order two-way
cluster-robust target.

For $t_B,t_S>0$, define
\(
V_w^\circ(t_B,t_S;\beta)
=
V_{B,w}^\circ(t_B;\beta)+V_{S,w}^\circ(t_S;\beta),
\)
where
\begin{align*}
V_{B,w}^\circ(t_B;\beta)
&=
\frac{e_1^B(e_0^B+t_Be_1^B)}{I}S_{a,1}^2
+
\frac{e_0^B(e_1^B+t_B^{-1}e_0^B)}{I}S_{a,0}^2,\\
V_{S,w}^\circ(t_S;\beta)
&=
\frac{e_1^S(e_0^S+t_Se_1^S)}{J}S_{b,1}^2
+
\frac{e_0^S(e_1^S+t_S^{-1}e_0^S)}{J}S_{b,0}^2.
\end{align*}
Every member is conservative for the asymptotic variance of $\widehat\tau_\beta$. Indeed,
\begin{align}
V_w^\circ(t_B,t_S;\beta)-\sigma_\beta^2
={}&
\frac{1}{It_B}
S_{\,t_Be_1^Ba(1)+e_0^Ba(0)}^2
+
\frac{1}{Jt_S}
S_{\,t_Se_1^Sb(1)+e_0^Sb(0)}^2
\ge0.
\label{eq:Vw-conservative-identity}
\end{align}
At $(t_B,t_S)=(1,1)$, we have
\(
V_w^\circ(1,1;\beta)
=
\beta^\top\Omega_{\mathrm{2w}}^\circ\beta,
\)
so this member is the deterministic first-order target of the two-way cluster-robust
variance estimator.

The terms depending on $t_B$ have the form
$t_B(e_1^BS_{a,1})^2+t_B^{-1}(e_0^BS_{a,0})^2$; the seller expression is analogous.
Thus, when the relevant standard deviations are positive, the unique interior minimizers are
\(t_B^\star=e_0^BS_{a,0}/(e_1^BS_{a,1})\), \(t_S^\star=e_0^SS_{b,0}/(e_1^SS_{b,1})\),
and
\begin{equation}\label{eq:Vw-infimum}
\inf_{t_B,t_S>0}V_w^\circ(t_B,t_S;\beta)
=
\frac{e_1^Be_0^B}{I}(S_{a,1}+S_{a,0})^2
+
\frac{e_1^Se_0^S}{J}(S_{b,1}+S_{b,0})^2
=
V_{\mathrm{opt},\beta}^\circ.
\end{equation}
If one standard deviation on a side is zero and the other is positive, the infimum is approached
as the corresponding tuning parameter tends to zero or infinity. If both are zero, every
positive tuning value is minimizing. Hence the closed form remains valid in all boundary cases.

\begin{corollary}[Exactness of variance targets for a scalar parameter]
\label{cor:scalar-variance-exactness}

For a fixed contrast
$\beta=(\beta_{11},\beta_{10},\beta_{01},\beta_{00})^\top$,
write
\(
\widetilde a_i(p)=a_i(p)-\bar a(p)\), \(
\widetilde b_j(q)=b_j(q)-\bar b(q).
\)
Then the following results hold.

\begin{enumerate}

\item 
We have
\(
\beta^\top\Omega_{\mathrm{2w}}^\circ\beta
=
\sigma_\beta^2
\)
if and only if for all $i,j$,
\[
e_1^B\widetilde a_i(1)+e_0^B\widetilde a_i(0)=0,
\quad
e_1^S\widetilde b_j(1)+e_0^S\widetilde b_j(0)=0.
\]

\item We have
\(
V_{\mathrm{opt},\beta}^\circ
=
\sigma_\beta^2
\)
if and only if there exist constants
$c_1^B,c_0^B,c_1^S,c_0^S\ge 0$ with $c_1^B+c_0^B=c_1^S+c_0^S=1$ such that for all $i,j$, 
\[
c_1^B\widetilde a_i(1)+c_0^B\widetilde a_i(0)=0,\quad c_1^S\widetilde b_j(1)+c_0^S\widetilde b_j(0)=0
.\]
\end{enumerate}
\end{corollary}

The proof of Corollary~\ref{cor:scalar-variance-exactness} is in Section~\ref{app:proof-exactness}.

The sample analogue, used in the numerical comparisons, is
\begin{align}
\widehat V_w(t_B,t_S;\beta)
={}&
\frac{e_1^B(e_0^B+t_Be_1^B)}{I}s_{a,1}^2
+\frac{e_0^B(e_1^B+t_B^{-1}e_0^B)}{I}s_{a,0}^2\notag\\
&+
\frac{e_1^S(e_0^S+t_Se_1^S)}{J}s_{b,1}^2
+\frac{e_0^S(e_1^S+t_S^{-1}e_0^S)}{J}s_{b,0}^2.
\label{eq:Vw-sample}
\end{align}
The same minimization gives
$\inf_{t_B,t_S>0}\widehat V_w(t_B,t_S;\beta)
=\widehat V_{\mathrm{opt},\beta}$, including the boundary cases above.
For fixed positive $t_B,t_S$,
Lemma~\ref{lem:estimated-projection-variances} gives
$\widehat V_w(t_B,t_S;\beta)=V_w^\circ(t_B,t_S;\beta)
+o_{\mathbb P}(I^{-1}+J^{-1})$.

\subsection{Illustration via buyer-spillover effect estimation}\label{app:buyer-spillover-illustration}

We use the buyer-spillover effect as a simple illustration of the variance
comparisons developed above. Its estimand is
\(
\tau_{B,\mathrm{spill}}
=
\bar Y(1,0)-\bar Y(0,0),
\)
corresponding to
$\beta=\beta^{B,\mathrm{spill}}=(0,1,0,-1)^\top$
in the cell ordering $(11),(10),(01),(00)$.
Table~\ref{tab:bspill-var-targets} compares its leading design variance with
several first-order conservative variance targets. The finite-population variances use
denominator $I-1$ for buyer populations and $J-1$ for seller populations.
\begin{table}[!htb]
\centering
\caption{Leading design variance and first-order conservative targets
for the buyer-spillover contrast.}
\label{tab:bspill-var-targets}

\begingroup
\scriptsize
\setlength{\tabcolsep}{4pt}
\renewcommand{\arraystretch}{1.35}

\resizebox{\textwidth}{!}{%
\begin{tabular}{@{}lll@{}}
\toprule
Method & Notation & First-order population target \\
\midrule

Ground truth
&
\(\sigma_{\beta}^2\)
&
\(\displaystyle
e_1^Be_0^B I^{-1}
S^2_{Y^B_{10}/e_1^B+Y^B_{00}/e_0^B}
+
e_1^Se_0^S J^{-1}
S^2_{(Y^S_{10}-Y^S_{00})/e_0^S}
\)
\\
\addlinespace[0.35em]

Ours: joint inference
&
\(\beta^\top\Omega_{\mathrm{2w}}^\circ\beta\)
&
\(\displaystyle
e_1^B I^{-1}S^2_{Y^B_{10}/e_1^B}
+
e_0^B I^{-1}S^2_{Y^B_{00}/e_0^B}
+
e_0^S J^{-1}
S^2_{(Y^S_{10}-Y^S_{00})/e_0^S}
\)
\\
\addlinespace[0.35em]

Ours: optimized variance
&
\(\displaystyle
V_{\mathrm{opt},\beta}^\circ
\)
&
\(\displaystyle
e_1^Be_0^B I^{-1}
(
S_{Y^B_{10}/e_1^B}
+
S_{Y^B_{00}/e_0^B}
)^2
+
e_1^Se_0^S J^{-1}
S^2_{(Y^S_{10}-Y^S_{00})/e_0^S}
\)
\\
\addlinespace[0.35em]

\citet{liu2025randomization}
&
\(\displaystyle
V_{\mathrm{Liu},\beta}^\circ
\)
&
\(\displaystyle
e_1^B I^{-1}S^2_{Y^B_{10}/e_1^B}
+
e_0^B I^{-1}S^2_{Y^B_{00}/e_0^B}
+
e_1^Se_0^S J^{-1}
S^2_{(Y^S_{10}-Y^S_{00})/e_0^S}
\)
\\
\addlinespace[0.35em]

\citet{masoero2026multiple}
&
\(V_{\mathrm{Mas},\beta}^\circ\)
&
\(\displaystyle
2[
e_1^Be_0^B I^{-1}S^2_{Y^B_{10}/e_1^B}
+
e_1^Se_0^S J^{-1}S^2_{Y^S_{10}/e_0^S}
+
e_1^Be_0^B I^{-1}S^2_{Y^B_{00}/e_0^B}
+
e_1^Se_0^S J^{-1}S^2_{Y^S_{00}/e_0^S}
]
\)
\\
\addlinespace[0.35em]

\citet{sudijono2026regression}
&
\(V_{\mathrm{Sud},\beta}^\circ\)
&
\(\displaystyle
[
\{
e_1^Be_0^B I^{-1}S^2_{Y^B_{10}/e_1^B}
+
e_1^Se_0^S J^{-1}S^2_{Y^S_{10}/e_0^S}
\}^{1/2}
+
\{
e_1^Be_0^B I^{-1}S^2_{Y^B_{00}/e_0^B}
+
e_1^Se_0^S J^{-1}S^2_{Y^S_{00}/e_0^S}
\}^{1/2}
]^2
\)
\\

\bottomrule
\end{tabular}%
}


\endgroup
\legend{All entries report first-order population targets; the
\(O((IJ)^{-1})\) buyer-seller interaction components in the cell-mean
variances are omitted.
Write \(Y^B_{10,i}=\bar Y_{i\cdot}(1,0)\),
\(Y^B_{00,i}=\bar Y_{i\cdot}(0,0)\),
\(Y^S_{10,j}=\bar Y_{\cdot j}(1,0)\), and
\(Y^S_{00,j}=\bar Y_{\cdot j}(0,0)\).
The family \(V_w^\circ(t_B,t_S;\beta)\) is defined in
Section~\ref{app:optimized-variance-family}.
The first-order two-way target satisfies
\(
\beta^\top\Omega_{\mathrm{2w}}^\circ\beta
=
V_w^\circ(1,1;\beta)
\).
For the buyer-spillover contrast,
\(
V_{\mathrm{Liu},\beta}^\circ
=
V_w^\circ(1,\infty;\beta)
\),
whereas
\(
V_{\mathrm{opt},\beta}^\circ
=
\inf_{t_B,t_S>0}V_w^\circ(t_B,t_S;\beta)
\).}
\end{table}
For the buyer-spillover contrast, the first-order population targets satisfy
\[
\sigma_\beta^2
\le
V_{\mathrm{opt},\beta}^\circ
\le
V_{\mathrm{Liu},\beta}^\circ
\le
\beta^\top\Omega_{\mathrm{2w}}^\circ\beta.
\]
Thus, the optimized target is at least as sharp as the Liu target, which in
turn is at least as sharp as the first-order two-way cluster-robust target.

\section{Proofs for unadjusted estimation and inference}\label{app:main-proofs}

This section proves the vector CLT, conservativeness of the two-way covariance estimator,
and the optimized variance results for scalar parameters in the unadjusted setting.
The proofs use the supporting results in Section~\ref{app:technical-identities};
in particular, Section~\ref{app:projection-variance-consistency}
establishes consistency of the estimated projection variances.

\subsection{Exact decomposition and variance of the interaction remainder}\label{subsec:exactdecomp}\label{app:auxiliary-results}

For $p,q\in\{0,1\}$, write
\[
U_{ip}=\mathbf 1\{W_i^B=p\}-e_p^B,
\qquad
V_{jq}=\mathbf 1\{W_j^S=q\}-e_q^S.
\]
The fixed treatment counts imply
$\sum_{i=1}^I U_{ip}=0$ for $p\in\{0,1\}$ and
$\sum_{j=1}^J V_{jq}=0$ for $q\in\{0,1\}$. For any fixed array $\{A_{ij}\}_{i\in[I],j\in[J]}$, define its doubly
centered version by
\(
\widetilde A_{ij}
=
A_{ij}-\bar A_{i\cdot}-\bar A_{\cdot j}+\bar A_{\cdot\cdot}.
\)

Lemma~\ref{lem:decomp-yhat} below decomposes the moment estimator into the main effects and interaction of the $U_{ip}$'s and $V_{jq}$'s.
\begin{lemma}\label{lem:decomp-yhat}
For each $(p,q)\in\{0,1\}^2$,
\begin{align}
\widehat Y(p,q)-\bar Y(p,q)
={}&
\frac{1}{I_p}\sum_{i=1}^I\bar Y_{i\cdot}(p,q)U_{ip}
+
\frac{1}{J_q}\sum_{j=1}^J\bar Y_{\cdot j}(p,q)V_{jq}
\notag\\
&+
\frac{1}{I_pJ_q}\sum_{i=1}^I\sum_{j=1}^J
Y_{ij}(p,q)U_{ip}V_{jq}.
\label{eq:exact-cell-decomposition}
\end{align}
The last term is unchanged if $Y_{ij}(p,q)$ is replaced by its doubly centered version
$\widetilde Y_{ij}(p,q)$.
\end{lemma}

\begin{proof}
Substitute
$\mathbf 1\{W_i^B=p\}=e_p^B+U_{ip}$ and
$\mathbf 1\{W_j^S=q\}=e_q^S+V_{jq}$ into the definition of
$\widehat Y(p,q)$. Expanding the product gives a constant term, a buyer term, a seller term,
and a buyer-seller interaction term. Since $I_p=Ie_p^B$ and $J_q=Je_q^S$, the constant term
is $\bar Y(p,q)$ and the two linear terms reduce to the row- and column-average terms in
\eqref{eq:exact-cell-decomposition}. Finally, the zero-sum identities for $U_{ip}$ and
$V_{jq}$ imply
\(\sum_{i,j}Y_{ij}(p,q)U_{ip}V_{jq} = \sum_{i,j}\widetilde Y_{ij}(p,q)U_{ip}V_{jq}\).
\end{proof}

Lemma \ref{lem:bilinear-remainder} below gives the variance of the interaction of the $U_{ip}$'s and $V_{jq}$'s.
\begin{lemma}\label{lem:bilinear-remainder}
Let $C=(c_{ij})$ be a fixed scalar array satisfying
$\sum_j c_{ij}=0$ for every $i$ and $\sum_i c_{ij}=0$ for every $j$, and define
\(
R_C
=
(IJ)^{-1}
\sum_{i=1}^I\sum_{j=1}^J
c_{ij}U_{i1}V_{j1}.
\)
Then
\begin{equation}\label{eq:bilinear-remainder-variance}
\var(R_C)
=
\frac{(e_1^Be_0^B)(e_1^Se_0^S)}
{IJ(I-1)(J-1)}
\sum_{i=1}^I\sum_{j=1}^J c_{ij}^2.
\end{equation}
Furthermore, if $(IJ)^{-1}\sum_{i,j}c_{ij}^2=O(1)$, then
$R_C=O_{\mathbb P}\{(IJ)^{-1/2}\}$.
\end{lemma}

\begin{proof}
Let
\(
\mathbf U_1=(U_{11},\ldots,U_{I1})^\top,
\) and \(
\mathbf V_1=(V_{11},\ldots,V_{J1})^\top.
\)
Let
$P_I=\mathrm I_I-I^{-1}\mathbf 1_I\mathbf 1_I^\top$
and define $P_J$ analogously. Under complete randomization,
\(\E(\mathbf U_1\mathbf U_1^\top) = e_1^Be_0^B I (I-1)^{-1}P_I\), \(\E(\mathbf V_1\mathbf V_1^\top) = e_1^Se_0^S J (J-1)^{-1}P_J\).
Buyer- and seller-side assignments are independent. Moreover, the zero row and
column sums of $C$ imply $P_I C P_J=C$. Therefore,
\begin{align*}
\var(R_C)
=
\frac{1}{I^2J^2}
\operatorname{tr}
\left\{
C^\top
\E(\mathbf U_1\mathbf U_1^\top)
C
\E(\mathbf V_1\mathbf V_1^\top)
\right\}=
\frac{(e_1^Be_0^B)(e_1^Se_0^S)}
{IJ(I-1)(J-1)}
\|C\|_{\mathrm F}^2,
\end{align*}
which proves the result.
\end{proof}

Lemma~\ref{lem:cell-variance-decomposition} below decomposes the variance of each cell-mean estimator into buyer-side, seller-side, and interaction components.
\begin{lemma}\label{lem:cell-variance-decomposition}
Define
\(R_{pq}^2=(I-1)^{-1}\sum_{i=1}^I \{\bar Y_{i\cdot}(p,q)-\bar Y(p,q)\}^2\) and \(C_{pq}^2=(J-1)^{-1}\sum_{j=1}^J \{\bar Y_{\cdot j}(p,q)-\bar Y(p,q)\}^2\).
Then
\begin{equation}\label{eq:cell-mean-variance-decomposition}
\var\{\widehat Y(p,q)\}
=
\frac{e_{1-p}^B}{e_p^BI}R_{pq}^2
+
\frac{e_{1-q}^S}{e_q^SJ}C_{pq}^2
+V_{pq}^{BS},
\end{equation}
where $V_{pq}^{BS}\ge 0$ is the variance of the interaction term in
\eqref{eq:exact-cell-decomposition}.
\end{lemma}

\begin{proof}
The buyer, seller, and interaction terms in
\eqref{eq:exact-cell-decomposition} have mean zero and are mutually
uncorrelated: the buyer and seller terms are uncorrelated by independence
of the two randomizations, while the interaction term has conditional mean
zero given either side's assignment. The first two variance formulas are the
usual complete-randomization formulas \citep{li2017general}, while the last
variance is nonnegative by definition.
\end{proof}

\subsection{Proof of Theorem~\ref{thm:vector-clt}}
\label{app:point-clt-proofs}
The proof proceeds in three steps. We first obtain the vector-valued H\'ajek
projection and the covariance of its first-order term. We then apply a
stratified permutation CLT to the combined buyer- and seller-side first-order
term, and finally show that the standardized interaction remainder is
asymptotically negligible. The second step uses the following lemma.
\begin{lemma}[Stratified permutational CLT; Theorem~2.1 and Remark~2.2(4) of \citet{tuvaandorj2024combinatorial}]
\label{lem:stratified-permutation-clt}
For each stratum $s\in[K]$, let $\pi_s$ be an independent
uniform permutation of $[n_s]$, and let
$\{a_{ij}^{(s)}:i,j\in[n_s]\}$ be a deterministic array satisfying
\(\sum_{i=1}^{n_s} a_{ij}^{(s)} = \sum_{j=1}^{n_s} a_{ij}^{(s)} = 0\), for all \(i,j\in[n_s]\).
Define
\(
T
=
\sum_{s=1}^K
\sum_{i=1}^{n_s}
a_{i,\pi_s(i)}^{(s)}
\)
and
\(
\sigma^2
=
\sum_{s=1}^K
(n_s-1)^{-1}
\sum_{i,j=1}^{n_s}
\bigl\{a_{ij}^{(s)}\bigr\}^2.
\)
Suppose $\sigma^2>0$ and, for some $\delta>0$,
\[
\frac{1}{\sigma^{2+\delta}}
\sum_{s=1}^K
\frac{1}{n_s}
\sum_{i,j=1}^{n_s}
\left|a_{ij}^{(s)}\right|^{2+\delta}
\longrightarrow 0.
\]
Then
\(
{T}/{\sigma}
\Rightarrow N(0,1).
\)
\end{lemma}

\begin{proof}[Proof of Theorem \ref{thm:vector-clt}]
Recall the centered assignment indicators
\(
U_{ip}\) and \(V_{jq}\) for \(p,q\in\{0,1\}\) in Section~\ref{subsec:exactdecomp}.
Then
\(
U_{i0}=-U_{i1},
\) and \(
V_{j0}=-V_{j1}.
\)
By \eqref{eq:exact-cell-decomposition}, collecting the four equations in the ordering
\(
(1,1),(1,0),(0,1),(0,0),
\)
we obtain the vector decomposition
\begin{equation}\label{eq:vector-hoeffding-proof}
\widehat{\mathbf Y}-\bar{\mathbf Y}
=
L_B+L_S+R_{BS},
\end{equation}
where the buyer, seller, and interaction terms are
\(L_B = I^{-1}\sum_{i=1}^I A_iU_{i1}\), \(L_S = J^{-1}\sum_{j=1}^J B_jV_{j1}\), and \(R_{BS} = (IJ)^{-1}\sum_{i=1}^I\sum_{j=1}^J D_{ij}U_{i1}V_{j1}\), respectively.
Here $A_i$ and $B_j$ are defined before \eqref{eq:Sigma-vector}, and
\[
D_{ij}
=
\left(
\frac{Y_{ij}(1,1)}{e_1^Be_1^S},
-\frac{Y_{ij}(1,0)}{e_1^Be_0^S},
-\frac{Y_{ij}(0,1)}{e_0^Be_1^S},
\frac{Y_{ij}(0,0)}{e_0^Be_0^S}
\right)^\top .
\]
Since $\sum_i U_{i1}=0$ and $\sum_j V_{j1}=0$, we may replace $D_{ij}$ by its componentwise
doubly centered version $\widetilde D_{ij}$ in the definition of $R_{BS}$.

First consider the leading term $L_B+L_S$. Its covariance matrix is exactly $\Sigma$:
\(\var(L_B+L_S) = \var(L_B)+\var(L_S) = \Sigma\),
because the buyer-side and seller-side assignments are independent. Indeed, by the complete
randomization variance formula,
\(\var(L_B) = e_1^B e_0^B (I-1)^{-1} \cdot I^{-1}\sum_{i=1}^I (A_i-\bar A)(A_i-\bar A)^\top\),
and similarly
\(\var(L_S) = e_1^S e_0^S (J-1)^{-1} \cdot J^{-1}\sum_{j=1}^J (B_j-\bar B)(B_j-\bar B)^\top\).

We next prove a multivariate CLT for $L_B+L_S$. Fix a nonzero
$t\in\mathbb R^4$ and let
\(\lambda_t=\Sigma^{-1/2}t\), \(x_i=\lambda_t^\top(A_i-\bar A)\), \(y_j=\lambda_t^\top(B_j-\bar B)\).
Because $\sum_iU_{i1}=\sum_jV_{j1}=0$,
\(\lambda_t^\top(L_B+L_S) = I^{-1}\sum_{i=1}^I x_iU_{i1} + J^{-1}\sum_{j=1}^J y_jV_{j1}\).
Define the centered treatment-label vectors
\[
c_k^B
=
\begin{cases}
 e_0^B, & 1\le k\le I_1,\\
 -e_1^B, & I_1<k\le I,
\end{cases}
\qquad
c_\ell^S
=
\begin{cases}
 e_0^S, & 1\le \ell\le J_1,\\
 -e_1^S, & J_1<\ell\le J.
\end{cases}
\]
Let $\pi_B$ and $\pi_S$ be independent uniform permutations of $[I]$ and
$[J]$, respectively. Then
\((U_{11},\ldots,U_{I1}) \overset d= (c_{\pi_B(1)}^B,\ldots,c_{\pi_B(I)}^B)\), \((V_{11},\ldots,V_{J1}) \overset d= (c_{\pi_S(1)}^S,\ldots,c_{\pi_S(J)}^S)\).
Indeed, each buyer assignment with $I_1$ treated units is generated by exactly
$I_1!I_0!$ permutations of the buyer labels, and the seller side is analogous.

For the two strata, define
\(a_{ik}^B=x_i c_k^B I^{-1}\), \(i,k\in[I]\), \(a_{j\ell}^S=y_j c_\ell^S J^{-1}\), \(j,\ell\in[J]\).
It follows that
\begin{equation}\label{eq:two-stratum-permutation-representation}
\lambda_t^\top(L_B+L_S)
\overset d=
\sum_{i=1}^I a_{i,\pi_B(i)}^B
+
\sum_{j=1}^J a_{j,\pi_S(j)}^S.
\end{equation}
Moreover,
\(\sum_i x_i=\sum_j y_j=0\), \(\sum_k c_k^B=\sum_\ell c_\ell^S=0\),
so every row and column sum of each array in
\eqref{eq:two-stratum-permutation-representation} is zero. Thus the leading
term is a two-stratum linear permutation statistic in the sense of
\citet{tian2025stratified} and satisfies the centering condition in
Lemma~\ref{lem:stratified-permutation-clt}.

The variance in Lemma~\ref{lem:stratified-permutation-clt} is
\begin{align*}
\sigma_t^2
&=
\frac{1}{I-1}\sum_{i,k=1}^I(a_{ik}^B)^2
+
\frac{1}{J-1}\sum_{j,\ell=1}^J(a_{j\ell}^S)^2\\
&=
\frac{e_1^Be_0^B}{I(I-1)}\sum_{i=1}^Ix_i^2
+
\frac{e_1^Se_0^S}{J(J-1)}\sum_{j=1}^Jy_j^2\\
&=
\lambda_t^\top\Sigma\lambda_t
=
t^\top t,
\end{align*}
where we used
$\sum_k(c_k^B)^2=Ie_1^Be_0^B$ and
$\sum_\ell(c_\ell^S)^2=Je_1^Se_0^S$. In particular,
$\sigma_t^2=t^\top t>0$.

It remains to verify the Lyapunov condition in
Lemma~\ref{lem:stratified-permutation-clt}.
By Jensen's inequality, Condition~\ref{cond:vector-clt}(iii) implies
\(I^{-1}\sum_{i=1}^I | \bar Y_{i\cdot}(p,q)-\bar Y(p,q) |^{2+\delta} \le C_Y\), \(J^{-1}\sum_{j=1}^J | \bar Y_{\cdot j}(p,q)-\bar Y(p,q) |^{2+\delta} \le C_Y\)
for every $(p,q)$. Together with Condition~\ref{cond:vector-clt}(i) and
the fixed dimension, this gives
\(I^{-1}\sum_{i=1}^I\|A_i-\bar A\|^{2+\delta} + J^{-1}\sum_{j=1}^J\|B_j-\bar B\|^{2+\delta} \le C\).
Since the treatment labels are bounded by one in absolute value, it follows that
\(I^{-1}\sum_{i,k=1}^I|a_{ik}^B|^{2+\delta} \le C I^{-1-\delta}\|\lambda_t\|^{2+\delta}\), \(J^{-1}\sum_{j,\ell=1}^J|a_{j\ell}^S|^{2+\delta} \le C J^{-1-\delta}\|\lambda_t\|^{2+\delta}\).
Condition~\ref{cond:vector-clt}(ii) further implies
\(\|\lambda_t\|^{2+\delta} \le C\rho_{I,J}^{-1-\delta/2}\|t\|^{2+\delta}\).
Since $\sigma_t^{2+\delta}=\|t\|^{2+\delta}$ and
$I^{-1-\delta}+J^{-1-\delta}\le\rho_{I,J}^{1+\delta}$, we obtain
\begin{align*}
\frac{1}{\sigma_t^{2+\delta}}
\left
\{
\frac1I\sum_{i,k=1}^I|a_{ik}^B|^{2+\delta}
+
\frac1J\sum_{j,\ell=1}^J|a_{j\ell}^S|^{2+\delta}
\right\}
&\le C\rho_{I,J}^{\delta/2}
\longrightarrow0.
\end{align*}
Therefore, Lemma~\ref{lem:stratified-permutation-clt} yields
\(\{t^\top t\}^{-1/2}\lambda_t^\top(L_B+L_S) \Rightarrow N(0,1)\).
Equivalently,
\(t^\top\Sigma^{-1/2}(L_B+L_S) \Rightarrow N(0,t^\top t)\).
By the Cram\'er--Wold device,
\begin{equation}\label{eq:leading-vector-clt}
\Sigma^{-1/2}(L_B+L_S)
\Rightarrow
N(0,\mathrm I_4).
\end{equation}

It remains to show that the interaction remainder is negligible. For any vector
$\lambda\in\mathbb R^4$,
\(\lambda^\top R_{BS} = (IJ)^{-1} \sum_{i=1}^I\sum_{j=1}^J \lambda^\top \widetilde D_{ij}U_{i1}V_{j1}\),
where $\widetilde D_{ij}$ is the componentwise doubly centered version of $D_{ij}$.

Applying Lemma~\ref{lem:bilinear-remainder} to the scalar coefficients
$\lambda^\top\widetilde D_{ij}$ gives, for a constant $C$ independent of $I,J$,
\(\var(\lambda^\top R_{BS}) \le C (IJ)^{-1} \{ (IJ)^{-1} \sum_{i=1}^I\sum_{j=1}^J (\lambda^\top\widetilde D_{ij})^2 \}\).
By Cauchy--Schwarz,
\((IJ)^{-1} \sum_{i=1}^I\sum_{j=1}^J (\lambda^\top\widetilde D_{ij})^2 \le \|\lambda\|^2 (IJ)^{-1} \sum_{i=1}^I\sum_{j=1}^J \|\widetilde D_{ij}\|^2\).
We now show that the last factor is bounded under
Conditions~\ref{cond:vector-clt}(i) and~\ref{cond:vector-clt}(iii).
For each $(p,q)$, define the doubly centered potential outcome
\(\widetilde Y_{ij}(p,q) = Y_{ij}(p,q) - \bar Y_{i\cdot}(p,q) - \bar Y_{\cdot j}(p,q) + \bar Y(p,q)\).
Double-centering is an orthogonal projection and therefore does not increase the Frobenius
norm. Condition~\ref{cond:vector-clt}(iii) implies, by the power-mean inequality,
\((IJ)^{-1} \sum_{i=1}^I\sum_{j=1}^J \{Y_{ij}(p,q)-\bar Y(p,q)\}^2 \le C\).
Hence
\((IJ)^{-1} \sum_{i=1}^I\sum_{j=1}^J \widetilde Y_{ij}(p,q)^2 \le C\).
Since each component of $\widetilde D_{ij}$ is a signed and rescaled version of
$\widetilde Y_{ij}(p,q)$, Condition~\ref{cond:vector-clt}(i) implies
\[
\frac{1}{IJ}
\sum_{i=1}^I\sum_{j=1}^J
\|\widetilde D_{ij}\|^2
\le
\sum_{p=0}^1\sum_{q=0}^1
\frac{1}{(e_p^B e_q^S)^2}
\frac{1}{IJ}
\sum_{i=1}^I\sum_{j=1}^J
\widetilde Y_{ij}(p,q)^2
\le C.
\]
Therefore,
\(\var(\lambda^\top R_{BS}) = O(\|\lambda\|^2/(IJ))\).

Taking $\lambda=\lambda_t=\Sigma^{-1/2}t$, we obtain
\(\var\{\lambda_t^\top R_{BS}\} = O(\|\lambda_t\|^2/(IJ))\).
By Condition~\ref{cond:vector-clt}(ii),
\(\|\lambda_t\|^2 = t^\top\Sigma^{-1}t \le \|t\|^2/\lambda_{\min}(\Sigma) \le C\|t\|^2/(I^{-1}+J^{-1})\).
Thus,
\(\var\{\lambda_t^\top R_{BS}\} = O( \{IJ(I^{-1}+J^{-1})\}^{-1} ) = O( (I+J)^{-1} ) \to0\).
Hence
\(t^\top\Sigma^{-1/2}R_{BS} = \lambda_t^\top R_{BS} =o_{\mathbb P}(1)\).
Combining \eqref{eq:leading-vector-clt} with the preceding relation for every fixed $t$ and applying
Slutsky's theorem completes the proof.
\end{proof}

\subsection{Two-way cluster-robust covariance}\label{app:proof-two-way}\label{app:variance-inference}

The following lemma connects the buyer- and seller-level score contributions
from the saturated regression to the corresponding design-based first-order
projections. Specifically, these cluster-level contributions are scaled and
centered sample analogues of the buyer- and seller-side projections.
\begin{lemma}\label{lem:twoway-score-decomposition}
For a fixed scalar contrast $\beta$, let
\(\widehat\psi_i^B(\beta)=\beta^\top(Z^\top Z)^{-1}Z_i^\top\widehat u_i\), \(\widehat\psi_j^S(\beta)=\beta^\top(Z^\top Z)^{-1}Z_{\cdot j}^\top\widehat u_{\cdot j}\), \(\delta_{ij}(\beta)=\beta^\top(Z^\top Z)^{-1}z_{ij}\widehat u_{ij}\).
Then
\[
\beta^\top\widehat\Omega_{\mathrm{2w}}\beta
=
\sum_{i=1}^I \widehat\psi_i^B(\beta)^2
+
\sum_{j=1}^J \widehat\psi_j^S(\beta)^2
-
\sum_{i=1}^I\sum_{j=1}^J \delta_{ij}(\beta)^2.
\]
Moreover, if $W_i^B=p$ and $W_j^S=q$, then
\(
\widehat\psi_i^B(\beta)={I}^{-1}\{\widehat a_i(p)-\widehat{\bar a}(p)\},
\) and \(\widehat\psi_j^S(\beta)={J}^{-1}\{\widehat b_j(q)-\widehat{\bar b}(q)\}\).
Consequently,
\(\sum_{i:W_i^B=p}\widehat\psi_i^B(\beta)^2 =I^{-2}(I_p-1)s_{a,p}^2\), \(\sum_{j:W_j^S=q}\widehat\psi_j^S(\beta)^2 =J^{-2}(J_q-1)s_{b,q}^2\).
The corresponding coordinatewise identities give the matrix decomposition used in the proof
of Proposition~\ref{prop:twoway-vector-conservative}.
\end{lemma}

The detailed algebra proving Lemma~\ref{lem:twoway-score-decomposition} is given in Section~\ref{app:cluster-score-decomposition}.

For buyers, define
\(A_i^{(1)} = \bigl( \bar Y_{i\cdot}(1,1)/e_1^B,\allowbreak \bar Y_{i\cdot}(1,0)/e_1^B,\allowbreak 0,0 \bigr)^\top\), \(A_i^{(0)} = \bigl( 0,0,\allowbreak \bar Y_{i\cdot}(0,1)/e_0^B,\allowbreak \bar Y_{i\cdot}(0,0)/e_0^B \bigr)^\top\).
For sellers, define
\(B_j^{(1)} = \bigl( \bar Y_{\cdot j}(1,1)/e_1^S,\allowbreak 0,\allowbreak \bar Y_{\cdot j}(0,1)/e_1^S,\allowbreak 0 \bigr)^\top\), \(B_j^{(0)} = \bigl( 0,\allowbreak \bar Y_{\cdot j}(1,0)/e_0^S,\allowbreak 0,\allowbreak \bar Y_{\cdot j}(0,0)/e_0^S \bigr)^\top\).
Thus $A_i=A_i^{(1)}-A_i^{(0)}$ and
$B_j=B_j^{(1)}-B_j^{(0)}$. The deterministic first-order target of the two-way cluster-robust covariance is
\begin{equation}\label{eq:Omega2w-pop-target}
\Omega_{\mathrm{2w}}^\circ
=
\frac{e_1^B}{I}\mathcal S_I(A^{(1)})
+
\frac{e_0^B}{I}\mathcal S_I(A^{(0)})
+
\frac{e_1^S}{J}\mathcal S_J(B^{(1)})
+
\frac{e_0^S}{J}\mathcal S_J(B^{(0)}).
\end{equation}

For $W_i^B=1$, let
$\widehat A_i^{(1)}=(\widehat{\bar Y}_{i\cdot}(1,1)/e_1^B,
\widehat{\bar Y}_{i\cdot}(1,0)/e_1^B,0,0)^\top$; for $W_i^B=0$, let
$\widehat A_i^{(0)}=(0,0,\widehat{\bar Y}_{i\cdot}(0,1)/e_0^B,
\widehat{\bar Y}_{i\cdot}(0,0)/e_0^B)^\top$.
Let $\widehat{\mathcal S}_{I,p}(A^{(p)})$ be their sample covariance matrix within buyer
state $p$. Define $\widehat B_j^{(q)}$ and
$\widehat{\mathcal S}_{J,q}(B^{(q)})$ analogously by replacing column means with their
buyer-sample analogues.

\subsubsection*{Proof of Proposition~\ref{prop:twoway-vector-conservative}}

\begin{proof}
The saturated-regression cluster-score algebra in
Lemma~\ref{lem:twoway-score-decomposition} gives
\begin{align*}
\widehat\Omega_{\mathrm{2w}}
={}&
\sum_{p=0}^1\frac{I_p-1}{I^2}
\widehat{\mathcal S}_{I,p}(A^{(p)})
+
\sum_{q=0}^1\frac{J_q-1}{J^2}
\widehat{\mathcal S}_{J,q}(B^{(q)})
-
\widehat\Omega_{BS},
\end{align*}
where the hatted covariance matrices are computed from the observed
buyer- and seller-side projection vectors defined above, and
$\widehat\Omega_{BS}$ is the buyer--seller intersection correction.
Let $\|\cdot\|_{\mathrm{op}}$ denote the operator norm.
Lemma~\ref{lem:estimated-projection-variances}, applied coordinatewise and then
by polarization, gives
\(\| \widehat{\mathcal S}_{I,p}(A^{(p)})-\mathcal S_I(A^{(p)}) \|_{\mathrm{op}} =o_{\mathbb P}(1)\), \(\| \widehat{\mathcal S}_{J,q}(B^{(q)})-\mathcal S_J(B^{(q)}) \|_{\mathrm{op}} =o_{\mathbb P}(1)\).
To control the intersection correction, note that both $Z^\top Z$ and
$\Gamma_{BS}$ are diagonal in the exposure-cell ordering. Hence
\(\widehat\Omega_{BS} := (Z^\top Z)^{-1}\Gamma_{BS}(Z^\top Z)^{-1} = \operatorname{diag}_{(p,q)} \{ (I_pJ_q)^{-2} \sum_{i:W_i^B=p}\sum_{j:W_j^S=q}\widehat u_{ij}^2 \}\).
For each $(p,q)$, the cell mean $\widehat Y(p,q)$ minimizes the within-cell
sum of squares, so
\((I_pJ_q)^{-1} \sum_{i:W_i^B=p}\sum_{j:W_j^S=q}\widehat u_{ij}^2 \le (I_pJ_q)^{-1} \sum_{i,j}G_{pq,ij} \{Y_{ij}(p,q)-\bar Y(p,q)\}^2 = O_{\mathbb P}(1)\).
Indeed, the expectation of the last sample average equals the corresponding
finite-population average, which is bounded by
Condition~\ref{cond:vector-clt}(iii) via the power-mean inequality. The overlap condition therefore gives
\begin{equation}\label{eq:intersection-negligible}
\|\widehat\Omega_{BS}\|_{\mathrm{op}}
=
O_{\mathbb P}\{(IJ)^{-1}\}
=
o_{\mathbb P}(\rho_{I,J}).
\end{equation}
Since
\(I^{-2}(I_p-1)=e_p^B I^{-1}+O(I^{-2})\), \(J^{-2}(J_q-1)=e_q^S J^{-1}+O(J^{-2})\),
we conclude that
\(\| \widehat\Omega_{\mathrm{2w}}-\Omega_{\mathrm{2w}}^\circ \|_{\mathrm{op}} =o_{\mathbb P}(\rho_{I,J})\).

For any two vector finite populations $X,Y$ of the same size and
$e_0+e_1=1$,
\[
e_1\mathcal S(X)+e_0\mathcal S(Y)
-
e_1e_0\mathcal S(X-Y)
=
\mathcal S(e_1X+e_0Y)
\succeq0.
\]
By the definition of $\Sigma$,
\(\Sigma = e_1^Be_0^B I^{-1} \mathcal S_I(A^{(1)}-A^{(0)}) + e_1^Se_0^S J^{-1} \mathcal S_J(B^{(1)}-B^{(0)})\).
Applying the preceding identity on the two sides gives
\[
\Omega_{\mathrm{2w}}^\circ-\Sigma
=
\frac{1}{I}
\mathcal S_I(e_1^BA^{(1)}+e_0^BA^{(0)})
+
\frac{1}{J}
\mathcal S_J(e_1^SB^{(1)}+e_0^SB^{(0)})
\succeq0.
\]
Combining this identity with the established expansion
$\widehat\Omega_{\mathrm{2w}}=\Omega_{\mathrm{2w}}^\circ
+o_{\mathbb P}(\rho_{I,J})$ proves \eqref{eq:Omega2w-conservative}.
\end{proof}

\subsection{Proof of Proposition~\ref{prop:optimized-scalar-variance}}

\begin{proof}
By Lemma~\ref{lem:estimated-projection-variances}, for $p,q\in\{0,1\}$,
\(s_{a,p}^2-S_{a,p}^2=o_{\mathbb P}(1)\), \(s_{b,q}^2-S_{b,q}^2=o_{\mathbb P}(1)\).
The moment and overlap conditions also imply that the population variances are $O(1)$ and
the sample variances are $O_{\mathbb P}(1)$. Since $x,y\ge0$ imply
$|\sqrt{x}-\sqrt{y}|\le\sqrt{|x-y|}$, we have
\(s_{a,p}-S_{a,p}=o_{\mathbb P}(1)\), \(s_{b,q}-S_{b,q}=o_{\mathbb P}(1)\).
Substitution into \eqref{eq:Vopt-sample} yields
\(\widehat V_{\mathrm{opt},\beta} - V_{\mathrm{opt},\beta}^\circ = o_{\mathbb P}(\rho_{I,J})\).

Because $\beta^\top A_i=a_i(1)-a_i(0)$ and
$\beta^\top B_j=b_j(1)-b_j(0)$,
\(\sigma_\beta^2 = e_1^Be_0^B I^{-1}S_{a,\Delta}^2 + e_1^Se_0^S J^{-1}S_{b,\Delta}^2\).
The triangle inequality for finite-population standard deviations gives
\(S_{a,\Delta}\le S_{a,1}+S_{a,0}\), \(S_{b,\Delta}\le S_{b,1}+S_{b,0}\).
Therefore $\sigma_\beta^2\le V_{\mathrm{opt},\beta}^\circ$.
For the upper bound, the definition of $\Omega_{\mathrm{2w}}^\circ$ gives
\begin{align*}
\beta^\top\Omega_{\mathrm{2w}}^\circ\beta
-V_{\mathrm{opt},\beta}^\circ
=
\frac{1}{I}\bigl(e_1^B S_{a,1}-e_0^B S_{a,0}\bigr)^2+\frac{1}{J}\bigl(e_1^S S_{b,1}-e_0^S S_{b,0}\bigr)^2
\ge0.
\end{align*}
This proves both inequalities in the proposition.
\end{proof}

\subsection{Proof of Corollary~\ref{cor:scalar-variance-exactness}}\label{app:proof-exactness}
\begin{proof}
For part (i), setting $(t_B,t_S)=(1,1)$ in
\eqref{eq:Vw-conservative-identity} gives
\(\beta^\top\Omega_{\mathrm{2w}}^\circ\beta-\sigma_\beta^2 = I^{-1} S_{e_1^Ba(1)+e_0^Ba(0)}^2 + J^{-1} S_{e_1^Sb(1)+e_0^Sb(0)}^2\).
Both terms on the right-hand side are nonnegative. Hence the gap is zero
if and only if both finite-population variances vanish. Since
\(e_1^Ba_i(1)+e_0^Ba_i(0) - \{e_1^B\bar a(1)+e_0^B\bar a(0)\} = e_1^B\widetilde a_i(1)+e_0^B\widetilde a_i(0)\),
and analogously on the seller side, this is equivalent to the conditions
in part (i).

For part (ii), write
\(S_{a,10} = (I-1)^{-1} \sum_{i=1}^I \widetilde a_i(1)\widetilde a_i(0)\), \(S_{b,10} = (J-1)^{-1} \sum_{j=1}^J \widetilde b_j(1)\widetilde b_j(0)\).
Using
\(S_{a,\Delta}^2
=
S_{a,1}^2+S_{a,0}^2-2S_{a,10},
\)
and its seller-side analogue, we obtain
\begin{align*}
V_{\mathrm{opt},\beta}^\circ-\sigma_\beta^2
=
\frac{2e_1^Be_0^B}{I}
\left(
S_{a,1}S_{a,0}+S_{a,10}
\right)+
\frac{2e_1^Se_0^S}{J}
\left(
S_{b,1}S_{b,0}+S_{b,10}
\right).
\end{align*}
By the Cauchy--Schwarz inequality,
\(
S_{a,10}\ge -S_{a,1}S_{a,0}
\) and \(
S_{b,10}\ge -S_{b,1}S_{b,0},
\)
so the gap is zero if and only if equality holds on both sides.

Consider the buyer side. If $S_{a,1}=0$, then
$\widetilde a_i(1)=0$ for all $i$, and we may take
$(c_1^B,c_0^B)=(1,0)$. If $S_{a,0}=0$, we may take
$(c_1^B,c_0^B)=(0,1)$. If both standard deviations are positive, equality
in Cauchy--Schwarz holds if and only if
\(\widetilde a_i(0) = -(S_{a,0}/S_{a,1})\widetilde a_i(1)\), for all \(i\).
Equivalently, taking
\(c_1^B = S_{a,0}/(S_{a,1}+S_{a,0})\), \(c_0^B = S_{a,1}/(S_{a,1}+S_{a,0})\),
gives
\(
c_1^B\widetilde a_i(1)+c_0^B\widetilde a_i(0)=0
\) for all $i$.
Conversely, the existence of nonnegative $c_1^B,c_0^B$ summing to one
and satisfying this identity implies either that one centered projection
is identically zero or that the two are negatively proportional, and hence
equality holds in the Cauchy--Schwarz bound.

The seller-side argument is identical. This proves part (ii).
\end{proof}

\subsection{Comparison with exposure-by-exposure bounds}
\label{app:compare-variance-bounds}

\subsubsection*{Proof of Proposition~\ref{prop:scalar-variance-comparison}}

\begin{proof}
Index the four exposure cells by $g=(p,q)$ and set
$w_{pq}=|\beta_{pq}|$ and $v_{pq}=\var\{\widehat Y(p,q)\}$. Direct expansion gives
\[
V_{\mathrm{Mas},\beta}^\circ-V_{\mathrm{Sud},\beta}^\circ
=
\frac12\sum_{g\ne h}w_gw_h(\sqrt{v_g}-\sqrt{v_h})^2
\ge0.
\]

For the first inequality, use $R_{pq}$ and $C_{pq}$ from
Lemma~\ref{lem:cell-variance-decomposition} and put
\(x_{pq} = \{e_{1-p}^B/(e_p^BI)\}^{1/2}R_{pq}\), \(y_{pq} = \{e_{1-q}^S/(e_q^SJ)\}^{1/2}C_{pq}\).
By \eqref{eq:cell-mean-variance-decomposition},
$x_{pq}^2+y_{pq}^2\le v_{pq}$. The triangle inequality for finite-population
standard deviations gives
\(\sqrt{e_1^Be_0^B I^{-1}}(S_{a,1}+S_{a,0}) \le \sum_{p,q}w_{pq}x_{pq}\),
and, analogously,
\(\sqrt{e_1^Se_0^S J^{-1}}(S_{b,1}+S_{b,0}) \le \sum_{p,q}w_{pq}y_{pq}\).
Therefore, Minkowski's inequality in $\mathbb R^2$ yields
\begin{align*}
V_{\mathrm{opt},\beta}^\circ
\le
\left(\sum_{p,q}w_{pq}x_{pq}\right)^2
+
\left(\sum_{p,q}w_{pq}y_{pq}\right)^2\le
\left\{
\sum_{p,q}w_{pq}\sqrt{x_{pq}^2+y_{pq}^2}
\right\}^2
\le
V_{\mathrm{Sud},\beta}^\circ.
\end{align*}
Combining the two inequalities proves the proposition.
\end{proof}

The second inequality is strict whenever two active cells have unequal $v_g$. The first is
strict whenever at least one of the triangle, Minkowski, or interaction-variance inequalities
used above is strict.

\section{Supporting sampling and regression results}\label{app:technical-identities}
\subsection{Regression reparameterization and residual identities}

Lemma~\ref{lem:ols-reparameterization} below shows that OLS coefficients and
cluster-robust covariance matrices transform equivariantly under an invertible
reparameterization.
\begin{lemma}
\label{lem:ols-reparameterization}
Let $R$ be an $n\times k$ full-column-rank design matrix and let
$H\in\mathbb R^{k\times k}$ be nonsingular. Set $\widetilde R=RH$.
Let $\widehat\beta$ and $\widehat{\widetilde\beta}$ be the OLS coefficient
vectors from regressing $Y$ on $R$ and $\widetilde R$, respectively. Then
\(
\widehat{\widetilde\beta}=H^{-1}\widehat\beta,
\)
and the two regressions have identical fitted values and residuals.

Let
\(\widehat M_R = \sum_{m=1}^M c_m \sum_{g\in\mathcal G_m} R_g^\top\widehat e_g\widehat e_g^\top R_g\),
where $\mathcal G_m$ are fixed collections of clusters and $c_m$ are fixed
scalars, and define the sandwich covariance estimator
\(
\widehat\Omega_R
=
(R^\top R)^{-1}\widehat M_R(R^\top R)^{-1}.
\)
Define $\widehat\Omega_{\widetilde R}$ analogously. Then
\(\widehat\Omega_{\widetilde R} = H^{-1}\widehat\Omega_RH^{-\top}\).
Consequently, for any fixed matrix $A$ with $k$ columns,
\(A\widehat\beta = AH\widehat{\widetilde\beta}\), \(A\widehat\Omega_RA^\top = AH\widehat\Omega_{\widetilde R}H^\top A^\top\).
\end{lemma}

\begin{proof}
Since $\widetilde R=RH$,
\(\widehat{\widetilde\beta} = (H^\top R^\top RH)^{-1}H^\top R^\top Y = H^{-1}\widehat\beta\).
Hence
$\widetilde R\widehat{\widetilde\beta}=R\widehat\beta$, so the two
regressions have identical fitted values and residuals. For every cluster $g$,
$\widetilde R_g^\top\widehat e_g
=H^\top R_g^\top\widehat e_g$, and therefore
$\widehat M_{\widetilde R}=H^\top\widehat M_RH$. Also,
\((\widetilde R^\top\widetilde R)^{-1} = H^{-1}(R^\top R)^{-1}H^{-\top}\).
Substitution gives
$\widehat\Omega_{\widetilde R}
=H^{-1}\widehat\Omega_RH^{-\top}$.
The final identities follow immediately.
\end{proof}

Corollary~\ref{prop:factor-reparameterization} below applies
Lemma~\ref{lem:ols-reparameterization} to show that the saturated
cell-indicator regression and the factor-based regression are equivalent reparameterizations of the same regression.
\begin{corollary}
\label{prop:factor-reparameterization}
Let $Z$ be the $IJ\times4$ design matrix of the saturated cell-indicator
regression in \eqref{eq:reg-Gij-sec31}, with columns ordered as
$(11),(10),(01),(00)$, and let $Z_{\mathrm{fac}}$ be the design matrix of
\(Y_{ij}^{\mathrm{obs}} = \gamma_0+\gamma_BW_i^B+\gamma_SW_j^S +\gamma_{BS}W_i^BW_j^S+u_{ij}\).
Define
\[
H=
\begin{pmatrix}
1&1&1&1\\
1&1&0&0\\
1&0&1&0\\
1&0&0&0
\end{pmatrix},
\qquad
\Lambda=H^{-1}
=
\begin{pmatrix}
0&0&0&1\\
0&1&0&-1\\
0&0&1&-1\\
1&-1&-1&1
\end{pmatrix}.
\]
Then $Z_{\mathrm{fac}}=ZH$. Hence, whenever all four exposure cells are nonempty,
\(\widehat{\boldsymbol\gamma} = \Lambda\widehat{\boldsymbol\mu}\), \(\widehat\Omega_{\gamma,\mathrm{2w}} = \Lambda\widehat\Omega_{\mathrm{2w}}\Lambda^\top\).
In particular,
\(
\widehat\gamma_0=\widehat Y(0,0)\), \(
\widehat\gamma_B=\widehat\tau_{B,\mathrm{spill}}\), \(
\widehat\gamma_S=\widehat\tau_{S,\mathrm{spill}}\), \(
\widehat\gamma_{BS}=\widehat\tau_{\mathrm{int}},
\)
and
\(
\widehat\tau_{\mathrm{tot}}
=
\widehat\gamma_B+\widehat\gamma_S+\widehat\gamma_{BS}.
\)
Moreover, for any fixed matrix $F$ with four columns,
\(F\widehat{\boldsymbol\mu} = FH\widehat{\boldsymbol\gamma}\), \(F\widehat\Omega_{\mathrm{2w}}F^\top = FH\widehat\Omega_{\gamma,\mathrm{2w}}H^\top F^\top\).
Thus the two parameterizations give identical fitted values, residuals,
and Wald inference.
\end{corollary}

\begin{proof}
For each pair $(i,j)$,
$(z_{ij}^{\mathrm{fac}})^\top=z_{ij}^\top H$, so $Z_{\mathrm{fac}}=ZH$.
The result follows from Lemma~\ref{lem:ols-reparameterization}.
\end{proof}

\subsection{Law of large numbers for exposure-cell averages}
Lemma~\ref{lem:cell-average-lln} below shows that, under the two-sided
randomization, the average of a fixed array within any exposure cell is close
to its finite-population average, with mean-square error of order
$\rho_{I,J}$.
\begin{lemma}\label{lem:cell-average-lln}
Let $\{A_{ij}\}_{i\in[I],j\in[J]}$ be a fixed finite-population array, where each
$A_{ij}$ is either a scalar, a fixed-dimensional vector, or a fixed-dimensional matrix. 
Define
\(\bar A_{i\cdot} = J^{-1}\sum_{j=1}^J A_{ij}\), \(\bar A_{\cdot j} = I^{-1}\sum_{i=1}^I A_{ij}\), \(\bar A_{\cdot\cdot} = (IJ)^{-1}\sum_{i=1}^I\sum_{j=1}^J A_{ij}\).
For each $(p,q)\in\{0,1\}^2$, define \(
\widehat A_{pq}
=
(I_pJ_q)^{-1}
\sum_{i=1}^I\sum_{j=1}^J
A_{ij}G_{pq,ij}.
\)
Suppose that, as $I,J\to\infty$,
\(
\min\{e_0^B,e_1^B,e_0^S,e_1^S\}\ge \underline e>0
\)
for some constant $\underline e$, and
\((IJ)^{-1}\sum_{i=1}^I\sum_{j=1}^J \|A_{ij}-\bar A_{\cdot\cdot}\|^2=O(1)\).
Then, for each fixed $(p,q)\in\{0,1\}^2$,
\(\E\| \widehat A_{pq}-\bar A_{\cdot\cdot} \|^2 = O(\rho_{I,J})\).
Consequently,
\(
\|
\widehat A_{pq}-\bar A_{\cdot\cdot}
\|
=
O_{\mathbb P}(\rho_{I,J}^{1/2}).
\)
\end{lemma}

\begin{proof}
We first prove the result for a scalar array. The vector and matrix cases follow by applying the same argument coordinatewise and using that the dimension is fixed.

Repeat the algebra of Lemma~\ref{lem:decomp-yhat} coordinatewise with the fixed array
$A_{ij}$ in place of $Y_{ij}(p,q)$. For each $(p,q)$,
\(\widehat A_{pq}-\bar A_{\cdot\cdot} = I_p^{-1}\sum_{i=1}^I \bar A_{i\cdot}U_{ip} + J_q^{-1}\sum_{j=1}^J \bar A_{\cdot j}V_{jq} + (I_pJ_q)^{-1}\sum_{i=1}^I\sum_{j=1}^J A_{ij}U_{ip}V_{jq}\).
Define $a_i^B=\bar A_{i\cdot}-\bar A_{\cdot\cdot}$, $a_j^S=\bar A_{\cdot j}-\bar A_{\cdot\cdot}$ and $\widetilde A_{ij}
=
A_{ij}-\bar A_{i\cdot}-\bar A_{\cdot j}+\bar A_{\cdot\cdot}$.
Since $\sum_i U_{ip}=0$ and $\sum_j V_{jq}=0$, we can rewrite this as
\(\widehat A_{pq}-\bar A_{\cdot\cdot} = B_{pq}+S_{pq}+R_{pq}\),
where
$
B_{pq}
=
{I_p^{-1}}\sum_{i=1}^I
a_i^B U_{ip}$, $
S_{pq}
=
{J_q^{-1}}\sum_{j=1}^J
a_j^S V_{jq},
$
and
\(R_{pq} = (I_pJ_q)^{-1} \sum_{i=1}^I\sum_{j=1}^J \widetilde A_{ij}U_{ip}V_{jq}\).

We bound the second moment of the three terms. For the buyer term, the complete-randomization
variance formula gives
\(\E B_{pq}^2 = \var(B_{pq}) \le C I^{-1} \cdot I^{-1}\sum_{i=1}^I (a_i^B)^2\),
where $C$ depends only on the lower bound $\underline e$. By Jensen's inequality,
\(I^{-1}\sum_{i=1}^I (a_i^B)^2 = I^{-1}\sum_{i=1}^I (\bar A_{i\cdot}-\bar A_{\cdot\cdot})^2 \le (IJ)^{-1}\sum_{i=1}^I\sum_{j=1}^J (A_{ij}-\bar A_{\cdot\cdot})^2 =O(1)\).
Hence
\(
\E B_{pq}^2=O(I^{-1}).
\)

Similarly,
\(\E S_{pq}^2 = \var(S_{pq}) \le C J^{-1} \cdot J^{-1}\sum_{j=1}^J (a_j^S)^2 = O(J^{-1})\).

For the interaction term, using the independence of the buyer and seller randomizations and
Lemma~\ref{lem:bilinear-remainder},
\(\E R_{pq}^2 = \var(R_{pq}) \le C (IJ)^{-1} \cdot (IJ)^{-1}\sum_{i=1}^I\sum_{j=1}^J \widetilde A_{ij}^2\).
Double centering does not increase the order of the second moment, so
\((IJ)^{-1}\sum_{i=1}^I\sum_{j=1}^J \widetilde A_{ij}^2 = O(1)\).
Therefore
\(
\E R_{pq}^2=O((IJ)^{-1}).
\)

Combining the three bounds,
\[
\E\left(\widehat A_{pq}-\bar A_{\cdot\cdot}\right)^2
\le
3\E B_{pq}^2+3\E S_{pq}^2+3\E R_{pq}^2
=
O(I^{-1}+J^{-1}+(IJ)^{-1})
=
O(\rho_{I,J}).
\]

For a vector or matrix array, vectorize $A_{ij}$ into a fixed-dimensional vector. The scalar
argument applies to each coordinate. Summing over the fixed number of coordinates gives
\(\E\| \widehat A_{pq}-\bar A_{\cdot\cdot} \|^2 = O(\rho_{I,J})\).
Markov's inequality then gives
\(
\|
\widehat A_{pq}-\bar A_{\cdot\cdot}
\|
=
O_{\mathbb P}(\rho_{I,J}^{1/2}).
\)
\end{proof}

\subsection{Consistency of estimated projection variances}\label{app:projection-variance-consistency}
\begin{lemma}\label{lemma:sample-var-consistency}
For each $N$, let $Y_1,\dots,Y_N$ be a finite population, with population mean
$\bar Y=N^{-1}\sum_{i=1}^N Y_i$ and finite-population variance
$S_Y^2=(N-1)^{-1}\sum_{i=1}^N (Y_i-\bar Y)^2$.
Suppose a simple random sample without replacement of size $n$ is drawn from this population, and let
$Z_i=\mathbf 1\{i\text{ is sampled}\}$,
$\hat{\bar Y}=n^{-1}\sum_{i=1}^N Z_iY_i$, and
$s_Y^2=(n-1)^{-1}\sum_{i=1}^N Z_i(Y_i-\hat{\bar Y})^2$
be the sample variance.
Assume that, as $N\to\infty$,
\(\min(n,N-n)/N\ge c>0\), \(N^{-1}\sum_{i=1}^N(Y_i-\bar Y)^4=O(1)\).
Then $s_Y^2-S_Y^2=o_{\mathbb P}(1)$.
\end{lemma}
\begin{proof}
Let $m_N=\max_i(Y_i-\bar Y)^2$. The fourth-moment condition implies
$m_N=O(N^{1/2})$ and $S_Y^2=O(1)$. Along any subsequence on which
$S_Y^2$ is bounded away from zero,
\(\{\min(n,N-n)\}^{-1}m_NS_Y^{-2} = O(N^{-1/2}) \to0\).
Proposition~1 of \citet{li2017general} therefore gives
$s_Y^2/S_Y^2\to_{\mathbb P}1$, and hence
$s_Y^2-S_Y^2=o_{\mathbb P}(1)$.
Along any subsequence on which $S_Y^2\to0$, unbiasedness,
$\E(s_Y^2)=S_Y^2$, and Markov's inequality imply
$s_Y^2=o_{\mathbb P}(1)$.
Every subsequence admits a further subsequence of one of these two types, which proves the result.
\end{proof}

Lemma~\ref{lem:estimated-projection-variances} below shows that the sample
variances of the estimated buyer- and seller-side projections consistently
estimate their finite-population counterparts. The same result holds for the
corresponding fixed-dimensional covariance matrices.
\begin{lemma}
\label{lem:estimated-projection-variances}
Under the conditions of Proposition~\ref{prop:twoway-vector-conservative}, for every
$p,q\in\{0,1\}$,
\(s_{a,p}^2-S_{a,p}^2=o_{\mathbb P}(1)\), \(s_{b,q}^2-S_{b,q}^2=o_{\mathbb P}(1)\).
The same conclusion holds in operator norm for the fixed-dimensional state-specific
projection covariance matrices used in the proof of
Proposition~\ref{prop:twoway-vector-conservative}.
\end{lemma}

\begin{proof}
We prove the buyer statement; the seller statement follows by interchanging the two sides.
For every buyer, including buyers not assigned to state $p$, define
\(\check Y_{i\cdot}(p,q) = J_q^{-1}\sum_{j:W_j^S=q}Y_{ij}(p,q)\), \(\check a_i(p) = (e_p^B)^{-1}\sum_{q=0}^1\beta_{pq}\check Y_{i\cdot}(p,q)\).
For a buyer with $W_i^B=p$, these quantities equal
$\widehat{\bar Y}_{i\cdot}(p,q)$ and $\widehat a_i(p)$, respectively. Let
$d_i(p)=\check a_i(p)-a_i(p)$. The complete-randomization variance formula on the seller
side and the overlap condition give
\(I^{-1}\sum_{i=1}^I\E_S\{d_i(p)^2\} \le C J^{-1}\sum_{q=0}^1I^{-1}\sum_{i=1}^I (J-1)^{-1}\sum_{j=1}^J \{Y_{ij}(p,q)-\bar Y_{i\cdot}(p,q)\}^2=O(J^{-1})\).
The last equality follows from the bounded second moments implied by the assumed fourth
moments. Because the buyer and seller randomizations are independent,
\(\E\{ I_p^{-1}\sum_{i:W_i^B=p}d_i(p)^2 \} =O(J^{-1})\),
and hence the average in braces is $o_{\mathbb P}(1)$.

Let $\widetilde s_{a,p}$ be the sample standard deviation of the fixed values
$\{a_i(p):W_i^B=p\}$. Jensen's inequality and the fourth-moment assumption give a uniform
centered fourth-moment bound for the finite population $\{a_i(p)\}_{i=1}^I$. Therefore,
Lemma~\ref{lemma:sample-var-consistency} yields
\(
\widetilde s_{a,p}^2-S_{a,p}^2=o_{\mathbb P}(1).
\)
The triangle inequality for the centered Euclidean norm gives
\begin{align*}
|s_{a,p}-\widetilde s_{a,p}|
\le
\left[
\frac{1}{I_p-1}\sum_{i:W_i^B=p}
\{d_i(p)-\bar d_p\}^2
\right]^{1/2}\le
\left[
\frac{1}{I_p-1}\sum_{i:W_i^B=p}d_i(p)^2
\right]^{1/2}
=o_{\mathbb P}(1),
\end{align*}
where $\bar d_p=I_p^{-1}\sum_{i:W_i^B=p}d_i(p)$. Both standard deviations are
$O_{\mathbb P}(1)$, so $s_{a,p}^2-S_{a,p}^2=o_{\mathbb P}(1)$.

For vector projections, apply the same argument to each coordinate. Consistency of each
cross-covariance follows by polarization, and fixed dimension converts entrywise convergence
to operator-norm convergence.
\end{proof}

\subsection{Exact saturated-regression cluster-score decomposition}\label{app:cluster-score-decomposition}
This subsection proves Lemma~\ref{lem:twoway-score-decomposition} in Section~\ref{app:proof-two-way}.

\begin{proof}
Recall
\(\widehat\Omega_{\mathrm{2w}} = (Z^\top Z)^{-1} \bigl( \Gamma_B+\Gamma_S-\Gamma_{BS} \bigr) (Z^\top Z)^{-1}\),
where
\(\Gamma_B=\sum_{i=1}^I Z_i^\top \widehat u_i\widehat u_i^\top Z_i\), \(\Gamma_S=\sum_{j=1}^J Z_{\cdot j}^\top \widehat u_{\cdot j}\widehat u_{\cdot j}^\top Z_{\cdot j}\), and \(\Gamma_{BS}=\sum_{i=1}^I\sum_{j=1}^J z_{ij}z_{ij}^\top \widehat u_{ij}^2\).
Then
\(\beta^\top\widehat\Omega_{\mathrm{2w}}\beta = \sum_{i=1}^I \widehat\psi_i^B(\beta)^2 + \sum_{j=1}^J \widehat\psi_j^S(\beta)^2 - \sum_{i=1}^I\sum_{j=1}^J \delta_{ij}(\beta)^2\),
where
\(\widehat\psi_i^B(\beta)=\beta^\top (Z^\top Z)^{-1}Z_i^\top \widehat u_i\), \(\widehat\psi_j^S(\beta)=\beta^\top (Z^\top Z)^{-1}Z_{\cdot j}^\top \widehat u_{\cdot j}\), and \(\delta_{ij}(\beta)=\beta^\top (Z^\top Z)^{-1}z_{ij}\widehat u_{ij}\).

We first compute the buyer-side contribution. Recall that the row regressor is ordered as
$z_{ij}=(G_{11,ij},G_{10,ij},G_{01,ij},G_{00,ij})^\top$,
and the contrast vector is ordered conformably as
$\beta=(\beta_{11},\beta_{10},\beta_{01},\beta_{00})^\top$.
Under this ordering,
\(Z^\top Z = \operatorname{diag} (I_1J_1,\ I_1J_0,\ I_0J_1,\ I_0J_0)\).

Fix a buyer \(i\) with \(W_i^B=p\). Since only the seller assignment varies within this buyer,
\(Z_i^\top\widehat u_i = \sum_{q=0}^1 \mathbf e_{pq} \sum_{j:W_j^S=q}\widehat u_{ij}\),
where \(\mathbf e_{pq}\) is the canonical basis vector associated with the cell \((p,q)\) in the
ordering
\((1,1),(1,0),(0,1),(0,0)\).
Hence
\(\widehat\psi_i^B(\beta) = \sum_{q=0}^1 \beta_{pq} (I_pJ_q)^{-1} \sum_{j:W_j^S=q}\widehat u_{ij}\).
By the saturated-regression identity in Section~\ref{subsec:plugin-reg-vector}, the fitted value in cell \((p,q)\) equals
\(\widehat Y(p,q)\). Therefore
$\widehat u_{ij}=Y_{ij}^{\mathrm{obs}}-\widehat Y(W_i^B,W_j^S)$.
It follows that, for \(W_i^B=p\),
\begin{align*}
\widehat\psi_i^B(\beta)
=
\sum_{q=0}^1
\frac{\beta_{pq}}{I_pJ_q}
\sum_{j:W_j^S=q}
\left\{
Y_{ij}^{\mathrm{obs}}-\widehat Y(p,q)
\right\} =
\frac{1}{I_p}
\sum_{q=0}^1
\beta_{pq}
\left\{
\widehat{\bar Y}_{i\cdot}(p,q)-\widehat Y(p,q)
\right\},
\end{align*}
where
$\widehat{\bar Y}_{i\cdot}(p,q)
=J_q^{-1}\sum_{j:W_j^S=q}Y_{ij}^{\mathrm{obs}}$.
Define
$\widehat a_i(p)
=(e_p^B)^{-1}\sum_{q=0}^1\beta_{pq}\widehat{\bar Y}_{i\cdot}(p,q)$
and
$\widehat{\bar a}(p)
=I_p^{-1}\sum_{i:W_i^B=p}\widehat a_i(p)$.
Since
$I_p^{-1}\sum_{i:W_i^B=p}\widehat{\bar Y}_{i\cdot}(p,q)
=\widehat Y(p,q)$,
we have
$\widehat{\bar a}(p)
=(e_p^B)^{-1}\sum_{q=0}^1\beta_{pq}\widehat Y(p,q)$.
Using \(e_p^B=I_p/I\), we obtain
\(\widehat\psi_i^B(\beta) = I^{-1} \{ \widehat a_i(p)-\widehat{\bar a}(p) \}\).
Therefore,
\begin{align*}
\sum_{i=1}^I \widehat\psi_i^B(\beta)^2
=
\sum_{p=0}^1
\sum_{i:W_i^B=p}
\frac{1}{I^2}
\left\{
\widehat a_i(p)-\widehat{\bar a}(p)
\right\}^2 =
\frac{I_1-1}{I^2}s_{a,1}^2
+
\frac{I_0-1}{I^2}s_{a,0}^2
=
\widehat V_{B,\mathrm{2w}}(\beta),
\end{align*}
where
$s_{a,p}^2
=(I_p-1)^{-1}
\sum_{i:W_i^B=p}
\{\widehat a_i(p)-\widehat{\bar a}(p)\}^2$.

The seller-side contribution is analogous. Fix a seller \(j\) with \(W_j^S=q\). Then
\(\widehat\psi_j^S(\beta) = \sum_{p=0}^1 \beta_{pq} (I_pJ_q)^{-1} \sum_{i:W_i^B=p} \widehat u_{ij}\).
Using the residual identity again,
\(\widehat\psi_j^S(\beta) = J_q^{-1} \sum_{p=0}^1 \beta_{pq} \{ \widehat{\bar Y}_{\cdot j}(p,q)-\widehat Y(p,q) \}\),
where
$\widehat{\bar Y}_{\cdot j}(p,q)
=I_p^{-1}\sum_{i:W_i^B=p}Y_{ij}^{\mathrm{obs}}$.
Define
$\widehat b_j(q)
=(e_q^S)^{-1}\sum_{p=0}^1\beta_{pq}\widehat{\bar Y}_{\cdot j}(p,q)$
and
$\widehat{\bar b}(q)
=J_q^{-1}\sum_{j:W_j^S=q}\widehat b_j(q)$.
Since
$J_q^{-1}\sum_{j:W_j^S=q}\widehat{\bar Y}_{\cdot j}(p,q)
=\widehat Y(p,q)$,
we have
$\widehat{\bar b}(q)
=(e_q^S)^{-1}\sum_{p=0}^1\beta_{pq}\widehat Y(p,q)$.
Using \(e_q^S=J_q/J\), we obtain
\(\widehat\psi_j^S(\beta) = J^{-1} \{ \widehat b_j(q)-\widehat{\bar b}(q) \}\).
Therefore,
\begin{align*}
\sum_{j=1}^J \widehat\psi_j^S(\beta)^2
=
\sum_{q=0}^1
\sum_{j:W_j^S=q}
\frac{1}{J^2}
\left\{
\widehat b_j(q)-\widehat{\bar b}(q)
\right\}^2 =
\frac{J_1-1}{J^2}s_{b,1}^2
+
\frac{J_0-1}{J^2}s_{b,0}^2
=
\widehat V_{S,\mathrm{2w}}(\beta),
\end{align*}
where
$s_{b,q}^2
=(J_q-1)^{-1}
\sum_{j:W_j^S=q}
\{\widehat b_j(q)-\widehat{\bar b}(q)\}^2$.

Finally, consider the intersection correction. If
\((W_i^B,W_j^S)=(p,q)\), then \(z_{ij}=\mathbf e_{pq}\) and hence
\(\delta_{ij}(\beta) = \beta_{pq} (I_pJ_q)^{-1}\widehat u_{ij} = \beta_{pq} (I_pJ_q)^{-1} \{ Y_{ij}^{\mathrm{obs}}-\widehat Y(p,q) \}\).
Therefore,
\begin{align*}
\sum_{i=1}^I\sum_{j=1}^J \delta_{ij}(\beta)^2
=
\sum_{p=0}^1\sum_{q=0}^1
\frac{\beta_{pq}^2}{I_p^2J_q^2}
\sum_{i:W_i^B=p}\sum_{j:W_j^S=q}
\left\{
Y_{ij}^{\mathrm{obs}}-\widehat Y(p,q)
\right\}^2 =
\widehat V_{BS,\mathrm{2w}}(\beta).
\end{align*}
Combining the buyer-side, seller-side, and intersection terms gives
\(\beta^\top\widehat\Omega_{\mathrm{2w}}\beta = \widehat V_{B,\mathrm{2w}}(\beta) + \widehat V_{S,\mathrm{2w}}(\beta) - \widehat V_{BS,\mathrm{2w}}(\beta)\),
as claimed. Applying the scalar identities to the coordinate vectors and their pairwise sums
gives the corresponding matrix identity by polarization.
\end{proof}
\subsection{Cluster-score implementation}
\label{app:cluster-score-implementation}
\label{subsubsec:cluster-score-aggregation}

The optimized variance estimator in \eqref{eq:Vopt-sample} can be implemented by a
minor adjustment to the usual two-way clustered variance estimator. The point estimator and
the saturated regression remain unchanged. The only change is how we aggregate the buyer-
and seller-cluster scores in the variance estimator.

For $i\in[I]$, $j\in [J]$, and a fixed scalar contrast \(\beta^\top\widehat{\mathbf Y}\), define the buyer- and seller-cluster
scores from the saturated regression \eqref{eq:reg-Gij-sec31} by
$
\widehat\psi_i^B(\beta)
=
\beta^\top (Z^\top Z)^{-1}Z_i^\top \widehat u_i,
$ and $
\widehat\psi_j^S(\beta)
=
\beta^\top (Z^\top Z)^{-1}Z_{\cdot j}^\top \widehat u_{\cdot j}.$ 
By Lemma~\ref{lem:twoway-score-decomposition} and \eqref{eq:intersection-negligible}, we have
\begin{align}\label{eq:beta2w_psiIJ}
    \beta^\top\widehat \Omega_{\mathrm{2w}}
\beta =
\sum_{i=1}^I
\{\widehat\psi_i^B(\beta)\}^2
+
\sum_{j=1}^J
\{\widehat\psi_j^S(\beta)\}^2+o_{\mathbb P}(I^{-1}+J^{-1}).
\end{align}
We define the first two terms on the right-hand side of \eqref{eq:beta2w_psiIJ} as $\widehat V_{\mathrm{add}}(\beta)$. 

Now separate these cluster-score contributions by exposure state. For $p,q\in\{0,1\},$ define
$
Q_{B,p}(\beta)
=
\sum_{i:W_i^B=p}
\{\widehat\psi_i^B(\beta)\}^2,
$ and $
Q_{S,q}(\beta)
=
\sum_{j:W_j^S=q}
\{\widehat\psi_j^S(\beta)\}^2,
$
then
\begin{align}\label{eq:add_Q}
\widehat V_{\mathrm{add}}(\beta)
=
Q_{B,1}(\beta)+Q_{B,0}(\beta)
+
Q_{S,1}(\beta)+Q_{S,0}(\beta)=\beta^\top\widehat\Omega_{\mathrm{2w}}\beta
+
o_{\mathbb P}(\rho_{I,J}).
\end{align}
Our optimized variance estimator is obtained by replacing this unit-weight aggregation with a
state-specific norm aggregation:
\begin{align}\label{eq:Vscore-sec33}
\widehat V_{\mathrm{score}}(\beta)
=
\left\{
\sqrt{e_0^BQ_{B,1}(\beta)}
+
\sqrt{e_1^BQ_{B,0}(\beta)}
\right\}^2
+
\left\{
\sqrt{e_0^SQ_{S,1}(\beta)}
+
\sqrt{e_1^SQ_{S,0}(\beta)}
\right\}^2 .
\end{align}

\begin{prop}[Cluster-score aggregation representation of \(\widehat V_{\mathrm{opt},\beta}\)]
\label{prop:cluster-score-aggregation-vopt}
Suppose the conditions of Proposition \ref{prop:twoway-vector-conservative} hold. For any fixed scalar contrast \(\beta^\top\widehat{\mathbf Y}\), the score-aggregated variance satisfies
\[
\widehat V_{\mathrm{score}}(\beta)
=
\widehat V_{\mathrm{opt},\beta}
+
O_{\mathbb P}(I^{-2}+J^{-2}).
\]
\end{prop}

Proposition~\ref{prop:cluster-score-aggregation-vopt} gives a direct
regression-style implementation of the optimized variance estimator. After fitting the saturated
regression, the point estimate, residuals, and cluster scores are unchanged. We only keep the
buyer and seller score contributions separately by exposure state and replace the unit-weight
sum in \eqref{eq:add_Q} with the aggregation rule in \eqref{eq:Vscore-sec33}. Thus,
\(\widehat V_{\mathrm{score}}(\beta)\) is an adjusted two-way clustered variance estimator:
it uses the usual cluster scores, but changes the final aggregation rule.

We can rewrite \(\widehat V_{\mathrm{score}}(\beta)\) as
\[
\widehat V_{\mathrm{score}}(\beta)
=
e_1^Be_0^B
\left\{
\sqrt{\frac{Q_{B,1}(\beta)}{e_1^B}}
+
\sqrt{\frac{Q_{B,0}(\beta)}{e_0^B}}
\right\}^2
+
e_1^Se_0^S
\left\{
\sqrt{\frac{Q_{S,1}(\beta)}{e_1^S}}
+
\sqrt{\frac{Q_{S,0}(\beta)}{e_0^S}}
\right\}^2 .
\]

Thus, the optimized estimator uses the usual buyer- and seller-cluster
scores, but rescales and combines the two assignment states within each
side before squaring.

\subsubsection*{Proof of Proposition~\ref{prop:cluster-score-aggregation-vopt}}

\begin{proof}
Lemma~\ref{lem:twoway-score-decomposition} gives
\(Q_{B,p}(\beta)=I^{-2}(I_p-1)s_{a,p}^2\), \(Q_{S,q}(\beta)=J^{-2}(J_q-1)s_{b,q}^2\).
Lemma~\ref{lem:estimated-projection-variances} implies
$s_{a,p}=O_{\mathbb P}(1)$ and $s_{b,q}=O_{\mathbb P}(1)$. Since
$I_p=e_p^BI$ and $J_q=e_q^SJ$,
\begin{align*}
&\left\{
\sqrt{e_0^BQ_{B,1}(\beta)}
+
\sqrt{e_1^BQ_{B,0}(\beta)}
\right\}^2\\
&\qquad=
\frac{e_1^Be_0^B}{I}
\left\{
\sqrt{1-I_1^{-1}}\,s_{a,1}
+
\sqrt{1-I_0^{-1}}\,s_{a,0}
\right\}^2\\
&\qquad=
\frac{e_1^Be_0^B}{I}(s_{a,1}+s_{a,0})^2
+
O_{\mathbb P}(I^{-2}).
\end{align*}
The analogous seller contribution equals
\(e_1^Se_0^S J^{-1}(s_{b,1}+s_{b,0})^2 + O_{\mathbb P}(J^{-2})\).
Combining these two expressions and using \eqref{eq:Vopt-sample} proves
\(
\widehat V_{\mathrm{score}}(\beta)
=
\widehat V_{\mathrm{opt},\beta}
+
O_{\mathbb P}(I^{-2}+J^{-2}).
\)
\end{proof}

\section{Regression adjustment}
\label{app:regression-adjustment-proofs}

\subsection{Asymptotic equivalence of covariance estimators}
\label{app:oracle-sandwich-equivalence}

For a fixed-dimensional regressor $r_{ij}$ and a scalar array $v_{ij}$, define
\(\Gamma_{R,B}(v) = \sum_{i=1}^I (\sum_{j=1}^J r_{ij}v_{ij}) (\sum_{j=1}^J r_{ij}v_{ij})^\top\), \(\Gamma_{R,S}(v) = \sum_{j=1}^J (\sum_{i=1}^I r_{ij}v_{ij}) (\sum_{i=1}^I r_{ij}v_{ij})^\top\), \(\Gamma_{R,BS}(v) = \sum_{i=1}^I\sum_{j=1}^J r_{ij}r_{ij}^\top v_{ij}^2\),
and let
$\Gamma_R(v)=\Gamma_{R,B}(v)+\Gamma_{R,S}(v)-\Gamma_{R,BS}(v)$.
Define $\Gamma_Z(v)$ analogously with $z_{ij}$ in place of $r_{ij}$.

The following lemma shows that the covariance matrix from the augmented
regression with estimated residuals is asymptotically equivalent to its
oracle counterpart based on the exposure indicators alone.
\begin{lemma}
\label{lem:oracle-sandwich-equivalence}
Let $R=(Z,C)$ be an
$(IJ)\times(4+k)$ design matrix, where $k$ is fixed and the row of $R$
corresponding to pair $(i,j)$ is $r_{ij}^\top=(z_{ij}^\top,c_{ij}^\top)$.
Put $H_\mu=(\mathrm I_4,0_{4\times k})$ and define
\(D=(IJ)^{-1}Z^\top Z\), \(M=(IJ)^{-1}Z^\top C\), \(Q=(IJ)^{-1}C^\top C\).
Suppose
\(
\p\{\lambda_{\min}(D)\wedge\lambda_{\min}(Q)\ge c\}\to1
\)
for some $c>0$,
\(
\|D\|_{\mathrm{op}}+\|Q\|_{\mathrm{op}}=O_{\mathbb P}(1),
\) and \(
\|M\|_{\mathrm{op}}=O_{\mathbb P}(\rho_{I,J}^{1/2}).
\)
Let $u_{ij}^\star$ and $\delta_{ij}$ be scalar arrays such that
$\widehat u_{ij}=u_{ij}^\star-\delta_{ij}$ and
\((IJ)^{-1}\sum_{i,j}\|r_{ij}\|^4=O_{\mathbb P}(1)\), \((IJ)^{-1}\sum_{i,j}|u_{ij}^\star|^4=O_{\mathbb P}(1)\), \((IJ)^{-1}\sum_{i,j}|\delta_{ij}|^4=o_{\mathbb P}(1)\).
Then
\begin{align*}
H_\mu(R^\top R)^{-1}\Gamma_R(\widehat u)
(R^\top R)^{-1}H_\mu^\top=
(Z^\top Z)^{-1}\Gamma_Z(u^\star)(Z^\top Z)^{-1}
+o_{\mathbb P}(\rho_{I,J})
\end{align*}
in operator norm.
\end{lemma}

\begin{proof}
Write
\(
H_\mu(R^\top R)^{-1}=(A,E),
\)
where $A$ and $E$ are the blocks corresponding to $Z$ and $C$.
The partitioned inverse formula and the assumptions on
$(D,M,Q)$ give
\begin{equation}\label{eq:generic-bread-bounds}
\begin{aligned}
\|A-(Z^\top Z)^{-1}\|_{\mathrm{op}}&=O_{\mathbb P}\left(\frac{\rho_{I,J}}{IJ}\right),
\\
\|E\|_{\mathrm{op}}&=O_{\mathbb P}\left(\frac{\rho_{I,J}^{1/2}}{IJ}\right),
\\
\|(R^\top R)^{-1}\|_{\mathrm{op}}&=O_{\mathbb P}((IJ)^{-1}).
\end{aligned}
\end{equation}
Indeed, the Schur complement satisfies
$D-MQ^{-1}M^\top=D+O_{\mathbb P}(\rho_{I,J})$ and remains nonsingular
with probability approaching one.

We next bound the three clustering components separately. For any scalar array
$v=(v_{ij})$, Cauchy--Schwarz gives
\(\sum_{i=1}^I \|\sum_{j=1}^Jr_{ij}v_{ij}\|^2 \le J\sum_{i,j}\|r_{ij}\|^2v_{ij}^2\), \(\sum_{j=1}^J \|\sum_{i=1}^Ir_{ij}v_{ij}\|^2 \le I\sum_{i,j}\|r_{ij}\|^2v_{ij}^2\).
For $v=u^\star$,
\(\sum_{i,j}\|r_{ij}\|^2|u_{ij}^\star|^2 \le (\sum_{i,j}\|r_{ij}\|^4)^{1/2} (\sum_{i,j}|u_{ij}^\star|^4)^{1/2} =O_{\mathbb P}(IJ)\).
Thus the buyer, seller, and intersection components based on $u^\star$ are,
respectively, $O_{\mathbb P}(IJ^2)$, $O_{\mathbb P}(I^2J)$, and $O_{\mathbb P}(IJ)$; their sum is
$O_{\mathbb P}(I^2J^2\rho_{I,J})$. Replacing $u^\star$ by $\delta$ in the same argument gives
$o_{\mathbb P}(IJ^2)$, $o_{\mathbb P}(I^2J)$, and $o_{\mathbb P}(IJ)$.

For two arrays $v,w$, let $\Gamma_{R,B}(v,w)$ denote the buyer component with
outer product
$(\sum_jr_{ij}v_{ij})(\sum_jr_{ij}w_{ij})^\top$, and define the seller and
intersection analogues. For each clustering component,
\(\|\Gamma_{R,B}(v,w)\|_{\mathrm{op}} \le \{\sum_i\|\sum_jr_{ij}v_{ij}\|^2\}^{1/2} \{\sum_i\|\sum_jr_{ij}w_{ij}\|^2\}^{1/2}\),
with analogous bounds for sellers and intersections. Applying these inequalities
to $(v,w)=(u^\star,\delta)$ and $(\delta,u^\star)$ shows that the buyer,
seller, and intersection cross components are respectively
$o_{\mathbb P}(IJ^2)$, $o_{\mathbb P}(I^2J)$, and $o_{\mathbb P}(IJ)$. Since
\(\Gamma_R(u^\star-\delta)-\Gamma_R(u^\star) = \Gamma_R(\delta)-\Gamma_R(u^\star,\delta) -\Gamma_R(\delta,u^\star)\),
where every term follows the same buyer-plus-seller-minus-intersection rule,
we obtain
\begin{equation}\label{eq:generic-residual-replacement}
\|\Gamma_R(\widehat u)-\Gamma_R(u^\star)\|_{\mathrm{op}}
=o_{\mathbb P}(I^2J^2\rho_{I,J}).
\end{equation}
Together with \eqref{eq:generic-bread-bounds}, this implies that replacing
$\widehat u$ by $u^\star$ in the full sandwich changes its dummy block by
$o_{\mathbb P}(\rho_{I,J})$.

Finally, partition the oracle middle matrix as
\[
\Gamma_R(u^\star)
=
\begin{pmatrix}
\Gamma_Z(u^\star)&\Gamma_{ZC}(u^\star)\\
\Gamma_{CZ}(u^\star)&\Gamma_C(u^\star)
\end{pmatrix}.
\]
The preceding componentwise bounds imply that every block is
$O_{\mathbb P}(I^2J^2\rho_{I,J})$. Expanding the dummy block gives
\begin{align*}
(A,E)\Gamma_R(u^\star)(A,E)^\top
=
A\Gamma_Z(u^\star)A^\top
+A\Gamma_{ZC}(u^\star)E^\top+E\Gamma_{CZ}(u^\star)A^\top
+E\Gamma_C(u^\star)E^\top.
\end{align*}
By \eqref{eq:generic-bread-bounds}, the last three terms are
$O_{\mathbb P}(\rho_{I,J}^{3/2})+O_{\mathbb P}(\rho_{I,J}^2)=o_{\mathbb P}(\rho_{I,J})$. The first term differs from
$(Z^\top Z)^{-1}\Gamma_Z(u^\star)(Z^\top Z)^{-1}$ by
$O_{\mathbb P}(\rho_{I,J}^2)=o_{\mathbb P}(\rho_{I,J})$. Combining this with
\eqref{eq:generic-residual-replacement} proves the result.
\end{proof}

\subsection{Two-way cluster-robust covariance from Fisher's regression}
\label{app:fisher-sandwich-definition}

To define the two-way clustered variance of Fisher's regression, we introduce the OLS residual from \eqref{eq:reg-fisher-adjustment} 
\(\widehat u_{ij}^{\mathrm F} = Y_{ij}^{\mathrm{obs}} - \sum_{p=0}^1\sum_{q=0}^1 G_{pq,ij}\widehat\mu_{\mathrm F,pq} - (X_{ij}-\bar X)^\top\widehat\theta_{\mathrm F}\).
Let \(r_{ij}^{\mathrm F}=(G_{11,ij},G_{10,ij},G_{01,ij},G_{00,ij},\dot X_{ij}^\top)^\top\) and
\(R_{\mathrm F}\) be the \(IJ\times(4+d)\) matrix with row
\((r_{ij}^{\mathrm F})^\top\).
Define
\(\Gamma_B^{\mathrm F} = \sum_{i=1}^I (R_{\mathrm F,i})^\top \widehat u_i^{\mathrm F}(\widehat u_i^{\mathrm F})^\top R_{\mathrm F,i}\), \(\Gamma_S^{\mathrm F} = \sum_{j=1}^J (R_{\mathrm F,\cdot j})^\top \widehat u_{\cdot j}^{\mathrm F}(\widehat u_{\cdot j}^{\mathrm F})^\top R_{\mathrm F,\cdot j}\), \(\Gamma_{BS}^{\mathrm F} = \sum_{i=1}^I\sum_{j=1}^J r_{ij}^{\mathrm F}(r_{ij}^{\mathrm F})^\top (\widehat u_{ij}^{\mathrm F})^2\).
The two-way clustered covariance matrix for all coefficients in Fisher's regression is
\[
\widehat V_{\eta,\mathrm{2w}}^{\mathrm F}
=
(R_{\mathrm F}^\top R_{\mathrm F})^{-1}
\left(
\Gamma_B^{\mathrm F}
+
\Gamma_S^{\mathrm F}
-
\Gamma_{BS}^{\mathrm F}
\right)
(R_{\mathrm F}^\top R_{\mathrm F})^{-1}.
\]
Let
\(
H_\mu=(\mathrm I_4,0_{4\times d}).
\)
The two-way clustered covariance matrix for the Fisher-adjusted exposure-mean estimator
$\widehat{\boldsymbol\mu}_{\mathrm F}$ is the dummy-coefficient block
\(
\widehat\Omega_{\mathrm{2w}}^{\mathrm F}
=
H_\mu
\widehat V_{\eta,\mathrm{2w}}^{\mathrm F}
H_\mu^\top .
\)

\subsection{Proof of Theorem~\ref{thm:fisher-adjustment}}

\begin{proof}
We prove the result in three steps.

\paragraph*{Step 1: consistency of the Fisher slope.}
By the Frisch--Waugh--Lovell representation of \eqref{eq:reg-fisher-adjustment}, we have 
\(\widehat\theta_{\mathrm F}
=
\widehat Q_{\mathrm F}^{-1}\widehat q_{\mathrm F},
\) with probability approaching one, where
\begin{align}
\widehat Q_{\mathrm F}
&=
\frac{1}{IJ}
\sum_{p=0}^1\sum_{q=0}^1
\sum_{i=1}^I\sum_{j=1}^J
G_{pq,ij}
\{X_{ij}-\widehat X(p,q)\}
\{X_{ij}-\widehat X(p,q)\}^\top,
\label{eq:Q-fisher-hat}
\\
\widehat q_{\mathrm F}
&=
\frac{1}{IJ}
\sum_{p=0}^1\sum_{q=0}^1
\sum_{i=1}^I\sum_{j=1}^J
G_{pq,ij}
\{X_{ij}-\widehat X(p,q)\}
\{Y_{ij}^{\mathrm{obs}}-\widehat Y(p,q)\}.
\label{eq:q-fisher-hat}
\end{align}
We first show
\(\widehat Q_{\mathrm F}=Q_{X}^{\circ}+o_{\mathbb P}(1)\), \(\widehat q_{\mathrm F}=q_{\mathrm F}^{\circ}+o_{\mathbb P}(1)\).
Apply Lemma~\ref{lem:cell-average-lln} to the fixed arrays
$X_{ij}$, $\dot X_{ij}\dot X_{ij}^\top$, $Y_{ij}(p,q)$, and
$\dot X_{ij}\dot Y_{ij}(p,q)$. Condition~\ref{cond:covariate-regularity} gives the required
second-moment bounds, so for each fixed $(p,q)$,
\(\widehat X(p,q)-\bar X=O_{\mathbb P}(\rho_{I,J}^{1/2})\), \(\widehat Y(p,q)-\bar Y(p,q)=O_{\mathbb P}(\rho_{I,J}^{1/2})\),
and the cell averages of $\dot X_{ij}\dot X_{ij}^\top$ and
$\dot X_{ij}\dot Y_{ij}(p,q)$ converge to their corresponding finite-population averages.
Using the identity
\begin{align*}
&\frac{1}{IJ}\sum_{i,j}G_{pq,ij}
\{X_{ij}-\widehat X(p,q)\}
\{X_{ij}-\widehat X(p,q)\}^\top\\
&\quad=
\frac{1}{IJ}\sum_{i,j}G_{pq,ij}\dot X_{ij}\dot X_{ij}^\top
-
\frac{I_pJ_q}{IJ}
\{\widehat X(p,q)-\bar X\}
\{\widehat X(p,q)-\bar X\}^\top,
\end{align*}
we obtain $\widehat Q_{\mathrm F}=Q_{X}^{\circ}+o_{\mathbb P}(1)$ after summing over the four exposure cells. Similarly,
\begin{align*}   
&\frac{1}{IJ}\sum_{i,j}G_{pq,ij}
\{X_{ij}-\widehat X(p,q)\}
\{Y_{ij}^{\mathrm{obs}}-\widehat Y(p,q)\}\\
&=
\frac{1}{IJ}\sum_{i,j}G_{pq,ij}\dot X_{ij}\dot Y_{ij}(p,q)
-
\frac{I_pJ_q}{IJ}
\{\widehat X(p,q)-\bar X\}
\{\widehat Y(p,q)-\bar Y(p,q)\},
\end{align*}
and therefore $\widehat q_{\mathrm F}=q_{\mathrm F}^{\circ}+o_{\mathbb P}(1)$.
Since $\lambda_{\min}(Q_{X}^{\circ})\ge c_Q>0$, we have
$\widehat Q_{\mathrm F}^{-1}=(Q_{X}^{\circ})^{-1}+o_{\mathbb P}(1)$. Hence
\(\widehat\theta_{\mathrm F} = \widehat Q_{\mathrm F}^{-1}\widehat q_{\mathrm F} = (Q_{X}^{\circ})^{-1}q_{\mathrm F}^{\circ}+o_{\mathbb P}(1) = \theta_{\mathrm F}^\star+o_{\mathbb P}(1)\).

\paragraph*{Step 2: oracle residualized exposure means.}
Define the oracle residualized observed cell mean
\(\widehat Y^{\mathrm F}(p,q) = (I_pJ_q)^{-1} \sum_{i=1}^I\sum_{j=1}^J \{Y_{ij}^{\mathrm{obs}}-\dot X_{ij}^\top\theta_{\mathrm F}^\star\} G_{pq,ij}\),
and collect these four means in
\(\widehat{\mathbf Y}^{\mathrm F} = \bigl( \widehat Y^{\mathrm F}(1,1), \widehat Y^{\mathrm F}(1,0), \widehat Y^{\mathrm F}(0,1), \widehat Y^{\mathrm F}(0,0) \bigr)^\top\).
Because the residualized potential outcomes have finite-population means $\bar Y(p,q)$,
Condition~\ref{cond:fisher-adjustment} and Theorem~\ref{thm:vector-clt} imply
\begin{equation}\label{eq:oracle-fisher-vector-clt}
\Sigma_{\mathrm F}^{-1/2}
\left(
\widehat{\mathbf Y}^{\mathrm F}
-
\bar{\mathbf Y}
\right)
\Rightarrow
N(0,\mathrm I_4).
\end{equation}

\paragraph*{Step 3: replacing the oracle slope by the estimated slope.}
Combining \eqref{eq:mu-fisher-cell-representation} with the definition of
$\widehat Y^{\mathrm F}(p,q)$ gives
\(\widehat\mu_{\mathrm F,pq} - \widehat Y^{\mathrm F}(p,q) = \{\widehat X(p,q)-\bar X\}^\top (\theta_{\mathrm F}^\star-\widehat\theta_{\mathrm F})\).
Therefore
\[
\left\|
\widehat{\boldsymbol\mu}_{\mathrm F}-\widehat{\mathbf Y}^{\mathrm F}
\right\|
\le
2
\max_{p,q}\|\widehat X(p,q)-\bar X\|
\cdot
\|\widehat\theta_{\mathrm F}-\theta_{\mathrm F}^\star\|
=
O_{\mathbb P}(\rho_{I,J}^{1/2})o_{\mathbb P}(1)
=
o_{\mathbb P}(\rho_{I,J}^{1/2}).
\]
Condition~\ref{cond:fisher-adjustment} includes the nondegeneracy condition in
Condition~\ref{cond:vector-clt}, so
$\lambda_{\min}(\Sigma_{\mathrm F})\ge c\rho_{I,J}$ for all sufficiently large $I,J$ and
$\|\Sigma_{\mathrm F}^{-1/2}\|_{\mathrm{op}}=O(\rho_{I,J}^{-1/2})$. Hence
\(\Sigma_{\mathrm F}^{-1/2} ( \widehat{\boldsymbol\mu}_{\mathrm F}-\widehat{\mathbf Y}^{\mathrm F} ) =o_{\mathbb P}(1)\).
The conclusion follows from \eqref{eq:oracle-fisher-vector-clt} and Slutsky's theorem.
\end{proof}

\subsection{Proof of Theorem~\ref{thm:fisher-2w-conservative}}

\begin{proof}
 Define the oracle residualized observed
outcome
\(Y_{ij}^{\mathrm F,\mathrm{obs},\star} = Y_{ij}^{\mathrm{obs}}-\dot X_{ij}^\top\theta_{\mathrm F}^\star\),
its cell mean $\widehat Y^{\mathrm F,\star}(p,q)$, and the corresponding
saturated-regression residual
\(\widehat u_{ij}^{\mathrm F,\star} = Y_{ij}^{\mathrm F,\mathrm{obs},\star} - \widehat Y^{\mathrm F,\star}(p,q)\), \(G_{pq,ij}=1\).
Let $\widehat\Omega_{\mathrm{2w}}(Y^{\mathrm F,\star})$ be the two-way covariance from regressing $Y^{\mathrm F,\mathrm{obs},\star}$ on $Z$.
Conditions~\ref{cond:fisher-adjustment} and~\ref{cond:fisher-variance}
allow Proposition~\ref{prop:twoway-vector-conservative} to be applied to the
residualized potential outcomes. Hence
\begin{equation}\label{eq:oracle-fisher-2w-conservative-proof}
\begin{aligned}
\widehat\Omega_{\mathrm{2w}}(Y^{\mathrm F,\star})-\Sigma_{\mathrm F}
&=
\frac{1}{I}\mathcal S_I\left(
e_1^BA^{\mathrm F,(1)}+e_0^BA^{\mathrm F,(0)}
\right)
\\ &\quad +
\frac{1}{J}\mathcal S_J\left(
e_1^SB^{\mathrm F,(1)}+e_0^SB^{\mathrm F,(0)}
\right)
+o_{\mathbb P}(\rho_{I,J}).
\end{aligned}
\end{equation}

It remains to compare this oracle covariance with the dummy-coefficient block
from Fisher's regression. Apply Lemma~\ref{lem:oracle-sandwich-equivalence}
with $C=\dot X$, $R=R_{\mathrm F}$, and
$u^\star=\widehat u^{\mathrm F,\star}$. The overlap condition gives the
required bounds for $D=(IJ)^{-1}Z^\top Z$, while
\(Q=(IJ)^{-1}\dot X^\top\dot X=Q_{X}^{\circ}\), \((IJ)^{-1}Z^\top\dot X=O_{\mathbb P}(\rho_{I,J}^{1/2})\)
by Condition~\ref{cond:covariate-regularity} and
Lemma~\ref{lem:cell-average-lln}. The same condition implies
\begin{equation}\label{eq:covariate-fourth-moment}
\frac1{IJ}\sum_{i,j}\|\dot X_{ij}\|^4
\le
\frac2{IJ}\sum_{i,j}
\|\dot X_{ij}\dot X_{ij}^\top-Q_{X}^{\circ}\|_{\mathrm F}^2
+2\|Q_{X}^{\circ}\|_{\mathrm F}^2
=O(1),
\end{equation}
where
$\|Q_{X}^{\circ}\|_{\mathrm F}
\le (IJ)^{-1}\sum_{i,j}\|\dot X_{ij}\|^2=O(1)$.
Thus $(IJ)^{-1}\sum_{i,j}\|r_{ij}^{\mathrm F}\|^4=O_{\mathbb P}(1)$.
Moreover, within each exposure cell,
\(\sum G_{pq,ij}|\widehat u_{ij}^{\mathrm F,\star}|^4 \le 16\sum G_{pq,ij} |Y_{ij}^{\mathrm F}(p,q)-\bar Y(p,q)|^4\).
Condition~\ref{cond:fisher-variance}, fixed dimension, and Markov's inequality
therefore give
$(IJ)^{-1}\sum_{i,j}|\widehat u_{ij}^{\mathrm F,\star}|^4=O_{\mathbb P}(1)$.

By the proof of Theorem~\ref{thm:fisher-adjustment},
$\Delta_{\mathrm F}:=\widehat\theta_{\mathrm F}-\theta_{\mathrm F}^\star=o_{\mathbb P}(1)$.
For $G_{pq,ij}=1$, the two residuals satisfy
\(\widehat u_{ij}^{\mathrm F} = \widehat u_{ij}^{\mathrm F,\star}-\widehat\delta_{ij}^{\mathrm F}\), \(\widehat\delta_{ij}^{\mathrm F} = \{X_{ij}-\widehat X(p,q)\}^\top\Delta_{\mathrm F}\).
For each cell, convexity gives
\(\|\widehat X(p,q)-\bar X\|^4 \le (I_pJ_q)^{-1}\sum_{i,j}G_{pq,ij}\|\dot X_{ij}\|^4\).
Consequently,
\begin{equation}\label{eq:within-cell-covariate-fourth-moment}
\frac1{IJ}\sum_{i,j}
\|X_{ij}-\widehat X(W_i^B,W_j^S)\|^4
\le
\frac{16}{IJ}\sum_{i,j}\|\dot X_{ij}\|^4
=O(1),
\end{equation}
where we used $\|a-b\|^4\le8\|a\|^4+8\|b\|^4$ and
\eqref{eq:covariate-fourth-moment}. Hence
\((IJ)^{-1}\sum_{i,j}|\widehat\delta_{ij}^{\mathrm F}|^4 \le \|\Delta_{\mathrm F}\|^4 (IJ)^{-1}\sum_{i,j} \|X_{ij}-\widehat X(W_i^B,W_j^S)\|^4 =o_{\mathbb P}(1)\).
All conditions of Lemma~\ref{lem:oracle-sandwich-equivalence} are therefore
satisfied, yielding
\begin{equation}\label{eq:fisher-sandwich-oracle-equivalence-proof}
\widehat\Omega_{\mathrm{2w}}^{\mathrm F}
=
\widehat\Omega_{\mathrm{2w}}(Y^{\mathrm F,\star})
+o_{\mathbb P}(\rho_{I,J})
\end{equation}
in operator norm. Combining
\eqref{eq:oracle-fisher-2w-conservative-proof} and
\eqref{eq:fisher-sandwich-oracle-equivalence-proof} proves the theorem.
\end{proof}

\subsection{Two-way cluster-robust covariance from Lin's regression}
\label{app:lin-sandwich-definition}

For later use, we define the two-way clustered covariance matrix from Lin's regression itself.
For an observation in exposure cell \((p,q)\), define the OLS residual from
\eqref{eq:reg-lin-adjustment} as
\(
\widehat u_{ij}^{\mathrm L}
=
Y_{ij}^{\mathrm{obs}}
-
\widehat\mu_{\mathrm L,pq}
-
\dot X_{ij}^\top\widehat\theta_{\mathrm L,pq}.
\)
Define the row regressor
\[
r_{ij}^{\mathrm L}
=
\left(
G_{11,ij},G_{10,ij},G_{01,ij},G_{00,ij},
G_{11,ij}\dot X_{ij}^\top,
G_{10,ij}\dot X_{ij}^\top,
G_{01,ij}\dot X_{ij}^\top,
G_{00,ij}\dot X_{ij}^\top
\right)^\top ,
\]
and let \(R_{\mathrm L}\) be the \(IJ\times(4+4d)\) matrix with row
\((r_{ij}^{\mathrm L})^\top\).
Write $R_{\mathrm L}=(Z,C_{\mathrm L})$, where $C_{\mathrm L}$ consists
of the last $4d$ columns, namely the exposure-specific covariate interactions.
Define
\(\Gamma_B^{\mathrm L} = \sum_{i=1}^I (R_{\mathrm L,i})^\top \widehat u_i^{\mathrm L}(\widehat u_i^{\mathrm L})^\top R_{\mathrm L,i}\), \(\Gamma_S^{\mathrm L} = \sum_{j=1}^J (R_{\mathrm L,\cdot j})^\top \widehat u_{\cdot j}^{\mathrm L}(\widehat u_{\cdot j}^{\mathrm L})^\top R_{\mathrm L,\cdot j}\), \(\Gamma_{BS}^{\mathrm L} = \sum_{i=1}^I\sum_{j=1}^J r_{ij}^{\mathrm L}(r_{ij}^{\mathrm L})^\top (\widehat u_{ij}^{\mathrm L})^2\).
The two-way clustered covariance matrix for all coefficients in Lin's regression is
\[
\widehat V_{\eta,\mathrm{2w}}^{\mathrm L}
=
(R_{\mathrm L}^\top R_{\mathrm L})^{-1}
\left(
\Gamma_B^{\mathrm L}
+
\Gamma_S^{\mathrm L}
-
\Gamma_{BS}^{\mathrm L}
\right)
(R_{\mathrm L}^\top R_{\mathrm L})^{-1}.
\]
Let
\(
H_{\mu,\mathrm L}
=
(\mathrm I_4,0_{4\times 4d}).
\)
The two-way clustered covariance matrix for the Lin-adjusted exposure-mean estimator
\(\widehat{\boldsymbol\mu}_{\mathrm L}\) is the dummy-coefficient block
\(
\widehat\Omega_{\mathrm{2w}}^{\mathrm L}
=
H_{\mu,\mathrm L}
\widehat V_{\eta,\mathrm{2w}}^{\mathrm L}
H_{\mu,\mathrm L}^\top .
\)

\subsection{Proof of Theorem~\ref{thm:lin-adjustment}}

\begin{proof}
We prove the result in three steps.

\paragraph*{Step 1: consistency of the cell-specific Lin slopes.}
By the Frisch--Waugh--Lovell representation of \eqref{eq:reg-lin-adjustment}, for each
\((p,q)\in\{0,1\}^2\),
\(\widehat\theta_{\mathrm L,pq} = \widehat Q_{\mathrm L,pq}^{-1}\widehat q_{\mathrm L,pq}\)
with probability approaching one, where
\begin{align}
\widehat Q_{\mathrm L,pq}
&=
\frac{1}{I_pJ_q}
\sum_{i=1}^I\sum_{j=1}^J
G_{pq,ij}
\{X_{ij}-\widehat X(p,q)\}
\{X_{ij}-\widehat X(p,q)\}^\top,
\label{eq:Q-lin-hat}
\\
\widehat q_{\mathrm L,pq}
&=
\frac{1}{I_pJ_q}
\sum_{i=1}^I\sum_{j=1}^J
G_{pq,ij}
\{X_{ij}-\widehat X(p,q)\}
\{Y_{ij}^{\mathrm{obs}}-\widehat Y(p,q)\}.
\label{eq:q-lin-hat}
\end{align}
We show that, for each fixed \((p,q)\),
\(\widehat Q_{\mathrm L,pq}=Q_{X}^{\circ}+o_{\mathbb P}(1)\), \(\widehat q_{\mathrm L,pq}=q_{\mathrm L,pq}^{\circ}+o_{\mathbb P}(1)\).

Apply Lemma~\ref{lem:cell-average-lln} to the fixed arrays
\(X_{ij}\), \(\dot X_{ij}\dot X_{ij}^\top\), \(Y_{ij}(p,q)\), and
\(\dot X_{ij}\dot Y_{ij}(p,q)\). Condition~\ref{cond:covariate-regularity} gives the required
second-moment bounds. Hence
\(\widehat X(p,q)-\bar X=O_{\mathbb P}(\rho_{I,J}^{1/2})\), \(\widehat Y(p,q)-\bar Y(p,q)=O_{\mathbb P}(\rho_{I,J}^{1/2})\),
and the cell averages of \(\dot X_{ij}\dot X_{ij}^\top\) and
\(\dot X_{ij}\dot Y_{ij}(p,q)\) converge to their corresponding finite-population averages.

Using
\(\widehat Q_{\mathrm L,pq} = (I_pJ_q)^{-1} \sum_{i,j} G_{pq,ij}\dot X_{ij}\dot X_{ij}^\top - \{\widehat X(p,q)-\bar X\} \{\widehat X(p,q)-\bar X\}^\top\),
we obtain
\(\widehat Q_{\mathrm L,pq}=Q_{X}^{\circ}+o_{\mathbb P}(1)\).
Similarly,
\(\widehat q_{\mathrm L,pq} = (I_pJ_q)^{-1} \sum_{i,j} G_{pq,ij}\dot X_{ij}\dot Y_{ij}(p,q) - \{\widehat X(p,q)-\bar X\} \{\widehat Y(p,q)-\bar Y(p,q)\}\),
and therefore
\(\widehat q_{\mathrm L,pq}=q_{\mathrm L,pq}^{\circ}+o_{\mathbb P}(1)\).
Since \(\lambda_{\min}(Q_{X}^{\circ})\ge c_Q>0\), we have
\(\widehat Q_{\mathrm L,pq}^{-1} = (Q_{X}^{\circ})^{-1}+o_{\mathbb P}(1)\).
Thus, for every \((p,q)\),
\(\widehat\theta_{\mathrm L,pq} = \widehat Q_{\mathrm L,pq}^{-1}\widehat q_{\mathrm L,pq} = (Q_{X}^{\circ})^{-1}q_{\mathrm L,pq}^{\circ}+o_{\mathbb P}(1) = \theta_{\mathrm L,pq}^\star+o_{\mathbb P}(1)\).
Because there are only four exposure cells,
\(\max_{p,q} \|\widehat\theta_{\mathrm L,pq}-\theta_{\mathrm L,pq}^\star\| = o_{\mathbb P}(1)\).

\paragraph*{Step 2: oracle residualized exposure means.}
Define the oracle residualized observed cell mean
\(\widehat Y^{\mathrm L}(p,q) = (I_pJ_q)^{-1} \sum_{i=1}^I\sum_{j=1}^J \{Y_{ij}^{\mathrm{obs}}-\dot X_{ij}^\top\theta_{\mathrm L,pq}^\star\} G_{pq,ij}\),
and collect these four means in
\(\widehat{\mathbf Y}^{\mathrm L} = \bigl( \widehat Y^{\mathrm L}(1,1), \widehat Y^{\mathrm L}(1,0), \widehat Y^{\mathrm L}(0,1), \widehat Y^{\mathrm L}(0,0) \bigr)^\top\).
Because the residualized potential outcomes \(Y_{ij}^{\mathrm L}(p,q)\) have finite-population
means \(\bar Y(p,q)\), Condition~\ref{cond:lin-adjustment} and
Theorem~\ref{thm:vector-clt} imply
\begin{equation}\label{eq:oracle-lin-vector-clt}
\Sigma_{\mathrm L}^{-1/2}
\left(
\widehat{\mathbf Y}^{\mathrm L}
-
\bar{\mathbf Y}
\right)
\Rightarrow
N(0,\mathrm I_4).
\end{equation}

\paragraph*{Step 3: replacing the oracle slopes by the estimated slopes.}
Combining \eqref{eq:mu-lin-cell-representation} with the definition of
\(\widehat Y^{\mathrm L}(p,q)\) gives
\(\widehat\mu_{\mathrm L,pq} - \widehat Y^{\mathrm L}(p,q) = \{\widehat X(p,q)-\bar X\}^\top ( \theta_{\mathrm L,pq}^\star-\widehat\theta_{\mathrm L,pq} )\).
Therefore,
\[
\left\|
\widehat{\boldsymbol\mu}_{\mathrm L}-\widehat{\mathbf Y}^{\mathrm L}
\right\|
\le
2
\max_{p,q}\|\widehat X(p,q)-\bar X\|
\cdot
\max_{p,q}\|\widehat\theta_{\mathrm L,pq}-\theta_{\mathrm L,pq}^\star\|
=
O_{\mathbb P}(\rho_{I,J}^{1/2})o_{\mathbb P}(1)
=
o_{\mathbb P}(\rho_{I,J}^{1/2}).
\]
Condition~\ref{cond:lin-adjustment} includes the nondegeneracy condition in
Condition~\ref{cond:vector-clt}, so
\(
\|\Sigma_{\mathrm L}^{-1/2}\|_{\mathrm{op}}
=
O(\rho_{I,J}^{-1/2}).
\)
Hence
\(\Sigma_{\mathrm L}^{-1/2} ( \widehat{\boldsymbol\mu}_{\mathrm L}-\widehat{\mathbf Y}^{\mathrm L} ) = o_{\mathbb P}(1)\).
The conclusion follows from \eqref{eq:oracle-lin-vector-clt} and Slutsky's theorem.
\end{proof}

\subsection{Proof of Theorem~\ref{thm:lin-2w-conservative}}

\begin{proof}
Define
\(Y_{ij}^{\mathrm L,\mathrm{obs},\star} = Y_{ij}^{\mathrm{obs}} - \sum_{p,q}G_{pq,ij}\dot X_{ij}^\top\theta_{\mathrm L,pq}^\star\),
its cell mean $\widehat Y^{\mathrm L,\star}(p,q)$, and the corresponding
saturated-regression residual
\(\widehat u_{ij}^{\mathrm L,\star} = Y_{ij}^{\mathrm L,\mathrm{obs},\star} - \widehat Y^{\mathrm L,\star}(p,q)\), \(G_{pq,ij}=1\).
Let $\widehat\Omega_{\mathrm{2w}}(Y^{\mathrm L,\star})$ be the two-way covariance from regressing $Y^{\mathrm L,\mathrm{obs},\star}$ on $Z$.
Conditions~\ref{cond:lin-adjustment} and~\ref{cond:lin-variance} imply
\begin{equation}\label{eq:oracle-lin-2w-conservative-proof}
\widehat\Omega_{\mathrm{2w}}(Y^{\mathrm L,\star})-\Sigma_{\mathrm L}
=
\frac{1}{I}\mathcal S_I\left(
e_1^BA^{\mathrm L,(1)}+e_0^BA^{\mathrm L,(0)}
\right)
+
\frac{1}{J}\mathcal S_J\left(
e_1^SB^{\mathrm L,(1)}+e_0^SB^{\mathrm L,(0)}
\right)
 +o_{\mathbb P}(\rho_{I,J}),
\end{equation}
by Proposition~\ref{prop:twoway-vector-conservative}.

We now apply Lemma~\ref{lem:oracle-sandwich-equivalence} with
$C=C_{\mathrm L}$, $R=R_{\mathrm L}$, and
$u^\star=\widehat u^{\mathrm L,\star}$. The matrix
$(IJ)^{-1}C_{\mathrm L}^\top C_{\mathrm L}$ is block diagonal, with cell
$(p,q)$ block
\((IJ)^{-1}\sum_{i,j}G_{pq,ij}\dot X_{ij}\dot X_{ij}^\top = e_p^Be_q^S Q_{X}^{\circ}+o_{\mathbb P}(1)\),
and
\(
(IJ)^{-1}Z^\top C_{\mathrm L}=O_{\mathbb P}(\rho_{I,J}^{1/2}).
\)
These relations follow from Lemma~\ref{lem:cell-average-lln}; overlap and
Condition~\ref{cond:covariate-regularity} give the required eigenvalue bounds.
The covariate fourth-moment bound in
\eqref{eq:covariate-fourth-moment} follows from the shared
Condition~\ref{cond:covariate-regularity}, and therefore
$(IJ)^{-1}\sum_{i,j}\|r_{ij}^{\mathrm L}\|^4=O_{\mathbb P}(1)$.
As in the Fisher case, Condition~\ref{cond:lin-variance} yields
\((IJ)^{-1}\sum_{i,j}|\widehat u_{ij}^{\mathrm L,\star}|^4=O_{\mathbb P}(1)\).

By the proof of Theorem~\ref{thm:lin-adjustment},
$\max_{p,q}\|\Delta_{\mathrm L,pq}\|=o_{\mathbb P}(1)$, where
$\Delta_{\mathrm L,pq}=\widehat\theta_{\mathrm L,pq}-
\theta_{\mathrm L,pq}^\star$. For $G_{pq,ij}=1$,
\(\widehat u_{ij}^{\mathrm L} = \widehat u_{ij}^{\mathrm L,\star}-\widehat\delta_{ij}^{\mathrm L}\), \(\widehat\delta_{ij}^{\mathrm L} = \{X_{ij}-\widehat X(p,q)\}^\top\Delta_{\mathrm L,pq}\).
Equation~\eqref{eq:within-cell-covariate-fourth-moment} and the preceding
slope consistency give
\((IJ)^{-1}\sum_{i,j}|\widehat\delta_{ij}^{\mathrm L}|^4 \le \max_{p,q}\|\Delta_{\mathrm L,pq}\|^4 (IJ)^{-1}\sum_{i,j} \|X_{ij}-\widehat X(W_i^B,W_j^S)\|^4 =o_{\mathbb P}(1)\).
Lemma~\ref{lem:oracle-sandwich-equivalence} now gives
\begin{equation}\label{eq:lin-sandwich-oracle-equivalence-proof}
\widehat\Omega_{\mathrm{2w}}^{\mathrm L}
=
\widehat\Omega_{\mathrm{2w}}(Y^{\mathrm L,\star})
+o_{\mathbb P}(\rho_{I,J})
\end{equation}
in operator norm. Combining this relation with
\eqref{eq:oracle-lin-2w-conservative-proof} proves the theorem.
\end{proof}

\subsection{Proof of Proposition~\ref{prop:anova-adjustment-efficiency}}
\label{app:anova-adjustment-efficiency}

\begin{proof}
By construction,
\(I^{-1}\sum_{i=1}^I X_i^B=0\), \(J^{-1}\sum_{j=1}^J X_j^S=0\),
and the doubly centered component satisfies
\(J^{-1}\sum_{j=1}^J X_{ij}^{BS}=0\), for every \(i\), \(I^{-1}\sum_{i=1}^I X_{ij}^{BS}=0\), for every \(j\).
Moreover, the three components $X_i^B$, $X_j^S$, and $X_{ij}^{BS}$
are mutually orthogonal over the full finite population.
Only the buyer and seller projection coefficients are needed for the leading
covariances. If the $X^{BS}$ block is singular, its least-squares coefficient
need not be unique, but its fitted contribution is unique and has zero row
and column averages. Thus it does not affect any covariance comparison below.

Let $\gamma_{pq}^B$ and
$\gamma_{pq}^S$ denote the population coefficients on
$X_i^B$ and $X_j^S$, respectively, in the ANOVA-decomposed Lin
regression for exposure cell $(p,q)$. Orthogonality of the three
covariate components gives
\(\gamma_{pq}^B = \{\mathcal S_I(X^B)\}^{-1} (I-1)^{-1} \sum_{i=1}^I X_i^B \{\bar Y_{i\cdot}(p,q)-\bar Y(p,q)\}\),
and
\(\gamma_{pq}^S = \{\mathcal S_J(X^S)\}^{-1} (J-1)^{-1} \sum_{j=1}^J X_j^S \{\bar Y_{\cdot j}(p,q)-\bar Y(p,q)\}\).

\paragraph*{Part (i): ANOVA-decomposed Fisher adjustment.}
Define the sign-and-fraction vectors
\(v_B= ((e_1^B)^{-1},(e_1^B)^{-1},-(e_0^B)^{-1},-(e_0^B)^{-1})^\top\), \(v_S= ((e_1^S)^{-1},-(e_0^S)^{-1},(e_1^S)^{-1},-(e_0^S)^{-1})^\top\).
The blockwise population normal equations for ANOVA-decomposed Fisher
regression give the common buyer and seller coefficients
$\gamma_{\mathrm F}^B=\sum_{p,q}e_p^Be_q^S\gamma_{pq}^B$ and
$\gamma_{\mathrm F}^S=\sum_{p,q}e_p^Be_q^S\gamma_{pq}^S$.
Since the other covariate components have zero averages on the relevant side,
\(A_i^{\mathrm{AF}}=A_i-v_B(X_i^B)^\top\gamma_{\mathrm F}^B\), \(B_j^{\mathrm{AF}}=B_j-v_S(X_j^S)^\top\gamma_{\mathrm F}^S\).
Put
\(h_B=(\gamma_{\mathrm F}^B)^\top\mathcal S_I(X^B)\gamma_{\mathrm F}^B\), \(h_S=(\gamma_{\mathrm F}^S)^\top\mathcal S_J(X^S)\gamma_{\mathrm F}^S\).
Both quantities are nonnegative. Under the conditions in part (i), for each
$(p,q)$,
\((\gamma_{pq}^B)^\top\mathcal S_I(X^B)\gamma_{\mathrm F}^B=h_B\), \((\gamma_{pq}^S)^\top\mathcal S_J(X^S)\gamma_{\mathrm F}^S=h_S\).
Consequently, the definitions of the projection slopes and score vectors imply
\((I-1)^{-1}\sum_i(A_i-\bar A)(X_i^B)^\top\gamma_{\mathrm F}^B=h_Bv_B\), \((J-1)^{-1}\sum_j(B_j-\bar B)(X_j^S)^\top\gamma_{\mathrm F}^S=h_Sv_S\).
Expanding the covariance of the adjusted scores therefore gives
\(\mathcal S_I(A)-\mathcal S_I(A^{\mathrm{AF}})=h_Bv_Bv_B^\top\), \(\mathcal S_J(B)-\mathcal S_J(B^{\mathrm{AF}})=h_Sv_Sv_S^\top\).
It follows that
\begin{equation}
\Sigma-\Sigma_{\mathrm{AF}}
=
\frac{e_1^Be_0^B}{I}h_Bv_Bv_B^\top
+
\frac{e_1^Se_0^S}{J}h_Sv_Sv_S^\top
\succeq0.
\label{eq:anova-fisher-no-harm-gap}
\end{equation}
This proves part (i).

\paragraph*{Part (ii): ANOVA-decomposed Lin adjustment.}
Define the $d\times4$ matrices
\(K_B^\star = \{\mathcal S_I(X^B)\}^{-1} (I-1)^{-1} \sum_{i=1}^I X_i^B(A_i-\bar A)^\top\),
and
\(K_S^\star = \{\mathcal S_J(X^S)\}^{-1} (J-1)^{-1} \sum_{j=1}^J X_j^S(B_j-\bar B)^\top\).
By the definitions of $A_i$ and $B_j$, the preceding cell-specific
projection coefficients imply
\(A_i^{\mathrm{AL}} = A_i-(K_B^\star)^\top X_i^B\), \(B_j^{\mathrm{AL}} = B_j-(K_S^\star)^\top X_j^S\).
The component $X_{ij}^{BS}$ does not enter either expression because
its row and column averages are zero. Therefore,
\[
\Sigma_{\mathrm{AL}}
=
\frac{e_1^Be_0^B}{I}
\mathcal S_I\left\{
A-(K_B^\star)^\top X^B
\right\}
+
\frac{e_1^Se_0^S}{J}
\mathcal S_J\left\{
B-(K_S^\star)^\top X^S
\right\}.
\]

For arbitrary $d\times4$ matrices $K_B$ and $K_S$, define
\[
\Sigma(K_B,K_S)
=
\frac{e_1^Be_0^B}{I}
\mathcal S_I\left\{
A-K_B^\top X^B
\right\}
+
\frac{e_1^Se_0^S}{J}
\mathcal S_J\left\{
B-K_S^\top X^S
\right\}.
\]
The normal equations give
\((I-1)^{-1} \sum_{i=1}^I X_i^B [ A_i-\bar A-(K_B^\star)^\top X_i^B ]^\top =0\),
with the analogous identity on the seller side. Hence the
finite-population Pythagorean identity yields
\begin{align}
\Sigma(K_B,K_S)-\Sigma_{\mathrm{AL}}
={}&
\frac{e_1^Be_0^B}{I}
(K_B-K_B^\star)^\top
\mathcal S_I(X^B)
(K_B-K_B^\star)
\nonumber\\
&+
\frac{e_1^Se_0^S}{J}
(K_S-K_S^\star)^\top
\mathcal S_J(X^S)
(K_S-K_S^\star)
\succeq0.
\label{eq:anova-lin-pythagorean}
\end{align}

The unadjusted estimator corresponds to $K_B=K_S=0$, and therefore
\(
\Sigma_{\mathrm{AL}}\preceq\Sigma.
\)

Next consider ordinary Fisher, ordinary Lin, and ANOVA-decomposed Fisher regressions. Since
\(J^{-1}\sum_{j=1}^J\dot X_{ij}=X_i^B\), \(I^{-1}\sum_{i=1}^I\dot X_{ij}=X_j^S\),
their residualized buyer- and seller-side first-order scores can be
written, for suitable matrices $K_B^m$ and $K_S^m$, as
\(A_i^m=A_i-(K_B^m)^\top X_i^B\), \(B_j^m=B_j-(K_S^m)^\top X_j^S\), \(m\in\{\mathrm F,\mathrm L,\mathrm{AF}\}\).
For ANOVA-decomposed Fisher regression specifically,
$K_B^{\mathrm{AF}}=\gamma_{\mathrm F}^Bv_B^\top$ and
$K_S^{\mathrm{AF}}=\gamma_{\mathrm F}^Sv_S^\top$.
These representations do not require the conditions in part (i).
Consequently,
\(\Sigma_m=\Sigma(K_B^m,K_S^m)\), \(m\in\{\mathrm F,\mathrm L,\mathrm{AF}\}\).
Applying \eqref{eq:anova-lin-pythagorean} with
$(K_B,K_S)=(K_B^m,K_S^m)$ gives
\(\Sigma_{\mathrm{AL}}\preceq\Sigma_{\mathrm F}\), \(\Sigma_{\mathrm{AL}}\preceq\Sigma_{\mathrm L}\), \(\Sigma_{\mathrm{AL}}\preceq\Sigma_{\mathrm{AF}}\).
Together with the unadjusted comparison, this proves part (ii).
\end{proof}
\subsection{Side-specific covariates}
\label{app:side-specific-covariates}
\begin{prop}[Side-specific covariates]
\label{prop:side-specific-covariates}

Suppose $X_{ij}=X_i^B$ for all $(i,j)$. Let
\(\bar X^B=I^{-1}\sum_{i=1}^I X_i^B\), \(Q_B = I^{-1}\sum_{i=1}^I (X_i^B-\bar X^B)(X_i^B-\bar X^B)^\top\).

\begin{enumerate}
\item For Fisher's regression,
\[
\beta_{11}+\beta_{10}=\beta_{01}+\beta_{00}=0
\quad\Longrightarrow\quad
\beta^\top\widehat{\boldsymbol\mu}_{\mathrm F}
=
\beta^\top\widehat{\mathbf Y}.
\]

\item If $Q_B$ is nonsingular, then for Lin's regression,
\(B_j^{\mathrm L}=B_j\), \(j=1,\ldots,J\), \(\Sigma_{\mathrm L}\preceq\Sigma\).
\end{enumerate}

The analogous results hold for $X_{ij}=X_j^S$ with the roles of buyers and
sellers reversed.
\end{prop}

\begin{proof}
For $p\in\{0,1\}$, define
\(\widehat X^B(p) = I_p^{-1}\sum_{i:W_i^B=p}X_i^B\).
Since $X_{ij}=X_i^B$,
\(\widehat X(p,q)=\widehat X^B(p)\).
Hence \eqref{eq:mu-fisher-cell-representation} gives, for any $\beta$,
\[
\beta^\top\widehat{\boldsymbol\mu}_{\mathrm F}
=
\beta^\top\widehat{\mathbf Y}
-
\sum_{p=0}^1
\left(\sum_{q=0}^1\beta_{pq}\right)
\{\widehat X^B(p)-\bar X^B\}^\top
\widehat\theta_{\mathrm F}.
\]
This proves part (i).

For part (ii), let
\(\dot X_i^B=X_i^B-\bar X^B\), \(M_B = I^{-1}\sum_{i=1}^I \dot X_i^B(A_i-\bar A)^\top\).
If $Q_B$ is nonsingular, the cell-specific population slopes imply
\(A_i^{\mathrm L} = A_i-M_B^\top Q_B^{-1}\dot X_i^B\).
Moreover,
\(\bar Y_{\cdot j}^{\mathrm L}(p,q) = \bar Y_{\cdot j}(p,q) - (I^{-1}\sum_i\dot X_i^B)^\top \theta_{\mathrm L,pq}^\star = \bar Y_{\cdot j}(p,q)\),
so $B_j^{\mathrm L}=B_j$.

The finite-population least-squares identity gives
\(\mathcal S_I(A^{\mathrm L}) = \mathcal S_I(A) - I (I-1)^{-1} M_B^\top Q_B^{-1}M_B\).
Therefore,
\(\Sigma-\Sigma_{\mathrm L} = e_1^Be_0^B (I-1)^{-1} M_B^\top Q_B^{-1}M_B \succeq0\).
The seller-side result follows symmetrically.
\end{proof}

\subsection{Optimized scalar inference after regression adjustment}
\label{app:adjusted-optimized-scalar}

This subsection formalizes the extension described in
Section~\ref{subsec:adjusted-optimized-scalar}. Fix a contrast
\(\beta=(\beta_{11},\beta_{10},\beta_{01},\beta_{00})^\top\), and let
\(m\in\{\mathrm F,\mathrm L\}\) index Fisher or Lin adjustment. Apply the projection
definitions in Section~\ref{subsec:scalar-inference} to the corresponding
oracle residualized potential outcomes \(Y_{ij}^m(p,q)\). In particular, define
\(a_i^m(p) = (e_p^B)^{-1}\sum_{q=0}^1 \beta_{pq}\bar Y_{i\cdot}^m(p,q)\), \(b_j^m(q) = (e_q^S)^{-1}\sum_{p=0}^1 \beta_{pq}\bar Y_{\cdot j}^m(p,q)\),
and let \(S_{a,p}^m\) and \(S_{b,q}^m\) be their finite-population standard
deviations. Let \(S_{a,\Delta}^m\) and \(S_{b,\Delta}^m\) denote the
finite-population standard deviations of
\(a_i^m(1)-a_i^m(0)\) and \(b_j^m(1)-b_j^m(0)\), respectively. Write
\(
\sigma_{m,\beta}^2
=
\beta^\top\Sigma_m\beta
\)
for the leading variance of
\(\widehat\tau_{m,\beta}:=\beta^\top\widehat{\boldsymbol\mu}_m\). The adjusted
optimized population target is
\begin{equation}\label{eq:adjusted-Vopt-population}
V_{\mathrm{opt},\beta}^{m,\circ}
=
\frac{e_1^Be_0^B}{I}
\bigl(S_{a,1}^m+S_{a,0}^m\bigr)^2
+
\frac{e_1^Se_0^S}{J}
\bigl(S_{b,1}^m+S_{b,0}^m\bigr)^2.
\end{equation}

The feasible residualized outcomes are
\(\widehat Y_{ij}^{\mathrm F,\mathrm{res}} = Y_{ij}^{\mathrm{obs}} -\dot X_{ij}^\top\widehat\theta_{\mathrm F}\), \(\widehat Y_{ij}^{\mathrm L,\mathrm{res}} = Y_{ij}^{\mathrm{obs}} - \sum_{p=0}^1\sum_{q=0}^1 G_{pq,ij}\dot X_{ij}^\top\widehat\theta_{\mathrm L,pq}\).
The first-order conditions in
\eqref{eq:mu-fisher-cell-representation} and
\eqref{eq:mu-lin-cell-representation} imply, for either \(m\),
\(\widehat\mu_{m,pq} = (I_pJ_q)^{-1} \sum_{i=1}^I\sum_{j=1}^J G_{pq,ij}\widehat Y_{ij}^{m,\mathrm{res}}\).
Thus the adjusted coefficient vector is exactly the vector of cell means of the
corresponding feasible residualized outcome.

Use \(\widehat Y_{ij}^{m,\mathrm{res}}\) in place of \(Y_{ij}^{\mathrm{obs}}\) in
the sample projection construction preceding \eqref{eq:Vopt-sample}. Denote the
resulting state-specific sample standard deviations by
\(s_{a,p}^m\) and \(s_{b,q}^m\), and define
\begin{equation}\label{eq:adjusted-Vopt-sample}
\widehat V_{\mathrm{opt},\beta}^m
=
\frac{e_1^Be_0^B}{I}
\bigl(s_{a,1}^m+s_{a,0}^m\bigr)^2
+
\frac{e_1^Se_0^S}{J}
\bigl(s_{b,1}^m+s_{b,0}^m\bigr)^2.
\end{equation}
Let \(\Omega_{\mathrm{2w}}^{m,\circ}\) denote the deterministic first-order target
of the two-way clustered covariance matrix evaluated on the oracle residualized
potential outcomes \(Y^m\). The equivalence results in
\eqref{eq:fisher-sandwich-oracle-equivalence-proof} and
\eqref{eq:lin-sandwich-oracle-equivalence-proof} show that this is also the
first-order target of the two-way clustered covariance from the corresponding
feasible adjusted regression.

\begin{prop}[Optimized scalar inference after regression adjustment]
\label{prop:adjusted-optimized-scalar}
Suppose Assumptions~\ref{ass:smrd-randomization} and
\ref{ass:local-interference} and
Condition~\ref{cond:covariate-regularity} hold.
For \(m=\mathrm F\), impose
Conditions~\ref{cond:fisher-adjustment} and~\ref{cond:fisher-variance};
for \(m=\mathrm L\), impose
Conditions~\ref{cond:lin-adjustment} and~\ref{cond:lin-variance}. Under the
corresponding conditions,
\(\widehat V_{\mathrm{opt},\beta}^m = V_{\mathrm{opt},\beta}^{m,\circ} + o_{\mathbb P}(I^{-1}+J^{-1})\),
and
\begin{equation}\label{eq:adjusted-Vopt-ordering}
\sigma_{m,\beta}^2
\le
V_{\mathrm{opt},\beta}^{m,\circ}
\le
\beta^\top\Omega_{\mathrm{2w}}^{m,\circ}\beta.
\end{equation}
Thus, the interval
\(\mathcal I_{m,\beta}(1-\alpha) = [ \widehat\tau_{m,\beta} - z_{1-\alpha/2}\sqrt{\widehat V_{\mathrm{opt},\beta}^m}, \; \widehat\tau_{m,\beta} + z_{1-\alpha/2}\sqrt{\widehat V_{\mathrm{opt},\beta}^m} ]\)
satisfies
\(\liminf_{I,J\to\infty} \p\{ \tau_\beta\in\mathcal I_{m,\beta}(1-\alpha) \} \ge 1-\alpha\).
\end{prop}

\begin{proof}
The residualized potential outcomes have the same exposure-cell means as the
original potential outcomes because
\((IJ)^{-1}\sum_{i,j}\dot X_{ij}=0\). Applying the scalar variance decomposition
to \(Y^m\) gives
\(\sigma_{m,\beta}^2 = e_1^Be_0^B I^{-1} \bigl(S_{a,\Delta}^m\bigr)^2 + e_1^Se_0^S J^{-1} \bigl(S_{b,\Delta}^m\bigr)^2\).
The triangle inequality for finite-population standard deviations therefore
gives the first inequality in \eqref{eq:adjusted-Vopt-ordering}. Likewise, the
deterministic variance-family argument in
Section~\ref{app:optimized-variance-family}, applied to \(Y^m\), shows that
\(V_{\mathrm{opt},\beta}^{m,\circ}\) is the infimum of that family, whereas its
\((1,1)\) member is
\(\beta^\top\Omega_{\mathrm{2w}}^{m,\circ}\beta\). This proves the second
inequality.

It remains to justify the use of the estimated slopes. Let
\(\widehat V_{\mathrm{opt},\beta}^{m,\star}\) denote the sample construction
based on the oracle residualized observed outcome
\(Y_{ij}^{m,\mathrm{obs},\star}\). The corresponding adjustment and variance
conditions allow Proposition~\ref{prop:optimized-scalar-variance} to be applied
to \(Y^m\), yielding
\(\widehat V_{\mathrm{opt},\beta}^{m,\star} = V_{\mathrm{opt},\beta}^{m,\circ} + o_{\mathbb P}(\rho_{I,J})\).
The slope-consistency results in the proofs of
Theorems~\ref{thm:fisher-adjustment} and~\ref{thm:lin-adjustment} give
\(\widehat\theta_{\mathrm F}-\theta_{\mathrm F}^\star=o_{\mathbb P}(1)\), \(\max_{p,q} \|\widehat\theta_{\mathrm L,pq}-\theta_{\mathrm L,pq}^\star\| =o_{\mathbb P}(1)\).
Together with the covariate moment condition, these relations imply, for
\(
d_{ij}^m
:=
\widehat Y_{ij}^{m,\mathrm{res}}
-
Y_{ij}^{m,\mathrm{obs},\star},
\)
that
\((IJ)^{-1}\sum_{i=1}^I\sum_{j=1}^J(d_{ij}^m)^2=o_{\mathbb P}(1)\).
For every $p\in\{0,1\}$, Cauchy--Schwarz and Jensen's inequality give
\begin{align*}
\frac{1}{I_p}\sum_{i:W_i^B=p}
\left\{
\widehat a_i^m(p)-\widehat a_i^{m,\star}(p)
\right\}^2&\le
C_\beta
\sum_{q=0}^1
\frac{1}{I_pJ_q}
\sum_{i:W_i^B=p}\sum_{j:W_j^S=q}(d_{ij}^m)^2
\\ &\le
\frac{C}{IJ}\sum_{i,j}(d_{ij}^m)^2
=o_{\mathbb P}(1),
\end{align*}
where the second inequality uses the overlap condition. The seller-side analogue holds
with $I$ and $J$ interchanged.

Let
$\Delta a_i^m(p)=\widehat a_i^m(p)-\widehat a_i^{m,\star}(p)$ and let
$\overline{\Delta a}_p^m$ be its sample mean among buyers assigned to $p$.
The reverse triangle inequality for centered Euclidean norms yields
\begin{align*}
\left|s_{a,p}^m-s_{a,p}^{m,\star}\right|
&\le
\left[
\frac{1}{I_p-1}
\sum_{i:W_i^B=p}
\{\Delta a_i^m(p)-\overline{\Delta a}_p^m\}^2
\right]^{1/2}\\
&\le
\left[
\frac{1}{I_p-1}
\sum_{i:W_i^B=p}\{\Delta a_i^m(p)\}^2
\right]^{1/2}
=o_{\mathbb P}(1).
\end{align*}
The same argument gives
\(s_{b,q}^m-s_{b,q}^{m,\star}=o_{\mathbb P}(1)\), \(q\in\{0,1\}\).
The oracle sample standard deviations are $O_{\mathbb P}(1)$ by
Lemma~\ref{lem:estimated-projection-variances}, and the feasible ones are
therefore also $O_{\mathbb P}(1)$. Substitution into
\eqref{eq:adjusted-Vopt-sample} gives
\(\widehat V_{\mathrm{opt},\beta}^m - \widehat V_{\mathrm{opt},\beta}^{m,\star} = o_{\mathbb P}(\rho_{I,J})\),
which proves the stated consistency.

For a fixed nonzero \(\beta\), the nondegeneracy condition gives
\(\sigma_{m,\beta}^2\ge c_{m,\beta}\rho_{I,J}\). Thus, for every
\(\varepsilon>0\),
\(\p\{ \widehat V_{\mathrm{opt},\beta}^m \ge (1-\varepsilon)\sigma_{m,\beta}^2 \} \longrightarrow 1\).
The adjusted vector CLT implies
\((\widehat\tau_{m,\beta}-\tau_\beta)/\sigma_{m,\beta} \Rightarrow N(0,1)\).
The coverage conclusion follows from this CLT and the preceding variance bound.
The case \(\beta=0\) is trivial.
\end{proof}

\end{document}